\documentclass[11pt,a4paper,reqno]{amsart}%
\usepackage{amsthm,amsmath,amsfonts,amssymb,amsxtra,appendix,bookmark,euscript,delarray,dsfont,bm,mathrsfs}
\usepackage[latin1]{inputenc}
\usepackage{color}

\theoremstyle{plain}
\newtheorem{theorem}{Theorem}
\newtheorem{lemma}[theorem]{Lemma}

\theoremstyle{remark}
\newtheorem{remark}[theorem]{Remark}

\newcommand\R{{\ensuremath {\mathbb R} }}
\newcommand\T{{\ensuremath {\mathbb T} }}

\newcommand\Z{{\ensuremath {\mathbb Z} }}

\newcommand\1{{\ensuremath {\mathds 1} }}
\newcommand\dGamma{{\rm d}\Gamma}

\newcommand\nn{\nonumber}

\renewcommand\phi{\varphi}
\newcommand{\bH}{\mathbb{H}}

\newcommand{\cC}{\mathcal{C}}

\newcommand{\cM}{\mathcal{M}}

\newcommand{\cE}{\mathcal{E}}

\newcommand{\cF}{\mathcal{F}}
\newcommand{\cN}{\mathcal{N}}

\newcommand{\eps}{\epsilon}

\newcommand{\F}{\mathcal{F}}

\def\d{\,{\rm d}}

\renewcommand{\epsilon}{\varepsilon}

\DeclareMathOperator{\tr}{{\rm Tr}}
\DeclareMathOperator{\Tr}{{\rm Tr}}

\renewcommand{\ge}{\geqslant}
\renewcommand{\le}{\leqslant}
\renewcommand{\geq}{\geqslant}
\renewcommand{\leq}{\leqslant}
\renewcommand{\hat}{\widehat}
\renewcommand{\tilde}{\widetilde}

\allowdisplaybreaks

\title[Derivation of 3D energy-critical NLS and excitations]{Derivation of 3D energy-critical nonlinear Schr\"odinger equation and Bogoliubov excitations for Bose gases}

\author[P.T. Nam]{Phan Th\`anh Nam}
\address{Department of Mathematics, LMU Munich, Theresienstrasse 39, 80333 Munich, Germany} 
\email{nam@math.lmu.de}

\author[R. Salzmann]{Robert Salzmann}
\address{Department of Mathematics, LMU Munich, Theresienstrasse 39, 80333 Munich, Germany} 
\email{r.salzmann@campus.lmu.de, rals.salzmann@web.de}

\begin{document}

\begin{abstract} We derive the 3D quintic NLS as the mean field limit of a Bose gas with three-body interactions. The quintic NLS is energy-critical, leading to  several new difficulties in comparison with the cubic NLS which emerges from Bose gases with pair-interactions. Our method is based on Bogoliubov's approximation, which also provides the information on the fluctuations around the condensate in terms of a norm approximation for the $N$-body wave function. 
\end{abstract}

\maketitle

\setcounter{tocdepth}{2}
\tableofcontents

\section{Introduction}

The 3D energy-critical nonlinear Schr\"odinger equation (NLS) reads
\begin{equation}
\label{eq:Hartree-NLS}
\left\{
\begin{aligned}
i\partial_t \varphi (t,x) &= - \Delta \varphi(t,x) + b_0  |\varphi (t,x)|^4\varphi(t,x), \quad x\in \R^3, \quad t>0 \\
\varphi(0,x) &= \varphi_0(x). 
\end{aligned}
\right.
\end{equation}
The well-posedness of \eqref{eq:Hartree-NLS} in the defocusing case $b_0>0$ was first proved by Bourgain \cite{Bourgain-99} and Grillakis \cite{Grillakis-00} for radial data and then by Colliander, Keel, Staffilani, Takaoka, and Tao \cite{ColKeelStafTao-08} for general data. In this paper, we will derive  \eqref{eq:Hartree-NLS} as a macroscopic description for the microscopic many-body Schr\"odinger equation of bosons in a mean-field limit.

From first principles of quantum mechanics, the dynamics of a Bose gas  in 3D with $N$ particles  is described by the $N$-body Schr\"odinger equation
\begin{align} \label{eq:MBSch}
\left\{
\begin{aligned}
i\partial_t \Psi_N(t) &= H_N \Psi_N(t),\\
\Psi_N(t=0) &= \Psi_{N,0}.
\end{aligned}
\right.
\end{align}
Here $\Psi_N(t)$ is a wave function in the symmetric space $L^2_s((\R^3)^N)$ and $H_N$ is the Hamiltonian of the system. In this paper, we consider the case  of non-relativistic bosons interacting via a three-body interaction potential,
\begin{align} \label{eq:HN}
H_N=\sum_{j=1}^N -\Delta_{x_i} + \frac{1}{N^2} \sum_{1\le i<j<k \le N} N^{6\beta} V( N^{\beta}(x_i-x_j), N^{\beta}(x_i-x_k)).
\end{align}
Here $V:\R^3\times \R^3\to \R$ has the following symmetry conditions
\begin{align}
V(x,y) = V(y,x), \quad\quad V(x-y,x-z)=V(y-x,y-z)=V(z-y,z-x).
\end{align}
These symmetry conditions ensure that the total interaction of 3 particles only depends on the relative distances between them. 

Note that for any fixed parameter $\beta>0$, in the limit $N\to \infty$ the re-scaled potential 
\begin{align} \label{eq:VN}
V_N(x-y,x-z) =  N^{6\beta} V( N^{\beta}(x-y), N^{\beta}(x-z)) 
\end{align}
converges weakly to the delta interaction
\begin{align} \label{eq:delta-int}
b_0 \delta_{x=y=z}, \quad b_0= \frac{1}{2}\int_{\R^3\times \R^3} V(x,y) \d x \d y.
\end{align}
The bigger $\beta$ is, the more singular the potential is. Nevertheless, we may think of $V_N$ as a quantity of order 1.  The coupling constant $N^{-2}$ in front of the interaction terms in \eqref{eq:HN} places us in the mean-field regime, when the kinetic energy and the interaction energy are comparable in the large $N$ limit (they are both of order $N$, given that the system occupies a volume of order 1).

We are interested in the macroscopic behavior of the system when $N\to \infty$. To the leading order we expect the system to exhibit the Bose-Einstein condensation. This is the phenomenon when most of particles occupy a common single quantum state, namely in terms of the wave function 
\begin{align} \label{eq:BEC-formal}
\Psi_N(t) \approx \varphi(t)^{\otimes N}
\end{align}
in an appropriate sense, for a  function $\varphi(t)$ in $L^2(\R^3)$. A formal computation using the limiting interaction potential \eqref{eq:delta-int} suggests that $\varphi(t)$ solves the quintic NLS \eqref{eq:Hartree-NLS}. Making this computation rigorous, however, is a nontrivial problem.

In the present paper, we will justify the approximation \eqref{eq:BEC-formal} (with $\varphi(t)$ solving the quintic NLS \eqref{eq:Hartree-NLS}) for all $0< \beta<1/6$, leading to an extension of the recent important result of X. Chen and  Holmer \cite{CheHol-18} (see also T. Chen and Pavlovi\'c \cite{ChePav-11} for related results in lower dimensions). Moreover, we will go beyond the leading order and obtain information on the fluctuations around the condensate, in terms of a norm approximation for the wave function. In particular, we will also extend the norm approximation obtained by X. Chen \cite{Chen-12} in the mean-field case $\beta=0$.

When the particles interact only via pair interactions, the condensate should be effectively described by the cubic NLS (instead of the quintic NLS \eqref{eq:Hartree-NLS}). The well-posedness of the defocusing cubic NLS has been proved by Bourgain \cite{Bourgain-98} and Dodson \cite{Dodson-12}.   The rigorous derivation of the cubic NLS from many-body Schr\"odinger equation is the subject of a vast literature; see  \cite{Spohn-80,BarGolMau-00,AdaGolTet-07,ErdSchYau-07,ErdSchYau-10,AmmNie-08,FroKnoSch-09,KlaMac-08,RodSch-09,KnoPic-10,Pickl-15,BenOliSch-15,CheHaiPavSei-15,BreSch-17}. In particular, we refer to the seminal work of Erd\"os, Schlein and Yau \cite{ErdSchYau-10} on the critical case $\beta=1$, where the cubic NLS is replaced by the Gross-Pitaevskii equation with the true scattering length of the pair interaction (see also \cite{KlaMac-08,BenOliSch-15,Pickl-15,CheHaiPavSei-15,BreSch-17} for later developments). The norm approximation with pair interactions has also attracted many studies \cite{Hepp-74,GinVel-79,GriMacMar-10,GriMac-13,LewNamSch-15,Pizzo-15,NamNap-15,BocCenSch-15,Kuz-17,NamNap-17,BreNamNapSch-17}; in particular, we refer to \cite{BreNamNapSch-17} for  the last development which covers all $0<\beta<1$ (the case $\beta=1$ remains open). 

We will benefit from the methods developed to handle the pair-interaction case, in particular the justification of Bogoliubov's argument \cite{Bogoliubov-47} in \cite{LewNamSch-15,NamNap-15,NamNap-17,BreNamNapSch-17}. However, it turns out that the analysis in the case of three-body interactions is significantly more complicated  and several new ideas are needed. Our main results are presented in the next section.
  
\bigskip

\noindent\textbf{Acknowledgements.} We thank Jean-Claude Cuenin, Thomas Chen, Xuwen Chen, Justin Holmer and the referee for helpful remarks. 

\section{Main results}

\subsection{Convergence of reduced density matrices} The proper meaning of the Bose-Einstein condensation \eqref{eq:BEC-formal},
$$
\Psi_N(t) \approx \varphi(t)^{\otimes N},
$$
should be given in terms of the reduced density matrices. Recall that the one-body density matrix $\gamma_{\Psi_N}^{(1)}$ of a $N$-body wave function $\Psi_N$ is a non-negative trace class operator on $L^2(\R^3)$ with kernel \begin{align} \label{eq:def-1pdm-HN}
\gamma_{\Psi_N}(x;y) = N\int \Psi_N (x,x_2, \dots , x_N) \overline{\Psi_N(y,x_2,\dots,x_N)}\, \mathrm{d} x_{2}\cdots \mathrm{d} x_{N}.
\end{align}
(We use the convention that inner products in Hilbert spaces are linear in the second argument and anti-linear in the first.)

The precise meaning of \eqref{eq:BEC-formal} reads, in the limit $N\to \infty$, 
\begin{align} \label{eq:BEC}
\frac{\langle \varphi(t), \gamma_{\Psi_N}^{(1)} \varphi(t)\rangle}{N} \to  1,
\end{align}
namely the expectation of the number of particles in mode $\varphi(t)$ is mostly equal to $N$. Equivalently, we can rewrite \eqref{eq:BEC} as
\begin{align} \label{eq:BEC-2}
\frac{\Tr (Q(t)  \gamma_{\Psi_N}^{(1)} Q(t) ) }{N} \to  0,\quad Q(t)=1-|\varphi(t)\rangle \langle \varphi(t)|.
\end{align}
Moreover, since $|\varphi(t) \rangle \langle \varphi(t)|$ is a rank-one projection, \eqref{eq:BEC}-\eqref{eq:BEC-2} is equivalent to the trace class convergence
\begin{align}\label{eq:BEC-tr}
\Tr \left|  N^{-1}\gamma_{\Psi_N}^{(1)} - |\varphi(t) \rangle \langle \varphi(t)| \right| \to 0.
\end{align}

As explained in the introduction, it is natural to expect that $\varphi(t)$ solves the quintic NLS \eqref{eq:Hartree-NLS}. Our first result is a rigorous justification of this fact.

\begin{theorem}[{Convergence of reduced density matrices}]	\label{thm:main1} Assume that $0\le V \in \cC_c(\R^6)$. Let $\varphi(t)$ be the solution to the quintic NLS \eqref{eq:Hartree-NLS} with initial function $\varphi(0)\in H^4(\R^3)$, $\|\varphi(0)\|_{L^2}=1$. Let $\Psi_N(t)$ be the solution to the Schr\"odinger equation \eqref{eq:MBSch} with the initial state $\Psi_{N,0}$ in $L^2_s((\R^3)^N)$ satisfying
\begin{align} \label{eq:leading-order}
\Tr \Big[ (1-\Delta) Q(0) \gamma_{\Psi_{N,0}}^{(1)} Q(0)  \Big] \le C.
\end{align}
Assume that $0<\beta<1/6$. Then for all $t>0$ and for all $\alpha< \min\{\beta/2,(1-6\beta)/4\}$ fixed, we have
\begin{equation}
\tr\left|N^{-1}\gamma^{(1)}_{\Psi_N(t)} - |\varphi(t)\rangle\langle \varphi(t)|\right| \le C_t N^{-\alpha}. 
\end{equation}
Here $C_t$ is continuous in $t$ and independent of $N$. 
\end{theorem}

This kind of results was obtained recently by X. Chen and Holmer \cite{CheHol-18} for bosons on the torus $\T^3$ and for $0<\beta<1/9$. In fact, our result can be extended to the torus case as well (see Appendix for further explanation). Unlike the BBGKY approach in \cite{CheHol-18}, our method gives explicit error estimate and can be adapted easily when an external potential or a magnetic field appears. On the other hand, here we do not cover the case $\varphi(0)\in H^1(\R^3)$ as in \cite{CheHol-18} which is of certain mathematical interest. We rather think of the physical situation when the initial state $\Psi_N(0)$ is the ground state of a trapped system, and in this case the condensate $\varphi(0)$ is expected to be sufficiently regular. 

%

The assumption $0<\beta<1/6$ in Theorem \ref{thm:main1} is a technical condition which allows us to control the interaction potential by the kinetic operator. Note that the total interaction potential felt by the $i$-th particle is
$$
\frac{1}{N^2} \sum_{j,k: j\ne i, k\ne k}N^{6\beta}V(N^\beta(x_i-x_j), N^\beta(x_i-x_k)).
$$
Thus the total interaction potential felt by a single particle may be as large as $N^{6\beta}$ in the worst case (when all particles collapse to a singular point). In our method, we need to control the potential energy {\em per particle} by the total kinetic energy of {\em all particles} which is normally of order $N$ (since the system occupies a volume of order $1$). This requires $N^{6\beta}\ll N$, namely $\beta<1/6$. This condition is reminiscent of the well-known threshold $\beta<1/3$ in the pair-interaction case \cite{ErdSchYau-07,Pickl-15,GriMac-13,NamNap-15} where the total interaction potential  felt by the $i$-th particle is
$$
\frac{1}{N} \sum_{j: j\ne i}N^{3\beta}\widetilde V(N^\beta(x_i-x_j))
$$
which is as large as $N^{3\beta}$ in the worst case. We will come back to the explanation for the smallness condition on $\beta$ with more details later. 

It is natural to expect that the result in Theorem \ref{thm:main1} holds true for larger $\beta$'s up to a  critical value  where some subtle correction emerges to the leading order due to few-particle scattering processes. We expect that the critical value is $\beta=1/2$, for which the coupling term $N^{6\beta-2}$ in front of the three-body interaction potential  scales the same as $N^{2\beta}$ of the Laplacian (it is different from the critical value $\beta=1$  in the two-body interaction case \cite{ErdSchYau-10} where $N^{3\beta-1}$ is compared with $N^{2\beta}$). More precisely, in this Gross-Pitaevskii like regime, we expect that the correct ansatz for the many-body wave function is 
$$
\Psi_N(t,x_1,...,x_N) \approx \prod_{j=1}^N u(t,x_j) \prod_{1\le p<k<\ell \le N} f(N^{1/2}(x_p-x_k), N^{1/2}(x_p-x_\ell))
$$
where $f$ solves the zero-scattering equation in $\R^6$:
$$
-\Delta f + \frac{1}{3}V f =0, \quad \lim_{|X|\to \infty} f(X)=1. 
$$
Consequently, the condensate should be described by the quintic NLS equation \eqref{eq:Hartree-NLS} with the coupling constant $b_0=\frac{1}{2}\int_{\R^6} V$ replaced by 
$$
a_0 =\inf\left\{ \frac{3}{2} \int_{\R^6}|\nabla f|^2 + \frac{1}{2} \int_{\R^6} V|f|^2 : \lim_{|X|\to \infty} f(X)=1 \right\}. 
$$
In fact, $a_0<b_0$ if $V\not\equiv 0$ and $\sqrt[4]{a_0/(4\pi^3)}$ plays the role of the scattering length of $V$ in $\R^6$. Deriving the quintic NLS equation in the Gross-Pitaevskii like regime is a very interesting open problem, which certainly requires a substantial improvement of the current method. 

\subsection{Norm approximation}

To describe the fluctuations around the condensate, it is convenient to switch to a Fock space representation where the number of particles is not fixed. We define the bosonic Fock space 
\[ \F = \bigoplus_{n \geq 0} L^2_s (\R^{3n}). \]
The creation operator $a^* (f)$ and the annihilation operator $a(f)$, for some $f \in L^2 (\R^3)$, are defined as 
\begin{align*}
(a^* (f) \Psi )^{(n)} (x_1,\dots,x_{n})&= \frac{1}{\sqrt{n}} \sum_{j=1}^{n} f(x_j) \Psi^{(n-1)} (x_1,\dots,x_{j-1},x_{j+1},\dots, x_{n+1}), \\
(a(f) \Psi)^{(n)} (x_1,\dots,x_{n}) &= \sqrt{n+1} \int \overline{f(x_n)} \Psi^{(n+1)} (x_1,\dots,x_n, x_{n+1}) \d x_{n+1}
\end{align*} 
for all $\Psi \in \cF$. They satisfy the canonical commutation relations (CCR)
\begin{equation}\label{eq:CCR}
[a(f),a(g)]=[a^*(f),a^*(g)]=0,\quad [a(f), a^* (g)]= \langle f, g \rangle, \quad \forall f,g \in L^2(\R^3).
\end{equation}
We will also use the operator-valued distributions $a_x^*$ and $a_x$ defined via
\begin{equation}\label{eq:opval}
a^*(f)=\int_{\R^3}  f(x) a_x^* \d x, \quad a(f)=\int_{\R^3} \overline{f(x)} a_x \d x, \quad \forall f\in  L^2(\R^3),
\end{equation}
which satisfy the CCR  
$$[a^*_x,a^*_y]=[a_x,a_y]=0, \quad [a_x,a^*_y]=\delta(x-y), \quad \forall x,y\in \R^3.$$

The creation and annihilation operators provide a convenient way to express many operators on Fock space in compact forms. For example, if $A$ is a self-adjoint operator on the one-particle space $L^2 (\R^3)$ with kernel $A(x,y)$, then we can define its quantization on Fock space by 
$$\dGamma (A) = \bigoplus_{n=0}^\infty \sum_{j=0}^n A_{j} = \iint  A(x;y) a_x^* a_y \, dx dy.$$
In particular, the number operator is 
$$
\cN = \dGamma (1) = \int a_x^* a_x  \d x
$$
and the $N$-body Hamiltonian $H_N$ can be extended to an operator on Fock space as
\begin{align} \label{eq:HN-2nd-quan}
H_N = \dGamma(-\Delta) + \frac{1}{6N^2} \iiint V_N(x,y,z) a_x^* a_y^* a_z^* a_x a_y a_z \d x \d y \d z. 
\end{align}

Now we want to go further to analyze the fluctuations around the condensate. We are interested in the norm approximation of the form 
\begin{equation} \label{eq:norm-cv-box}
\boxed{\Psi_N(t) \approx  \sum_{n=0}^N u(t)^{\otimes (N-n)} \otimes_s \psi_n(t) := \sum_{n=0}^N \frac{(a^*(u(t)))^{N-n}}{(N-n)!} \psi_n(t)}
\end{equation}
where 
$$
\Phi(t)=(\psi_n(t))_{n=0}^\infty \in \F_+(t) = \bigoplus_{n=0}^\infty \Big( \{u(t)\}^\bot \Big)^{\otimes_s n}$$ 
describes excited particles, with
$$
 \{u(t)\}^\bot = Q(t) L^2(\R^3), \quad Q(t)=1-|u(t)\rangle \langle u(t)|.
$$

The natural candidate for the condensate state $u(t)$ can be obtained by formally inserting the purely uncorrelated ansatz  $u(t)^{\otimes N}$ into the Schr\"odinger equation \eqref{eq:MBSch}, leading to the quintic Hartree equation
\begin{equation}
\label{eq:Hartree}
\left\{
\begin{aligned}
i\partial_tu(t,x) &=   \Big( - \Delta +  \frac{1}{2}\iint |u(t,y)|^2 V_N(x-y,x-z)|u(t,z)|^2\d y \d z\Big) u(t,x) \\
u(0,x) &= u_0(x). 
\end{aligned}
\right.
\end{equation}
We have ignored the $N$-dependence in the notation of $u(t)$ for simplicity. The weak convergence \eqref{eq:delta-int} implies that $u(t)$ converges to the solution $\varphi(t)$ of the quintic NLS \eqref{eq:Hartree-NLS} when $N\to \infty$. The difference between $u(t)$ and $\varphi(t)$ is not visible in the leading order of the condensate (Theorem \ref{thm:main1}). However, the choice of $u(t)$ is better for the refined estimate  \eqref{eq:norm-cv-box}. The behavior of $u(t)$ can be controlled in a uniform way in $N$ (see Section \ref{sec:Hartree} for details).

Bogoliubov's approximation \cite{Bogoliubov-47} suggests that the excited state $\Phi(t)=(\psi_n(t))_{n=0}^\infty$ is determined by the Bogoliubov equation 
\begin{equation}
\label{eq:Bog}
\left\{
\begin{aligned}
i\partial_t \Phi(t) &= \bH(t) \Phi(t), \\
\Phi(0) &= \Phi_0
\end{aligned}
\right.
\end{equation}
with the quadratic generator
\begin{align} \label{eq:Bog-Hamiltonian}
\bH(t) =  \d\Gamma(h+K_1) + \frac{1}{2} \iint \Big( K_2(t,x,y) a_x^* a_y^* + \overline{K_2(t,x,y)}a_x a_y \Big).
\end{align} 
Here
$$
h(t) := -\Delta + \frac{1}{2} \iint  V_N(x-y,x-z)|u(t,y)|^2|u(t,z)|^2 \d y \d z
$$
is the one-body Hartree operator appearing in \eqref{eq:Hartree},  and 
\begin{align*}
K_1 &:= Q(t)\tilde{K}_1 Q(t), \quad \tilde{K}_1[f](x) = \iint V_N(x-y,x-z) |u(t,z)|^2 \overline{u(t,y)} u(t,x) f(y) \d y \d z, \\
K_2 &:= \Big(Q(t)\otimes Q(t)\Big)\tilde{K}_2, \quad \tilde{K}_2(x,y) = \Big(\int V_N(x-y,x-z) |u(t,z)|^2  \d z \Big) u(t,x)u(t,y).
\end{align*}

The existence and uniqueness of the solution to the Bogoliubov equation \eqref{eq:Bog} is well-known \cite{LewNamSch-15} (see Section \ref{sec:Bog-dyn} for further discussions). 

Our second result is a rigorous derivation for the Bogoliubov equation \eqref{eq:Bog}.

\begin{theorem}[Norm approximation] \label{thm:main2} Assume that $0\le V \in \cC_c(\R^6)$. Let $u(t)$ be the Hatree dynamics \eqref{eq:Hartree} with an initial state $u(0)\in H^4(\R^3)$, $\|u(0)\|_{L^2}=1$. Let $\Phi(t)=(\psi_n(t))_{n=0}^\infty$ be the Bogoliubov dynamics \eqref{eq:Bog} such that the initial state $\Phi(0)$ satisfies
\begin{equation} \label{eq:Phi0-mainthm}
\Phi_0 \in \bigoplus_{n=0}^\infty \Big( \{u(0)\}^\bot \Big)^{\otimes_s n},\quad \|\Phi(0)\|=1,  \quad \langle\Phi_0,\dGamma(1-\Delta)\Phi_0\rangle \le C. 
\end{equation}
Let $\Psi_N(t)$ be the Schr\"odinger dynamics \eqref{eq:MBSch} with the initial state
$$
\Psi_{N,0} = \sum_{n=0}^N u(0)^{\otimes (N-n)} \otimes_s \psi_n(0). 
$$
Assume that $0<\beta<1/6$. Then for all $t>0$ and for all $\alpha<(1-6\beta)/4$ fixed, we have
\begin{equation}
\left\|\Psi_N(t)- e^{-i\int_0^t \chi(s) \d s}\sum_{n=0}^N u(t)^{\otimes (N-n)} \otimes_s \psi_n(t)  \right\|_{L^2((\R^3)^N)}^2 \le C_{t,\alpha} N^{-\alpha}
\end{equation}
where
\begin{equation} \label{eq:chi}
\chi(t)=  \frac{2N+3}{6} \iiint V_N(x-y,x-z)|u(t,x)|^2|u(t,y)|^2|u(t,z)|^2 \d x\d y\d z.
\end{equation}
\end{theorem}

As explained in \cite{NamNap-15}, the one-particle density matrix $(\gamma(t), \alpha(t))$ of the Bogoliubov dynamics $\Phi(t)$, defined by the kernels
$$
\gamma(t,x,y)= \left\langle \Phi(t), a_y^* a_x \Phi(t) \right\rangle,\quad  \alpha(t,x,y)= \left\langle \Phi(t), a_x a_y  \Phi(t)\right\rangle,
$$
is the unique solution to the 
\begin{equation} \label{eq:linear-Bog-dm} 
\left\{
\begin{aligned}
i\partial_t \gamma &= h \gamma - \gamma h + K_2 \alpha - \alpha^* K_2^*, \\
i\partial_t \alpha &= h \alpha + \alpha h^{\rm T} + K_2  + K_2 \gamma^{\rm T} + \gamma K_2,\\
\gamma(0)&=\gamma_{\Phi_0}, \quad \alpha(0)  = \alpha_{\Phi_0}.
\end{aligned}
\right.
\end{equation}
Thus our result in Theorem \ref{thm:main2} also gives a rigorous derivation for \eqref{eq:linear-Bog-dm} as an effective description for the density of the  excited particles. Moreover, note that if $\Phi_0$ is a {\em quasi-free state}, then the solution $\Phi(t)$ to the Bogoliubov equation \eqref{eq:Bog} is a quasi-free state for all $t>0$ and \eqref{eq:linear-Bog-dm} is indeed equivalent to the Bogoliubov equation \eqref{eq:Bog}. Nevertheless, our Theorem \ref{thm:main2} works in a general situation and does not require the quasi-free restriction.

Our result in Theorem \ref{thm:main2} extends the norm approximation obtained by X. Chen \cite{Chen-12} on the mean-field case $\beta=0$ (to be precise, the work in \cite{Chen-12} deals with the setting of the fluctuations around coherent states in Fock space rather than the fluctuations around factorized states in $N$-particle space, but our method applies for both cases). Our analysis is different from \cite{Chen-12} and will be explained below. 

\subsection{Ideas of the proofs} Now let us quickly explain the main ingredients of the proofs of Theorem  \ref{thm:main1} and Theorem  \ref{thm:main2}. 

First at all, as a preliminary step, we need to prove the well-posedness of the quintic Hartree equation \eqref{eq:Hartree}. In particular, it is important to derive the uniform (i.e. $N$-independent) bound in $H^4(\R^3)$,  
$$ \|u(t,\cdot)\|_{H^4(\R^3)} \le C_t \|u(0,\cdot)\|_{H^4(\R^3)}$$
(which in turn provides uniform bound on $\|u(t,\cdot)\|_{L^\infty}$ and $\|\partial_t u(t,\cdot)\|_{L^\infty}$ by Sobolev's embedding). The proof requires nontrivial modifications from the analysis for the quintic NLS \eqref{eq:Hartree-NLS}  in \cite{ColKeelStafTao-08}. More precisely, we will treat the  quintic Hartree equation \eqref{eq:Hartree} as a perturbation of the quintic NLS \eqref{eq:Hartree-NLS} and use the method developed in \cite{ColKeelStafTao-08} to extend the Strichartz's estimate in $L^{10}_{t,x}$ for the quintic Hartree solution.

Next, we start the many-body analysis with the general approach as in the pair-interaction case \cite{LewNamSch-15,NamNap-15,NamNap-17,BreNamNapSch-17}. This approach is based on a unitary transformation $U_N(t)$ introduced in \cite{LewNamSerSol-15}  which maps the original $N$-particle space $L^2((\R^3)^N)$ to the truncated Fock space built up on the orthogonal complement $\{u(t)\}^\bot$ of quintic Hartree equation solution: 
$$
U_N(t) : \sum_{n=0}^N u(t)^{\otimes (N-n)} \otimes_s \psi_n(t) \mapsto \bigoplus_{n=0}^N \psi_n(t)
$$
Heuristically, this operator $U_N(t)$ factors out the condensate and implements the c-number substitution in Bogoliubov's idea \cite{Bogoliubov-47}. Thus it remains to analyze the transformed dynamics 
 $$\Phi_N(t)=U_N(t) \Psi_N(t)$$
 in the excited Fock space. The assertion in Theorem \ref{thm:main1} is essentially equivalent to
$$\langle \Phi_N(t), \cN \Phi_N(t) \rangle \ll N. $$

To propagate the latter bound in time, we need to show that the generator of $\Phi_N$ can be controlled by its kinetic part. The main difficulty lies on the fact that the interaction part of the generator of $\Phi_N$ depends heavily on $N$ and behaves badly when there are too many particles.  To overcome this difficulty, we will use the localization technique in Fock space and focus on low particle sectors. This idea was used also in the pair-interaction case \cite{LewNamSch-15,NamNap-15,NamNap-17,BreNamNapSch-17}. However, while in the pair-interaction case it is sufficient to restrict to $\cN \ll N$, in our three-body interaction case we need to restrict further to $\cN \ll N^{1-2\beta}$. This is due to the cubic term 
\begin{equation} \label{eq:cubic-term-intro}
\frac{1}{\sqrt{N}}\iiint V_N(x-y,x-z) a_x^* a_y^* a_z^* \d x \d y \d z
\end{equation}
which does not appear in the pair-interaction case. By using the diagonalization  result on the quadratic Hamiltonian \cite{NamNapSol-16}, we can bound the above cubic term by
$$
\eta \dGamma(1-\Delta) + \eta^{-1}C_t M N^{4\beta-1}, \quad \forall \eta\ge C_t\sqrt{MN^{2\beta-1}}
$$
on the truncated Fock space of $\cN\le M$. The condition $M\ll N^{1-2\beta}$ is necessary to take $\eta$ of order 1. This leads to a good kinetic bound for the truncated dynamics $\Phi_{N,M}$, which has a generator similar to that of $\Phi_N$ but restricted to the truncated space $\cN\le M$. 

To complete the proof of Theorem \ref{thm:main1}, we need to show that the truncated dynamics $\Phi_{N,M}$ is sufficiently close to $\Phi_N$. Heuristically, this step is doable if $M$ is sufficiently large, namely the effect of the cut-off $\cN\le M$ is negligible. Technically, this step will be done by comparing the two generators and using the kinetic estimate of $\Phi_{N,M}$ (plus  the round kinetic bound $O(N)$ for $\Phi_{N}$), resulting the condition $M\gg N^{4\beta}$.   Putting the latter condition together with the previous one $M\ll N^{1-2\beta}$, we obtain the net condition $\beta<1/6$ at the end. 

The norm approximation in Theorem \ref{thm:main2} requires to compare $\Phi_N(t)$ with the Bogoliubov dynamics $\Phi(t)$. From the proof of Theorem \ref{thm:main1}, we know that $\Phi_N$ is close to $\Phi_{N,M}$ if $M\gg N^{4\beta}$. Therefore, it is natural to compare the truncated dynamics $\Phi_{N,M}$ with the Bogoliubov dynamics $\Phi(t)$. It turns out that this step can be done if $M\ll N^{1-5\beta}$. Thus if $\beta<1/9$ then the norm approximation follows. To improve the range of $\beta$, we will use an iteration technique: we will compare $\Phi_{N,M}$ with a further truncated dynamics $\Phi_{N,\widetilde M}$ with $\widetilde M \ll M$ (where we can use the improved kinetic bounds for both, instead of using the round kinetic bound $O(N)$ for $\Phi_N$). This technique allows us to bring down the cut-off parameter $\widetilde M$ to $N^{\beta}$ (after many but finite iteration steps). And finally, we compare  $\Phi_{N,\widetilde M}$ with $\Phi(t)$, which  requires $\widetilde M\ll N^{1-5\beta}$. All this leads to the condition $\beta<1/6$ for the norm approximation in Theorem \ref{thm:main2}, which is fortunately the same as the condition in Theorem \ref{thm:main1}. 

In order to extend our results  to more singular potentials (i.e. larger $\beta$'s), it is crucial to have a better bound for the cubic term \eqref{eq:cubic-term-intro}. This issue goes beyond the current knowledge of Bogoliubov theory and seems very interesting. We refer to \cite{BocBreCenSch-18} for a recent important contribution to the analysis of the cubic term in the pair-interaction case. The problem in the three-body interaction case, however, is completely open. 

\medskip

\noindent 
{\bf Organization of the paper.} We will discuss the quintic Hartree equation \eqref{eq:Hartree} in Section \ref{sec:Hartree} and Bogoliubov equation \eqref{eq:Bog} in Section \ref{sec:Bog-dyn}. Then in Section \ref{sec:gen-stra} we explain the general strategy to derive these effective equations from the many-body Schr\"odinger equation \eqref{eq:MBSch}. Then we settle in Section \ref{sec:op-Fock} key operator estimates on Fock space. Our main Theorems \ref{thm:main1} and  \ref{thm:main2} are proved in  Sections \ref{sec:thm1} and \ref{sec:thm2}, respectively. In Appendix, we explain the extension of our results with $\R^3$ replaced by the torus $\T^3$.

\section{Quintic Hartree equation}\label{sec:Hartree}
In this section we study the well-posedness of the quintic Hartree equation \eqref{eq:Hartree},
$$
\left\{
\begin{aligned}
i\partial_tu(t,x) &= \Big(- \Delta +  \frac{1}{2}\iint  V_N(x-y,x-z)|u(t,y)|^2|u(t,z)|^2\d y \d z \Big) \,u(t,x) \\
u(0,x) &= u_0(x). 
\end{aligned}
\right.
$$

We will prove 
\begin{theorem}[Uniform estimates for quintic Hartree equation] \label{thm:Hartree}
	Let $u_0 \in H^4(\R^3)$ be an initial state. Then for every time $T>0$,  there exists a unique solution to equation~\eqref{eq:Hartree} on $[0,T]$ and it satisfies
	\begin{align}
	\label{eq:H2bound1}
	\|u(t,\cdot)\|_{H^4(\R^3)} \le C_t, \quad \|\partial_t u(t,\cdot)\|_{H^2(\R^3)} \le C_t.
	\end{align}
	Here the constant $C_t$ is dependent on $\|u_0\|_{H^4}$ and $t$ (it is continuous and increasing in $t$), but independent of $N$. 
	
	Moreover, when $N\to \infty$, $u(t)$ converges to the solution $\varphi(t)$ to the quintic NLS \eqref{eq:Hartree-NLS} (with the same initial condition $\varphi_0=u_0$): \begin{equation} \label{eq:H-to-NLS}
		\|u(t,\cdot) - \phi(t,\cdot)\|^2_{L^2(\R^3)} \le C_t N^{-\beta}.
		\end{equation}

\end{theorem} 
\begin{remark}
	Note that from \eqref{eq:H2bound1} and Sobolev's inequality we obtain 
	$$\|u(t,\cdot)\|_{L^\infty(\R^3)} \le C_t, \quad \|\partial_t u(t,\cdot)\|_{L^\infty} \le C_t.$$
	These bounds will be used repeatedly in the paper. 
	\end{remark}

The well-posedness result of the energy-critical NLS \eqref{eq:Hartree-NLS}, 
\begin{equation*}
\left\{
\begin{aligned}
i\partial_t\varphi(t,x) &= - \Delta \varphi(t,x) + b_0 |\varphi(t,x)|^4\varphi(t,x) \\
\varphi(0,x) &= \varphi_0(x) 
\end{aligned}
\right.
\end{equation*}
was proved in 2008 by Colliander, Keel, Staffilani, Takaoka, and Tao \cite{ColKeelStafTao-08}. The adaptation from the local equation \eqref{eq:Hartree-NLS} to the nonlocal one \eqref{eq:Hartree} is not obvious and we will explain the details below. Note that a similar result to Theorem \ref{thm:Hartree} for the cubic Hartree equation (which involves only a two-body interaction potential) was proved in 2013 by Grillakis and Machedon \cite{GriMac-13}. It turns out that the analysis for the 3-body case is significantly more complicated and we could not simply follow the analysis in  \cite{GriMac-13}. 

Our proof of Theorem  \ref{thm:Hartree} is organized as follows. First, in Section \ref{sec:Har-local} we will prove the existence of uniqueness of a local solution to \eqref{eq:Hartree}. This step is standard; we will follow  \cite{caz-03,Tao-06,LinPo-15} and the references therein. Next, in Section \ref{sec:Har-sho} we extend the local solution to a global one and derive the $N$-independent estimates. This is the crucial step where we need to interpret \eqref{thm:Hartree} as a perturbation of \eqref{eq:Hartree-NLS}. By employing the time-extension technique in  \cite{ColKeelStafTao-08}, we can go from local estimate to global estimate. This will be explained in Sections and \ref{sec:Har-sho} and \ref{sec:Har-lon}. Finally, we conclude the proof of Theorem  \ref{thm:Hartree} in Section \ref{sec:har-con}.

\subsection{Local well-posedness} \label{sec:Har-local}

First we prove the existence and uniqueness  of a solution to equation~\eqref{eq:Hartree} in a short time interval.  

\begin{lemma} [\textbf{Local existence of quintic Hartree equation}]\label{lem:localtheory} 
	For every $u_0 \in H^1(\R^3)$, there exists a constant $\eps>0$ depending on $\|u_0\|_{H^1(\R^3)}$ such that for any time interval $I$ containing $0$ with
	\begin{align}
	\label{linearsmall}
	\|e^{it\Delta}u_0\|_{L^{6}_tW^{1,\frac{18}{7}}_x(I\times\R^3)} \le \eps,
	\end{align}
	there exists a unique solution $u\in C(I,H^1(\R^3))\cap L^6_tW^{1,\frac{18}{7}}_x(I\times\R^3)$ of \eqref{eq:Hartree}.
\end{lemma} 

Here the smallness condition \eqref{linearsmall} basically requires that $|I|$ is sufficiently small (which corresponds to the locality in time). 

\begin{proof}
	
	We will use a fixed-point argument similarly to the the energy-critical NLS case \cite[Theorem 5.5]{LinPo-15} (see also \cite{caz-03,Tao-06} and the references therein). 	Let
	\begin{align}
	E(I,a) = \Big\{&v \in C(I:H^1(\R^3))\cap L^6_tW^{1,\frac{18}{7}}_x(I\times\R^3):\quad \|v\|_{L^{6}_tW^{1,\frac{18}{7}}_x(I\times\R^3)}\le a\Big\},
	\end{align}
	equipped with the norm
	\begin{align*}
	\|v\|_{E(I,a)} = \sup_{t\in I}\|v(t)\|_{H^1(\R^3)}+\|v\|_{L^{6}_tW^{1,\frac{18}{7}}_x(I\times\R^3)}.
	\end{align*}
	This makes $E(I,a)$ a complete metric space.
	Define for $u\in E(I,a)$
	
	\begin{align}
	\Phi(u)(t) = e^{it\Delta}u_0 - \frac{i}{2}\int_0^t e^{i\Delta(t-s)} \iint |u(s,y)|^2 V_N(x-y,x-z)|u(s,z)|^2\d y \d z \,u(s) ds
	\end{align}
	We want to prove that $\Phi$ is a contraction map on $E(I,a)$ and then apply the contraction mapping principle.
	
	In the following let us ignore the time dependence in the notation of $u$ for simplicity. By the product rule for gradient and H\"older's, Young's, Sobolev's inequalities we have 
	\begin{align} \label{firstlocalestimate}
	&\left\|\nabla\left( \iint|u(y)|^2V_N|u(z)|^2\d y \d z\, \right) u\right\|_{L^\frac{6}{5}_tL^\frac{18}{11}_x(I\times\R^3)}\nn\\
	&\lesssim \left\|\left( \iint (|u| |\nabla u|)V_N|u|^2\d y \d z\, \right) u\right\|_{L^\frac{6}{5}_tL^\frac{18}{11}_x(I\times\R^3)} \nn\\
	&\lesssim \|\nabla u\|_{L^6_tL^\frac{18}{7}_x(I\times\R^3)} \left\|\iint|u|^2V_N|u|^2\d y \d z\right\|_{L^\frac{3}{2}_tL^\frac{9}{2}_x(I\times\R^3)} \nn \\ &\lesssim \|\nabla u\|_{L^6_tL^\frac{18}{7}_x(I\times\R^3)} \|V\|_{L^1(\R^6)}\|u\|^4_{L^6_tL^{18}_x(I\times\R^3)} \nn\\ 
	&\lesssim \|\nabla u\|^5_{L^6_tL^\frac{18}{7}_x(I\times\R^3)}.
	\end{align}
	Similarly
	\begin{align*}
	\left\|\iint|u|^2V_N|u|^2\d y \d z\, u\right\|_{L^\frac{6}{5}_tL^\frac{18}{11}_x(I\times\R^3)} \lesssim \| \nabla u\|^4_{L^6_tL^\frac{18}{7}_x(I\times\R^3)} \|u\|_{L^6_tL^\frac{18}{7}_x(I\times\R^3)}.
	\end{align*}
	Hence using \eqref{linearsmall} and Strichartz estimate \cite{Strich-77} (also consider \cite{Tao-06}) we obtain
	\begin{align*}
	&\|\Phi(u)\|_{L^6_tW^{1,\frac{18}{7}}_x(I\times\R^3)} \\&\le \|e^{it\Delta}u_0\|_{L^{6}_tW^{1,\frac{18}{7}}_x(I\times\R^3)} +  \left\|\frac{i}{2}\int_0^t e^{i\Delta(t-s)} \iint |u(s)|^2 V_N|u(s)|^2\d y \d z \,u(s) ds\right\|_{L^{6}_tW^{1,\frac{18}{7}}_x(I\times\R^3)} \\
	&\lesssim \eps + \left\|\iint|u|^2V_N|u|^2\d y \d z\, u\right\|_{L^\frac{6}{5}_tW^{1,\frac{18}{11}}_x(I\times\R^3)} \lesssim \eps + \|u\|^5_{L^{6}_tW^{1,\frac{18}{7}}_x(I\times\R^3)}.	
	\end{align*}
	Therefore, for $u\in E(I,a)$ we have
	$
	\|\Phi(u)\|_{L^{6}_tW^{1,\frac{18}{7}}_x(I\times\R^3)} \lesssim \eps + a^5
	$
	and hence choosing $a>0$ small enough and then $\eps>0$ small we have
	\begin{align*} \|\Phi(u)\|_{L^{6}_tW^{1,\frac{18}{7}}_x(I\times\R^3)} \le a.
	\end{align*}
	Moreover, again using Strichartz estimate we have
	\begin{align*}
	\sup_{t\in[0,T]}\|\Phi(u)(t)\|_{H^1(\R^3)} &\lesssim \|u_0\|_{H^1(\R^3)}+ \left\|\iint|u|^2V_N|u|^2\d y \d z\, u\right\|_{L^\frac{6}{5}_tW^{1,\frac{18}{11}}_x(I\times\R^3)} \\ &\lesssim \|u_0\|_{H^1(\R^3)}+ \|u\|^5_{L^{6}_tW^{1,\frac{18}{7}}_x(I\times\R^3)}.
	\end{align*} 
	Putting alltogether we obtain $\Phi(E(I,a)) \subseteq E(I,a)$.

	To show that $\Phi$ is a contraction map we see for $u,v \in E(I,a)$
	\begin{align*}
	\|\nabla(\Phi(u)-\Phi(v))\|_{L^{6}_tL^\frac{18}{7}_x} &\lesssim \left\| \iint |u|^2V_N|u|^2\d y \d z\,\nabla u - \iint |v|^2V_N|v|^2\d y \d z\,\nabla v\right\|_{L^\frac{6}{5}_tL^\frac{18}{11}_x} \\
	&\le\left\| \iint |u|^2V_N|u|^2\d y \d z\,\nabla(u - v)\right\|_{L^\frac{6}{5}_tL^\frac{18}{11}_x} \\
	&+ \left\| \left(\iint |u|^2V_N|u|^2\d y \d z - \iint |v|^2V_N|v|^2\d y \d z\right)\nabla v\right\|_{L^\frac{6}{5}_tL^\frac{18}{11}_x}.
	\end{align*}
	The first term can be estimated similarly to \eqref{firstlocalestimate}, which gives
	\begin{align*}
	&\left\| \iint |u|^2V_N|u|^2\d y \d z\,\nabla( u - v)\right\|_{L^\frac{6}{5}_tL^\frac{18}{11}_x} \lesssim \|\nabla u\|^4_{L^6_tL^{\frac{18}{7}}_x(I\times\R^3)}\|\nabla(u-v)\|_{L^{6}_tL^\frac{18}{7}_x} \\
	&\le a^4\|\nabla(u-v)\|_{L^{6}_tL^\frac{18}{7}_x}.
	\end{align*}
	The second one follows by 
	\begin{align*}
	&\left\| \left(\iint |u|^2V_N|u|^2\d y \d z - \iint |v|^2V_N|v|^2\d y \d z\right)\nabla v\right\|_{L^\frac{6}{5}_tL^\frac{18}{11}_x(I\times\R^3)} \\ &\le \left\|\iint |u|^2V_N|u|^2\d y \d z - \iint |v|^2V_N|v|^2\d y \d z\right\|_{L^\frac{3}{2}_tL^\frac{9}{2}_x(I\times\R^3)}\|\nabla v\|_{L^{6}_tL^\frac{18}{7}_x(I\times\R^3)} \\
	&\lesssim \left\|\iint (|u|^2-|v|^2)V_N|u|^2\d y \d z\right\|_{L^\frac{3}{2}_tL^\frac{9}{2}_x(I\times\R^3)}\|\nabla v\|_{L^{6}_tL^\frac{18}{7}_x(I\times\R^3)} \\
	&\lesssim \left\|\iint |u-v|(|u|+|v|)V_N|u|^2\d y \d z\right\|_{L^\frac{3}{2}_tL^\frac{9}{2}_x}\|\nabla v\|_{L^{6}_tL^\frac{18}{7}_x(I\times\R^3)} \\
	&\lesssim\|V_N\|_{L^1(\R^6)}\left\|\|u-v\|_{L^{18}_x(\R^3)}\||u|+|v|\|_{L^{18}_x(\R^3)}\|u\|^2_{L^{18}_x(\R^3)}\right\|_{L^\frac{3}{2}_t(I)}\|\nabla v\|_{L^{6}_tL^\frac{18}{7}_x(I\times\R^3)} \\
	&\lesssim \left((\|\nabla u\|_{L^{6}_tL^\frac{18}{7}_x}+\|\nabla v\|_{L^{6}_tL^\frac{18}{7}_x})\|\nabla u\|^2_{L^{6}_tL^\frac{18}{7}_x}\|\nabla v\|_{L^{6}_tL^\frac{18}{7}_x}\right) \|\nabla(u-v)\|_{L^{6}_tL^\frac{18}{7}_x(I\times\R^3)} \\
	&\lesssim a^4\|\nabla(u-v)\|_{L^{6}_tL^\frac{18}{7}_x(I\times\R^3)}.
	\end{align*}
	Putting all together we obtain
	\begin{align*}
	\|\nabla(\Phi(u)-\Phi(v))\|_{L^{6}_tL^\frac{18}{7}_x} \lesssim a^4\|\nabla(u-v)\|_{L^{6}_tL^\frac{18}{7}_x(I\times\R^3)}.
	\end{align*}
	Similarly we obtain
	\begin{align*}
	\|\Phi(u)-\Phi(v)\|_{L^{6}_tL^\frac{18}{7}_x} \lesssim a^4 \|u-v\|_{L^{6}_tL^\frac{18}{7}_x}
	\end{align*}
	and together
	\begin{align}
	\|\Phi(u)-\Phi(v)\|_{L^{6}_tW^{1,\frac{18}{7}}_x} \lesssim a^4 \|u-v\|_{L^{6}_tW^{1,\frac{18}{7}}_x}.
	\end{align}
	Furthermore, using the above and Strichartz estimate we obtain
	\begin{align}
	\sup_{t \in I}\|\Phi(u)-\Phi(v)\|_{H^1(\R^3)} &\lesssim \left\| \iint |u|^2V_N|u|^2\d y \d z\, u - \iint |v|^2V_N|v|^2\d y \d z\, v\right\|_{L^\frac{6}{5}_tW^{1,\frac{18}{11}}_x} \nn\\
	&\lesssim a^4 \|u-v\|_{L^{6}_tW^{1,\frac{18}{7}}_x}.
	\end{align}
	By this we have
	\begin{align*}
	\|\Phi(u) - \Phi(v)\|_{E(I,a)} \lesssim a^4 \|u-v\|_{E(I,a)}
	\end{align*}
	which gives for $a>0$ small enough that $\Phi$ is a contraction map on $E(I,a)$. Using the contracting mapping principle we have that there exists a unique \mbox{}\\$u \in C(I:H^1(\R^3))\cap L^6_tW^{1,\frac{18}{7}}_x(I\times\R^3)$ which solves \eqref{eq:Hartree} with initial data $u_0\in H^1(\R^3)$.
\end{proof}

In the following we will derive the global theory of equation~\eqref{eq:Hartree}. We will use the well-known global well-posedness theory of the energy-critical NLS
\begin{equation}
i\partial_t \tilde{u} = -\Delta\tilde{u} +b_0|\tilde{u}|^4\tilde{u}
\end{equation}
(consider for example \cite{ColKeelStafTao-08} and \cite{Tao-06}) and consider equation~\eqref{eq:Hartree} as a pertubation. For that we will prove an adapted version of \cite[Lemma 3.9 and Lemma 3.10]{ColKeelStafTao-08} for the quintic Hartree equation, which will give global spacetime bounds for solutions of equation~\eqref{eq:Hartree}.

\subsection{Short time pertubation} \label{sec:Har-sho}

\begin{lemma}[\textbf{Short time Pertubation}] \label{lem:Short time Pertubation}
	Let $I\subset \R$ be a compact interval  and $\tilde{u}$ be a solution of
	\begin{equation}
	\label{perteq}
	i\partial_t\tilde{u} = - \Delta\tilde{u} + \frac{1}{2}\iint |\tilde{u}(y)|^2 V_N(x-y,x-z)|\tilde{u}(z)|^2\d y \d z\,\tilde{u} + e
	\end{equation}
	on $I\times\R^3$
	for some function $e$. Assume that we have
	
	\begin{align}
	\label{eq:tiny u}	
	\|\nabla\tilde{u}\|_{L^{10}_t L^\frac{30}{13}_x (I \times \R^3)}	&\le \eps_0, \\
	\label{tiny error}
	\|\nabla e\|_{L^2_t L^\frac{6}{5}_x (I \times \R^3)}	&\le \eps,
	\end{align}
	for some constants $\eps_0, \eps >0$ small enough.
	For $t_0\in I$ let $u(t_0)$  be close to $\tilde{u}(t_0)$. More precisely:
	\begin{align}
	\label{eq:closesness}
	\|\nabla e^{i(t-t_0)\Delta}(u(t_0)-\tilde{u}(t_0))\|_{L^{10}_t 	L^{\frac{30}{13}}_x (I \times \R^3)} \le \eps.
	\end{align}
	We then conclude that there exists a solution $u$ of \eqref{eq:Hartree} on $I\times \R^3$ with initial state $u(t_0)$ at $t_0$,
	which fulfills the following spacetime bounds:
	\begin{align}
	\label{concludetiny1}
	\|u-\tilde{u}\|_{L^{10}_{t,x}(I\times \R^3)} \lesssim \|\nabla(u-\tilde{u})\|_{L^{10}_t 	L^{\frac{30}{13}}_x (I \times \R^3)} \lesssim \eps, \\ 
	\label{concludetiny2}
	\|\nabla(i\partial_t + \Delta) (u-\tilde{u}) \|_{L^2_t L^{\frac{6}{5}}_x(I\times\R^3)}	\lesssim \eps.
	\end{align} 
\end{lemma}
\begin{proof}
	By the local theory given in Lemma \ref{lem:localtheory} we can proove \eqref{concludetiny1}-\eqref{concludetiny2} as a priori estimates, meaning that we assume that $u$ already exists on $I$.
	
	Let $v = u -\tilde{u}$. Define
	\begin{align*}
	S(I) = \|\nabla(i\partial_t + \Delta) v \|_{L^2_tL^\frac{6}{5}_x(I\times\R^3)}.
	\end{align*}
	Using \eqref{eq:closesness}, Strichartz estimate and Duhamel's formula we can estimate the $L^{10}_tL^\frac{30}{13}_x(I\times\R^3)$ Norm of $v$ by $S(I)$:
	
	\begin{align}
	\label{eq:estimate v}
	\|\nabla v\|_{L^{10}_tL^\frac{30}{13}_x(I\times\R^3)} &\le 
	\|\nabla (v - e^{i(t-t_0)\Delta}v(t_0))\|_{L^{10}_tL^\frac{30}{13}_x(I\times\R^3)} + \|\nabla e^{i(t-t_0)\Delta}v(t_0)\|_{L^{10}_tL^\frac{30}{13}_x(I\times\R^3)} \nn\\
	&\lesssim S(I) + \eps.
	\end{align}
	On the other hand we know that $v$ solves the following equation:
	\begin{align}
	\label{eq:dynamicsv}
	i\partial_t v = & - \Delta v + \frac{1}{2}\iint |v(y)+\tilde{u}(y)|^2 V_N(x-y,x-z)|v(z)+\tilde{u}(z)|^2\d y \d z(v+\tilde{u})\nn \\ &- \frac{1}{2}\iint |\tilde{u}(y)|^2 V_N(x-y,x-z)|\tilde{u}(z)|^2\d y \d z\tilde{u} - e
	\end{align} 
	
	In order to estimate $S(I)$ we have to estimate several terms of \eqref{eq:dynamicsv} which contain different powers of $v$ and $\tilde{u}$. Since all of the terms can be estimated in the same strategy, we just give a few examples. We start with the term $(\iint |v(y)|^2 V_N(x-y,x-z)|v(z)|^2\d y \d z )\nabla v$.	By H\"older's inequality we have
	
	\begin{align*}
	&\left\|\iint |v(y)|^2 V_N(x-y,x-z)|v(z)|^2\d y \d z \nabla v \right\|_{L^2_tL^\frac{6}{5}_x(I\times\R^3)} 
	\\ &\le \left\|\iint |v(y)|^2 V_N(x-y,x-z)|v(z)|^2\d y \d z\right\|_{L^\frac{5}{2}_tL^\frac{5}{2}_x(I\times\R^3)} \|\nabla v\|_{L^{10}_tL^\frac{30}{13}_x(I\times\R^3)}.
	\end{align*}
	The first term can be estimated using first  Young's inequality and then Sobolev's inequality and \eqref{eq:estimate v}:
	
	\begin{align*}
	&\left\|\iint |v(y)|^2 V_N(x-y,x-z)|v(z)|^2\d y \d z\right\|_{L^\frac{5}{2}_tL^\frac{5}{2}_x} \\
	&	\lesssim \left\|\iint |v(y)|^4 V_N(x-y,x-z)\d y \d z\right\|_{L^\frac{5}{2}_tL^\frac{5}{2}_x} \\
	&\le \left\| \|V\|_{L^1(\R^6)} \|v\|_{L^{10}_x(\R^3)}^4 \right\|_{L^\frac{5}{2}_t(I)} 
	 \lesssim \,	\|v\|^4_{L^{10}_{t,x}(\R^3)} \\
	& \lesssim \|\nabla v\|_{L^{10}_tL^\frac{30}{13}_x(I\times\R^3)}^4 \lesssim (S(I) + \eps)^4.
	\end{align*}
	Hence, using the above and \eqref{eq:estimate v} agiain we obtain
	\begin{align*}
	\left\|\iint |v(y)|^2 V_N(x-y,x-z)|v(z)|^2\d y \d z \nabla v \right\|_{L^2_tL^\frac{6}{5}_x(I\times\R^3)} \lesssim (S(I)+\eps)^5.
	\end{align*}
	For the readers convenience we also consider the mixed term  $(\iint |\tilde{u}|^2 V_N\nabla|v|^2 \d y \d z) v$ as a second example. First we use again H\"older's inequality and get
	
	\begin{align*}
	\left\|\iint |\tilde{u}|^2 V_N\nabla|v|^2 \d y \d z \, v\right\|_{L^2_tL^\frac{6}{5}_x(I\times\R^3)} \le
	\left\|\iint |\tilde{u}|^2 V_N\nabla|v|^2 \d y \d z\right\|_{L^\frac{5}{2}_tL^\frac{15}{11}_x(I\times\R^3)} \|v\|_{L^{10}_{t,x}(I\times\R^3)}.
	\end{align*}
	The first term can then be estimated by Minkowski's inequality
	\begin{align*}
	&\left\|\iint\ |\tilde{u}|^2 V_N\nabla|v|^2 \d y \d z\right\|_{L^\frac{5}{2}_tL^\frac{15}{11}_x(I\times\R^3)} \lesssim
	\left\|\iint|\tilde{u}|^2 V_N|\nabla v| |v| \d y \d z\right\|_{L^\frac{5}{2}_tL^\frac{15}{11}_x(I\times\R^3)} \\ & \le
	\left\| \iint V_N(y,z) \left(\int |\tilde{u}(x-y)|^{\frac{30}{11}}|v(x-z)|^{\frac{15}{11}}|\nabla v(x-z)|^{\frac{15}{11}} \d x \right)^{\frac{11}{15}} \d y \d z \right\|_{L^\frac{5}{2}_t(I)}\\
	&\le \left\| \|V\|_{L^1(\R^6)} \|\tilde{u}\|_{L^{10}_x(\R^3)}^2 \|v\|_{L^{10}_x(\R^3)} \|\nabla v\|_{L^\frac{30}{13}_x(\R^3)} \right\|_{L^\frac{5}{2}_t(I)} \\
	&\lesssim \|\tilde{u}\|^2_{L^{10}_{t,x}(I\times\R^3)}\|v\|_{L^{10}_{t,x}(I\times\R^3)}\|\nabla v\|_{L^{10}_tL^\frac{30}{13}_x(I\times\R^3)} \\
	&\lesssim \|\nabla\tilde{u}\|^2_{L^{10}_tL^\frac{30}{13}_x(I\times\R^3)}\|\nabla v\|^2_{L^{10}_tL^\frac{30}{13}_x(I\times\R^3)} \lesssim 
	\eps_0^2 \left(S(I)+\eps\right)^2.
	\end{align*}
	Since by Sobolev's inequaltiy and \eqref{eq:estimate v} we have that $\|v\|_{L^{10}_{t,x}(I\times\R^3)} \lesssim \|\nabla v\|_{L^{10}_tL^\frac{30}{13}_x(I\times\R^3)} \lesssim S(I) + \eps$, we conclude
	\begin{align*}
	\left\|\iint |\tilde{u}|^2 V_N\nabla|v|^2 \d y \d z \, v\right\|_{L^2_tL^\frac{6}{5}_x(I\times\R^3)} \lesssim \eps_0^2 \left(S(I)+\eps\right)^3
	\end{align*}

	These and similar estimates yield
	\begin{align} \label{eq:St<=Stpower}
	S(I) \lesssim \eps + \sum_{j=1}^5 \left(S(I) + \eps \right)^j \eps_0^{5-j}.
	\end{align} 
	If we choose $\eps,\eps_0$ small enough, a standard continuity argument gives $S(I) \lesssim \eps$. To be precise, let us take a small, fixed constant $\lambda>0$ independent of $\eps,\eps_0$  and assume that 
	\begin{align} \label{eq:St<=Stpower-a}S(I)\le \lambda. 
	\end{align}
	Then from \eqref{eq:St<=Stpower} we deduce that, for $\eps,\eps_0,\lambda$ sufficiently small, 
	$$
	S(I) \le C_0\eps + C_0\left( \sum_{j=1}^5 \left(\lambda + \eps \right)^{j-1} \eps_0^{5-j} \right) (S(I)+\eps) \le C_0 \eps + \frac{1}{2} (S(I)+\eps),
	$$
	and hence
	\begin{align} \label{eq:St<=Stpower-b}
	S(I)\le (2C_0+1)\eps. 
	\end{align}
	Here $C_0>0$ is a constant independent of $\eps,\eps_0,\lambda$. Finally, note that as a priori the assumption \eqref{eq:St<=Stpower-a} holds only if $|I|$ is sufficiently small. However, using \eqref{eq:St<=Stpower-b} with $\eps$ much smaller than $\lambda$ we can extend the bounds \eqref{eq:St<=Stpower-a} and  \eqref{eq:St<=Stpower-b} to any large interval $I$.

	This proves \eqref{concludetiny2}. Using	\eqref{eq:estimate v} we also conclude \eqref{concludetiny1}. This finishes the proof.
\end{proof}

\subsection{Long time pertubation} \label{sec:Har-lon}
We now prove a version of Lemma~\ref{lem:Short time Pertubation} without the smallness condition \eqref{eq:tiny u} by using Lemma~\ref{lem:Short time Pertubation} iteratively.
\begin{lemma} [\textbf{Long time Pertubation}]
	\label{Long time Pertubation}
	Let $I$ be a compact interval and $\tilde{u}$ be a function on $I\times\R^3$ which solves
	\begin{equation}
	i\partial_t\tilde{u} = - \Delta\tilde{u} + \frac{1}{2}\iint |\tilde{u}(y)|^2 V_N(x-y,x-z)|\tilde{u}(z)|^2\d y \d z\,\tilde{u} + e,
	\end{equation}
	for some function $e$ on $I\times\R^3$. Moreover assume that $\tilde{u}$ fulfills the following spacetime bounds:
	\begin{align}
	\label{10spacetime}
	\|\tilde{u}\|_{L^{10}_{t,x}(I\times \R^3)} &\le M, \\
	\|\tilde{u}\|_{L^{\infty}_t \dot H^1_x (I \times \R^3)} &\le E, \\
	\|\nabla e\|_{L^2_t L^{\frac{6}{5}}_x (I \times \R^3)}	&\le \eps
	\end{align}
	for some constants $M,E \ge 0$ and some small enough $\eps > 0$.
	For $t_0\in I$ let $u(t_0)$ be close to $\tilde{u}(t_0)$ in the sense that
	\begin{align}
	\label{eq:closesness2}
	\|\nabla e^{i(t-t_0)\Delta}(u(t_0)-\tilde{u}(t_0)\|_{L^{10}_t 	L^{\frac{30}{13}}_x (I \times \R^3)} \le \eps.
	\end{align}
	We then conclude that there exists a solution $u$  of \eqref{eq:Hartree} on $I\times \R^3$ with initial state $u(t_0)$ at $t_0$,
	which fulfills the following spacetime bounds:
	\begin{align}
	\label{concludetinylong1}
	\|u-\tilde{u}\|_{L^{10}_{t,x}(I\times \R^3)} \lesssim \|\nabla(u-\tilde{u})\|_{L^{10}_t 	L^{\frac{30}{13}}_x (I \times \R^3)} \lesssim C(M,E)\eps, \\ 
	\label{concludetinylong2}
	\|\nabla(i\partial_t + \Delta) (u-\tilde{u}) \|_{L^2_t L^{\frac{6}{5}}_x(I\times\R^3)}	\lesssim C(M,E)\eps.
	\end{align} 
	
\end{lemma}
\begin{proof} 
	Let $\delta >0$. Using \eqref{10spacetime} we can split $I$ into finetly many subintervals $I_1 , ... ,  I_{C(M)}$ such that 
	\begin{align*}
	\|\tilde{u}\|_{L^{10}_{t,x}(I_j \times \R^3)} \le \delta
	\end{align*}
	for each $j$. Using Duhamel's formula, Strichartz's estimate and estimates similar to the ones in Lemma~\ref{lem:localtheory} and Lemma~\ref{lem:Short time Pertubation} we obtain
	\begin{align*}
	&\|\nabla \tilde{u}\|_{L^{10}_t 	L^{\frac{30}{13}}_x (I_j \times \R^3)} \\ 
	&\lesssim
	\sup_{t\in I_j}\|\tilde{u}(t)\|_{\dot H^1_x(\R^3)} + \left\|\nabla \left(\frac{1}{2} \iint |\tilde{u}(y)|^2 V_N(x-y,x-z)|\tilde{u}(z)|^2\d y \d z\,\tilde{u} + e\right)\right\|_{L^2_t L^{\frac{6}{5}}_x (I_j \times \R^3)}  \\
	&\lesssim E + \|\tilde{u}\|^4_{L^{10}_{t,x}(I_j \times \R^3)}\|\nabla \tilde{u}\|_{L^{10}_t L^\frac{30}{13}_x(I_j \times\R^3)} + \eps \\
	&\le E + \delta^4\|\nabla \tilde{u}\|_{L^{10}_t L^\frac{30}{13}_x(I_j \times\R^3)} + \eps.
	\end{align*}
	With $\delta$ small enough this gives $\|\nabla \tilde{u}\|_{L^{10}_t L^\frac{30}{13}_x(I_j\times\R^3)} \lesssim E +\eps$. Summing over all subintervals gives
	\begin{align}
	\label{eq:1010normbound}
	\|\nabla\tilde{u}\|_{L^{10}_t 	L^{\frac{30}{13}}_x (I \times \R^3)} \le C(M,E).
	\end{align}
	If we now choose $\eps_0$ as in Lemma~\ref{lem:Short time Pertubation} and use \eqref{eq:1010normbound} we can split $I$ again into finitely many subintervals $I_1 , ... ,  I_{C(M,E,\eps_0)}$ with $I_j = [t_j,t_{j+1}]$ and
	\begin{align}
	\|\nabla\tilde{u}\|_{L^{10}_t 	L^{\frac{30}{13}}_x (I_j \times \R^3)} \le \eps_0.
	\end{align}
	We can now apply Lemma~\ref{lem:Short time Pertubation} inductively. Using this for the first subinterval $I_0$ gives
	\begin{align*}
	\|u-\tilde{u}\|_{L^{10}_{t,x}(I_0\times \R^3)} \lesssim \|\nabla(u-\tilde{u})\|_{L^{10}_t 	L^{\frac{30}{13}}_x (I_0 \times \R^3)} \lesssim \eps, \\ 
	\|\nabla(i\partial_t + \Delta) (u-\tilde{u}) \|_{L^2_t L^{\frac{6}{5}}_x(I_0\times\R^3)}	\lesssim \eps.
	\end{align*}
	Now proceeding iteratively using Duhamel's formula we see that
	\begin{align*}
	&\|\nabla e^{i\Delta(t-t_1)}(u(t_1)-\tilde{u}(t_1))\|_{L^{10}_t 	L^{\frac{30}{13}}_x (I \times \R^3)} \\ &\le \|\nabla e^{i\Delta(t-t_0)}(u(t_0)-\tilde{u}(t_0))\|_{L^{10}_t 	L^{\frac{30}{13}}_x (I \times \R^3)} \\
	&+ \left\| e^{i\Delta(t-t_1)} i\int_{I_0} e^{i\Delta(t_1-s)} \nabla(i\partial_t + \Delta) (u-\tilde{u})(s) \d s\right\|_{L^{10}_t 	L^{\frac{30}{13}}_x(I\times\R^3)} \\
	&\lesssim \eps +
	\left\|\int_{I_0} e^{i\Delta(t_1-s)} \nabla(i\partial_t + \Delta) (u-\tilde{u})(s) \d s\right\|_{L^2(\R^3)} \\
	&\lesssim \eps +
	\|\nabla(i\partial_t + \Delta) (u-\tilde{u}) \|_{L^2_t L^{\frac{6}{5}}_x(I_0\times\R^3)}
	\lesssim \eps.
	\end{align*}
	For $\eps>0$ small enough we can iterate this procedure and obtain
	\begin{align*}
	\|u-\tilde{u}\|_{L^{10}_{t,x}(I_j\times \R^3)} \lesssim \|\nabla(u-\tilde{u})\|_{L^{10}_t 	L^{\frac{30}{13}}_x (I_j \times \R^3)} \lesssim C(j)\eps, \\ 
	\|\nabla(i\partial_t + \Delta) (u-\tilde{u}) \|_{L^2_t L^{\frac{6}{5}}_x(I_j\times\R^3)}	\lesssim C(j)\eps,
	\end{align*}
	for all $j$. Summing over all finite intervals we obtain \eqref{concludetinylong1} and \eqref{concludetinylong2}. 
\end{proof}
\subsection{Conclusion of Theorem \ref{thm:Hartree}} \label{sec:har-con}
We now apply the two previous Lemmas to prove Theorem~ \ref{thm:Hartree}.
\begin{proof}[Proof of Theorem~\ref{thm:Hartree}.] {\bf Global well-posedness.} 	We want to apply Lemma~\ref{Long time Pertubation} with $\tilde u$ being the solution of the perturbed equation given by the quintic NLS
	\begin{align*}
	i\partial_t \tilde{u} &= -\Delta\tilde{u} +b_0|\tilde{u}|^4\tilde{u} \\
	&=  - \Delta\tilde{u} + \frac{1}{2}\iint |\tilde{u}(y)|^2 V_N(x-y,x-z)|\tilde{u}(z)|^2\d y \d z\,\tilde{u} + e_N,
	\end{align*}
	with initial state $\tilde u(0,x) = u_0$.
	Here we have defined the pertubation
	\begin{equation*}
	e_N = b_0|\tilde{u}|^4\,\tilde{u} - \frac{1}{2}\iint |\tilde{u}(y)|^2 V_N(x-y,x-z)|\tilde{u}(z)|^2\d y \d z\,\tilde{u}.
	\end{equation*}
	In order to use Lemma~\ref{Long time Pertubation} we want to show that $\|\nabla e_N\|_{L^2_t L^\frac{6}{5}_x (I \times \R^3)}$ is arbitrarily small for $N$ large. 
	We see that
	\begin{align*}
	\|\nabla e_N\|_{L^2_t L^\frac{6}{5}_x (I \times \R^3)} &\le \left\|\left(\frac{1}{2}\iint |\tilde{u}(y)|^2 V_N(x-y,x-z)|\tilde{u}(z)|^2\d y \d z- b_0|\tilde{u}|^4 \right)\nabla\tilde{u}\right\|_{L^2_t L^\frac{6}{5}_x (I \times \R^3)}  \\
	&\quad\quad+ \left\|\nabla\left(\frac{1}{2}\iint |\tilde{u}(y)|^2 V_N(x-y,x-z)|\tilde{u}(z)|^2\d y \d z- b_0|\tilde{u}|^4 \right)\tilde{u}\right\|_{L^2_t L^\frac{6}{5}_x (I \times \R^3)} \\
	&\le \left\|\frac{1}{2}\iint |\tilde{u}(y)|^2 V_N(x-y,x-z)|\tilde{u}(z)|^2\d y \d z- b_0|\tilde{u}|^4 \right\|_{L^\frac{5}{2}_tL^\frac{5}{2}_x}\|\nabla\tilde{u}\|_{L^{10}_tL^{\frac{30}{13}}_x} \\
	&\quad\quad+\left\|\nabla\left(\frac{1}{2}\iint |\tilde{u}(y)|^2 V_N(x-y,x-z)|\tilde{u}(z)|^2\d y \d z- b_0|\tilde{u}|^4\right)\right\|_{L^\frac{5}{2}_tL^\frac{15}{11}_x}\|\tilde{u}\|_{L^{10}_{t,x}}.
	\end{align*}
	We now show that the first term $$\left\|\frac{1}{2}\iint |\tilde{u}(y)|^2 V_N(x-y,x-z)|\tilde{u}(z)|^2\d y \d z- b_0|\tilde{u}|^4 \right\|_{L^\frac{5}{2}_tL^\frac{5}{2}_x}$$ will be arbitrarily small for $N$ big enough. The second term will follow similarly. 
	We see by Minkowski's inequality  
	\begin{align}
	\label{eq:eNsmalllast}
	&\nn\left\|\frac{1}{2}\iint |\tilde u(y)|^2 V_N(x-y,x-z)|\tilde u(z)|^2\d y \d z - b_0|\tilde u|^4\right\|_{L^\frac{5}{2}_tL^\frac{5}{2}_x} \\ &\nn = \left\|\frac{1}{2}\iint V_N(y,z)\left(|\tilde u(x-y)|^2|\tilde u(x-z)|^2 - |\tilde u(x)|^4\right)\d y \d z\right\|_{L^\frac{5}{2}_tL^\frac{5}{2}_x} \\&\nn\le\frac{1}{2}
	\iint V_N(y,z) \left\||\tilde u(x-y)|^2|\tilde u(x-z)|^2 - |\tilde u(x)|^4\right\|_{L^\frac{5}{2}_tL^\frac{5}{2}_x}\d y \d z  \\
	&\nn\le 
	\iint V_N(y,z) \left\||\tilde u(x-y)|^2- |\tilde u(x)|^2\right\|_{L^5_tL^5_x}\left\|\tilde u\right\|^2_{L^{10}_tL^{10}_x}\d y \d z  \\
	&\nn\le 
	2\iint V_N(y,z) \left\|\tilde u(x-y)- \tilde u(x)\right\|_{L^{10}_tL^{10}_x}\left\|\tilde u\right\|^3_{L^{10}_tL^{10}_x}\d y \d z \\
	&\nn= 2\iint_{|y| \le C N^{-\beta}} V_N(y,z) \left\|\tilde u(x-y)- \tilde u(x)\right\|_{L^{10}_tL^{10}_x}\left\|\tilde u\right\|^3_{L^{10}_tL^{10}_x}\d y \d z \\
	&\le C\left\|\tilde u\right\|^3_{L^{10}_tL^{10}_x}\sup_{|y|\le CN^{-\beta}}\left\|\tilde u(x-y)- \tilde u(x)\right\|_{L^{10}_tL^{10}_x}.
	\end{align}
	Here we have used the fact that $V_N(y,\cdot)$ is zero for $|y| > CN^{-\beta}$ for some $C>0$. Note that $\tilde u$ is independent of $N$ and $\left\|\tilde u\right\|^2_{L^{10}_tL^{10}_x} \le C$ by \cite[Theorem 1.1]{ColKeelStafTao-08}, we have the continuity by translation
	$$
	\lim_{|y|\to 0}\left\|\tilde u(x-y)- \tilde u(x)\right\|_{L^{10}_tL^{10}_x}=0.
	$$
	Therefore, the right side of \eqref{eq:eNsmalllast} is arbitrarily small for large $N$. 
		
	We are now able to apply Lemma~\ref{Long time Pertubation}, which gives the existence and uniqueness of a solution $u$ to \eqref{eq:Hartree} (we omit the $N$-dependence of $u$ in the notation). 
	
	Finally we come to the proof of \eqref{eq:H2bound1}.
	By Lemma~\ref{Long time Pertubation} we also know that the solution $u$ to \eqref{eq:Hartree} obeys the following spacetime bound
	\begin{align}
	\int_{0}^{T} \int_{\R^3} |u(t,x)|^{10} \d x \d t < \infty.
	\end{align}
	Using this we can split up $[0,T]$ into finitely many subintervals $I_0, ...,I_K
	$ such that on each $I_j$ 
	\begin{equation}
	\|u\|_{L^{10}_{t,x}(I_j\times\R^3)} \le \delta
	\end{equation}
	for some small $\delta>0$. Now for any multi-index $\alpha$ with $|\alpha| \le 4$ we obtain by using Strichartz estimate on the first interval $I_0$
	\begin{align}
	\label{Laplacnorm}
	\|D^\alpha u\|_{L^{10}_tL^\frac{30}{13}_x(I_0\times\R^3)} &\lesssim \|D^\alpha u_0\|_{L^2(\R^3)} + \left\|\frac{1}{2}\iint|u|^2V_N|u|^2\d y\d zD^\alpha u\right\|_{L^2_tL^\frac{6}{5}_x(I_0\times\R^3)}\nn \\
	&\lesssim \|u_0\|_{H^4(\R^3)} + \|u\|^4_{L^{10}_{t,x}(I_0\times\R^3)}\|D^\alpha u\|_{L^{10}_tL^\frac{30}{13}_x(I_0\times\R^3)} \nn\\
	&\lesssim  \|u_0\|_{H^4(\R^3)} + \delta^4\|D^\alpha u\|_{L^{10}_tL^\frac{30}{13}_x(I_0\times\R^3)}.
	\end{align}
	For $\delta>0$ small enough we obtain
	\begin{align*}
	\|D^\alpha u\|_{L^{10}_tL^\frac{30}{13}_x(I_0\times\R^3)} \lesssim \|u_0\|_{H^4(\R^3)}.
	\end{align*}
	Using this and Strichartz estimate again, we get 
	\begin{align*}
	\|D^\alpha u(t,\cdot)\|_{L^2(\R^3)} \lesssim \|u_0\|_{H^4(\R^3)} + \left\|\frac{1}{2}\iint|u|^2V_N|u|^2\d y\d z D^\alpha u\right\|_{L^2_tL^\frac{6}{5}_x(I_0\times\R^3)} \lesssim \|u_0\|_{H^4(\R^3)}
	\end{align*}
	for any $t\in I_0$ and from this $\|u(t,\cdot)\|_{H^4(\R^3)} \lesssim \|u_0\|_{H^4(\R^3)}$ for any $t\in I_0$. This procedure can now be iterated and we obtain
	\begin{align*}
	\|u(t,\cdot)\|_{H^4(\R^3)} \lesssim \|u_0\|_{H^4(\R^3)}
	\end{align*}
	for all $t\ge0$. 
	
	Thus we have proved the first bound in \eqref{eq:H2bound1}. The second bound follows from the first and the Hartree equation \eqref{eq:Hartree}. Indeed, we have 
	\begin{align*}
	\|\partial_tu(t,\cdot)\|_{L^2(\R^3)} \le \|\Delta u(t,\cdot)\|_{L^2(\R^3)} + \|u(t,\cdot)\|^4_{L^\infty(\R^3)}\|u(t,\cdot)\|_{L^2(\R^3)} \le C_t
	\end{align*}
	and a similar estimate with $-\Delta (\partial_tu(t,\cdot))$. This finishes the proof of  \eqref{eq:H2bound1}.


\bigskip

\noindent
{\bf Convergence to the quintic NLS solution.}  	Now we turn to the proof of \eqref{eq:H-to-NLS}. We compute the derivative of the norm distance using equation \eqref{eq:Hartree} and \eqref{eq:Hartree-NLS} This gives
		\begin{align}
		\frac{\d}{\d t}\|u(t)- \phi(t)\|^2_{L^2(\R^3)} &= 2\Re i\big\langle u(t),\left(\frac{1}{2}\iint|u(t)|^2V_N|u(t)|^2\d y\d z - b_0|\phi(t)|^4\right)\phi(t)\big\rangle\nn \\
		&=2\Re i\big\langle u(t),\left(\frac{1}{2}\iint|u(t)|^2V_N|u(t)|^2\d y\d z - b_0|u(t)|^4\right)\phi(t)\big\rangle +\nn \\
		&\label{derivativeNLSHartree}\quad+2\Re i\big\langle u(t),\left(b_0|u(t)|^4-b_0|\phi(t)|^4\right)\phi(t)\big\rangle.
		\end{align}
		To estimate the first term in \eqref{derivativeNLSHartree} we see that
		\begin{align*}
		&\Big|\frac{1}{2}\iint|u(t,y)|^2V_N(x-y,x-z)|u(t,z)|^2\d y\d z - b_0|u(t,x)|^4\Big|\\
		&=
		\Big|\frac{1}{2}\iint V_N(x-y,x-z)(|u(t,y)|^2|u(t,z)|^2-|u(t,x)|^4)\d y\d z\Big| \\
		&\le \Big|\iint V_N(x-y,x-z)(|u(t,y)|^2-|u(t,x)|^2)|u(t,z)|^2\d y\d z\Big|  \\
		&+\Big|\iint V_N(x-y,x-z)(|u(t,z)|^2-|u(t,x)|^2)|u(t,x)|^2\d y\d z\Big|.
		\end{align*}
		We now proceed with the first term, since both terms can be estimated similarly. Using that $V$ has compact support and hence $V(y,\cdot) =0$ for $|y|>C$ for some $C>0$, this gives
		\begin{align*}
		&\Big|\iint V_N(x-y,x-z)(|u(t,y)|^2-|u(t,x)|^2)|u(t,z)|^2\d y\d z\Big| \\
		&\le\iint_{|y| \le CN^{-\beta}} V_N(y,x-z)\left||u(t,x-y)|^2-|u(t,x)|^2\right||u(t,z)|^2\d y\d z \\
		&=\iint_{|y| \le CN^{-\beta}} V_N(y,x-z)\left|\int_0^1\nabla|u(t,x -sy)|^2\cdot y\d s\right||u(t,z)|^2\d y\d z \\
		&\le CN^{-\beta}\int_0^1\iint_{|y| \le CN^{-\beta}} V_N(y,x-z) |\nabla u(t,x -sy)||u(t,x -sy)||u(t,z)|^2\d y\d z \d s \\
		&\le C_t N^{-\beta}.
		\end{align*}
		In the last inequality we have used Theorem~\ref{thm:Hartree} to bound all factors containing $u$ by $C_t$. From this we obtain
		\begin{align*}
		|\big\langle u(t),\left(\frac{1}{2}\iint|u(t)|^2V_N|u(t)|^2\d y\d z - b_0|u(t)|^4\right)\phi(t)\big\rangle| \le C_tN^{-\beta}.
		\end{align*}
		
		Now for the second term in \eqref{derivativeNLSHartree} we use the elementary inequality\\ $\left||a|^4 - |b|^4\right|\le C|a-b|(|a|^3+|b|^3)$ and obtain
		\begin{align*}
		|\big\langle u(t),\left(|u(t)|^4-|\phi(t)|^4 \right)\phi(t)\big\rangle| &\le \int |u(t,x)||\phi(t,x)|\left||u(t,x)|^4-|\phi(t,x)|^4\right|\d x \\
		&\le C\int |u(t,x)-\phi(t,x)|(|u(t,x)|^5+|\phi(t,x)|^5)\d x \\
		&\le C \|u(t)-\phi(t)\|_{L^2(\R^3)} (\|u(t)\|^5_{L^{10}(\R^3)}+\|\phi(t)\|^5_{L^{10}(\R^3)}) \\ 
		&=C_t  \|u(t)-\phi(t)\|_{L^2(\R^3)}.
		\end{align*}
		Here we have used $\|\phi(t)\|^5_{L^{10}(\R^3)} \le C \|\phi(t)\|^5_{H^2(\R^3)} \le C$, which was proven in \cite[Corollary 1.2]{ColKeelStafTao-08}, and $\|u(t)\|^5_{L^{10}(\R^3)} \le C \|u(t)\|^5_{H^2(\R^3)} \le C_t$, which holds true by Theorem~\ref{thm:Hartree}.
		
		Putting now both estimates together, we obtain
		\begin{align*}
		\frac{\d}{\d t}\|u(t)- \phi(t)\|^2_{L^2(\R^3)} \le C_t\left(N^{-\beta} +  \|u(t)-\phi(t)\|_{L^2(\R^3)}\right),
		\end{align*}
		which completes the proof of \eqref{eq:H-to-NLS}.
	\end{proof}

\section{Bogoliubov equation} \label{sec:Bog-dyn}

In this section we discuss the Bogoliubov dynamics \eqref{eq:Bog},
\begin{equation*}
\left\{
\begin{aligned}
i\partial_t \Phi(t) &= \bH(t) \Phi(t), \\
\Phi(0) &= \Phi_0.
\end{aligned}
\right.
\end{equation*}
Recall that the quadratic generator $\bH(t)$ in \eqref{eq:Bog-Hamiltonian} is built up on the Hartree dynamics $u(t)$ in \eqref{eq:Hartree} with $u_0\in H^4(\R^3)$. All useful properties of  \eqref{eq:Bog} are collected in the following 

\begin{theorem}[Bogoliubov dynamics]\label{thm:Bog-equa}
	Let $\Phi_0$ be a unit vector in $\cF(\{u_0\}^{\bot})$ such that 
$$
\langle\Phi_0,\dGamma(1-\Delta)\Phi_0\rangle \le C. 
$$
Then the Bogoliubov equation \eqref{eq:Bog} has a unique global solution $\Phi(t)$ such that $\Phi(t) \in \cF(\{u(t)\}^\bot)$ for all $t>0$ and 
\begin{align} \label{eq:Bog-dyn-N-bounds-kinetic0}
	\langle\Phi(t),\d\Gamma(1-\Delta)\Phi(t)\rangle \le C_{t,\eps}N^{\beta+\eps}, \quad \forall \eps>0.
	\end{align}
	
\end{theorem}

The key technical result of this section is the following

\begin{lemma}[\textbf{Bounds on Bogoliubov Hamiltonian}] \label{lem:Bogoliubovbounds}
	For every $\eps > 0$ and $\eta > 0$ we have
	\begin{align}
	\label{eq:Bogbound1}\pm \Big( \bH(t) + \d\Gamma(\Delta) \Big) &\le \eta\d\Gamma(1-\Delta) + C_t \cN + C_{t,\eps} \eta^{-1} N^{\beta +\eps}, \\\label{eq:Bogbound2} 
	\pm \partial_t \bH(t) &\le \eta\d\Gamma(1-\Delta) + C_t \cN + C_{t,\eps} \eta^{-1} N^{\beta +\eps}, \\\label{eq:Bogbound3}
	\pm i[\bH(t),\cN] &\le \eta\d\Gamma(1-\Delta) + C_t \cN + C_{t,\eps} \eta^{-1} N^{\beta +\eps}.
	\end{align}
\end{lemma}

To prove Lemma~\ref{lem:Bogoliubovbounds}, we will use a well-known property on the ground state energy of quadratic Hamiltonians, see e.g. \cite{NamNapSol-16,Derezinski-17}. The following result is taken from \cite[Lemma 9]{NamNapSol-16}.

\begin{lemma}[Pairing term estimate] \label{lem:Pairtermbound} Let $H>0$ be a self-adjoint operator on $L^2(\R^3)$ and let $K$ be a Hilbert-Schmidt operator on $L^2(\R^3)$ with symmetric kernel $K(x,y)=K(y,x)$ and satisfying $KH^{-1}K^* \le H$. Then  
	$$ \pm \frac{1}{2} \iint \Big( K(x,y) a_x^* a_y^* + \overline{K(x,y)}a_x a_y \Big) \d x \d y \le \dGamma (H) + \frac{1}{2} \| H_x^{-1/2}K\|_{L^2(\R^6)}^2$$
	as quadratic forms on Fock space.
\end{lemma}

In application of Lemma \ref{lem:Pairtermbound}, the following kernel estimates will be useful. 
\begin{lemma}[Kernel estimate] \label{lem:Hilbertschmidtbound} Let $K_2$ be the operator on $L^2(\R^3)$ with kernel $K_2(t,x,y)$ as in \eqref{eq:Bog-Hamiltonian}. Then we have $\|K_2\|_{op}\le C_t$ and for all $\eps >0$,
	\begin{align}
	\label{eq:Kernelbound}
	\|(1-\Delta)^{-\frac{1}{2}} K_2(t,\cdot,\cdot)\|^2_{L^2(\R^6)} & \le C_{t,\eps} N^{\beta +\eps}, \\ \label{eq:derivativekernelbound}
	\|(1-\Delta)^{-\frac{1}{2}} \partial_tK_2(t,\cdot,\cdot)\|^2_{L^2(\R^6)} &\le C_{t,\eps} N^{\beta +\eps}.
	\end{align}
	
\end{lemma}
\begin{proof}[Proof of Lemma~\ref{lem:Hilbertschmidtbound}] First, we consider the operator bound. For every $f\in L^2(\R^3)$, we denote $\tilde f=Q(t) f$ and use the Cauchy-Schwarz inequality to estimate 
 \begin{align*}
		|\langle f, K_2 f \rangle &= \left| \iiint \overline{\tilde f(x)} {\tilde f(y)} |u(t,z)|^2 u(t,x)u(t,y) V_N(x-y,x-z) \d x \d y \d z\right|\\
		&\le  \|u(t,\cdot)\|_{L^\infty(\R^3)}^4 \iiint \frac{|\tilde f(x)|^2+|\tilde f(y)|^2}{2} V_N(x-y,x-z) \d x \d y \d z \\
		& = \|u(t,\cdot)\|_{L^\infty(\R^3)}^4 \|\tilde f\|_{L^2(\R^3)}^2 \|V_N\|_{L^1(\R^6)} \le C_t \|f\|_{L^2(\R^3)}^2.
		\end{align*}
		Therefore, 
		\begin{align*}
		\|K_2\|_{op} \le C_t.
		\end{align*}
		
		Next, to prove \eqref{eq:Kernelbound} we use an interpolation argument as in \cite{GriMac-13,NamNap-15,NamNap-17}.  By definition we know that 
	$$K_2= \Big(Q(t)\otimes Q(t)\Big)\tilde K_2$$
	 and hence 
	 $$K_2 - \tilde K_2 = \Big((Q(t)-1)\otimes 1\Big)\tilde K_2 + \Big(Q(t) \otimes (Q(t)-1)\Big)\tilde K_2.$$
	We now want to prove a $L^2$ bound on the first term, the second term will follow similarly. We see that
	\begin{align*}
	&\left\|(Q(t)-1)\otimes 1\tilde K_2(t,\cdot,\cdot)\right\|^2_{L^2(\R^6)} = \left\||u(t)\rangle\langle u(t)|\otimes 1\tilde K_2(t,\cdot,\cdot)\right\|^2_{L^2(\R^6)} \\
	&=\iint \left|\int \overline{u(t,\tilde{x})}\tilde{K_2}(t,\tilde{x},y)\d\tilde{x}\right|^2|u(t,x)|^2\d x \d y \\
	&= \iint \left|\iint\overline{u(t,\tilde{x})}|u(t,z)|^2u(t,\tilde{x})	V_N(\tilde{x}-y,\tilde{x}-z)u(t,y)\d z\d\tilde{x}\right|^2|u(t,x)|^2 \d x \d y \\
	&\le \|u(t,\cdot)\|^8_{L^\infty(\R^3)}\|V\|^2_{L^1(\R^6)}\|u(t,\cdot)\|^4_{L^2(\R^3)} \le C_t,
	\end{align*}
	where we have used Theorem~\ref{thm:Hartree} in the last inequality. This gives
	\begin{equation*}
	\|K_2(t,\cdot,\cdot) - \tilde K_2(t,\cdot,\cdot)\|^2_{L^2(\R^6)} \le C_t.
	\end{equation*}
	Since $(1-\Delta_x)^{-\frac{1}{2}} \le 1 $ on $L^2$ we see that
	\begin{equation}
	\label{eq:projectionbound}
	\|(1-\Delta_x)^{-\frac{1}{2}}K_2(t,\cdot,\cdot) - (1-\Delta_x)^{-\frac{1}{2}}\tilde K_2(t,\cdot,\cdot)\|^2_{L^2(\R^6)} \le C_t.
	\end{equation}
	Hence, we only need to prove \eqref{eq:Kernelbound} with $K_2$ replace by $\tilde K_2$ to get the desired result. We will prove 
	\begin{align}
	&\|(1-\Delta)^{-\frac{3}{4}-\eps}\tilde K_2(t,\cdot,\cdot)\|^2_{L^2(\R^6)}\le C_{t,\eps}\label{eq:K2bound1}\\
	&\|\tilde K_2(t,\cdot,\cdot)\|_{L^2(\R^6)}^2 \le C_t N^{3\beta}\label{eq:K2bound2}
	\end{align}
	for any $\eps >0$ and then use interpolation.
	To prove \eqref{eq:K2bound1} we first calculate the Fourier transform of $\tilde K_2$:
	\begin{align*}
	\widehat{\tilde K_2}(t,p,q) &= \iint e^{-i2\pi(p\cdot x+q\cdot y)} \int |u(t,z)|^2 V_N(x-y,x-z)\d z \, u(t,x) u(t,y)\d x \d y \\
	&= \iiint e^{-i2\pi(p\cdot x+q\cdot y)} |u(t,x-z)|^2 V_N(x-y,z)u(t,x) u(t,y) \d x \d y \d z \\
	&= \iiint e^{-i2\pi(p\cdot \tilde{x}+(p+q)\cdot y)} |u(t,\tilde{x}+y-z)|^2 V_N(\tilde{x},z)u(t,\tilde{x}+y) u(t,y) \d\tilde{x}\d y \d z \\
	&=\iint e^{-i2\pi p\cdot \tilde{x}} V_N(\tilde{x},z) \widehat{(|u_{\tilde{x}-z}|^2u_{\tilde{x}}u)}(p+q) \d\tilde{x}\d z.
	\end{align*}
	Here we have defined the short-hand notation for the translation $u_{\tilde x}(\cdot) = u(\tilde x +\cdot)$.
	Therefore, by the Cauchy-Schwarz inequality
	\begin{align*}
	&|\widehat{\tilde K_2}(t,p,q)|^2 \le \left(\iint |V_N(\tilde{x},z)| |\widehat{(|u_{\tilde{x}-z}|^2u_{\tilde{x}}u)}(p+q)| \d\tilde{x}\d z\right)^2 \\ &\le \|V\|_{L^1(\R^6)}\iint |V_N(\tilde{x},z)|\, |\widehat{(|u_{\tilde{x}-z}|^2u_{\tilde{x}}u)}(p+q)|^2 \d \tilde{x}\d z.
	\end{align*}
	Moreover, using Theorem~\ref{thm:Hartree} and Plancherel's we have that
	\begin{align*}
	\int |\widehat{(|u_{\tilde{x}-z}|^2u_{\tilde{x}}u)}(p+q)|^2 \d q \le
	\|u(t,\cdot)\|^6_{L^\infty(\R^3)} \int |u(t,x)|^2\d x  = C_t.
	\end{align*}
	Hence, we see that for all $\eps>0$
	\begin{align*}
	&\|(1-\Delta)^{-\frac{3}{4}-\eps}\tilde K_2\|^2_{L^2(\R^6)} = \iint (1+|p|^2)^{-\frac{3}{2} - 2\eps} |\widehat{\tilde K_2}(t,p,q)|^2\d p\d q  \\
	&\le
	\|V\|_{L^1(\R^6)} \iiiint (1+|p|^2)^{-\frac{3}{2} - 2\eps} V_N(\tilde{x},z) |\widehat{(|u_{\tilde{x}-z}|^2u_{\tilde{x}}u)}(p+q)|^2 \d p\d q \d\tilde{x}\d z \\
	&\le C_t\|V\|_{L^1(\R^6)}^2 \int (1+|p|^2)^{-\frac{3}{2} - 2\eps} \d p = C_{t,\eps},
	\end{align*}
	where we have used $\int (1+|p|^2)^{-\frac{3}{2} - 2\eps} \d p\le C_\eps$ in three dimensions.
	To prove \eqref{eq:K2bound2} we calculate 
	\begin{align*}
	&\|\tilde K_2\|_{L^2(\R^6)}^2 = \iint |u(t,x)|^2|u(t,y)|^2 \left|\int |u(t,z)|^2 V_N(x-y,x-z) \d z\right|^2\d x \d y \\ 
	&\le \|u(t,\cdot)\|^6_{L^\infty(\R^3)} \iint |u(t,x)|^2 \left|\int N^{6\beta}V(N^\beta(x-y),N^\beta (x-z)\d z\right|^2 \d x \d y  \\
	&= \|u(t,\cdot)\|_{L^\infty(\R^3)}^6 \|u(t,\cdot)\|_{L^2(\R^3)} \int\left|\int N^{3\beta} V(N^\beta y,z)\d z\right|^2\d y \\
	&= C_t N^{3\beta} \int\left|\int V(y,z)\d z\right|^2\d y = C_t N^{3\beta},
	\end{align*}
	where we have used that $V$ has compact support.
	Therefore, by interpolation
	\begin{align*}
	\|(1-\Delta)^{-\frac{1}{2}} \tilde K_2(t,\cdot,\cdot)\|^2_{L^2(\R^6)} \le  C_{t,\eps} N^{\beta+\eps},
	\end{align*}
	which proves \eqref{eq:Kernelbound}. 
	
	For the proof of \eqref{eq:derivativekernelbound} we see that
	\begin{align*}
	\partial_tK_2(t) = \partial_tQ(t) \otimes Q(t)\tilde K_2(t) + Q(t) \otimes \partial_tQ(t)\tilde K_2(t) + Q(t) \otimes Q(t)\partial_t \tilde K_2(t).
	\end{align*}
	Similarly to the derivation of \eqref{eq:projectionbound} one can prove bounds for each term and obtain
		\begin{equation*}
	\|(1-\Delta_x)^{-\frac{1}{2}}\partial_tK_2(t,\cdot,\cdot) - (1-\Delta_x)^{-\frac{1}{2}} \partial_t\tilde K_2(t,\cdot,\cdot)\|^2_{L^2(\R^6)} \le C_t.
	\end{equation*}
	Therefore, we again only need to prove \eqref{eq:derivativekernelbound} with $\partial_tK_2$ replaced by $\partial_t\tilde K_2$. This works similar to the derivation of \eqref{eq:Kernelbound} and we omit the details. This finishes the proof.
\end{proof}

Now we are able give 
	\begin{proof}[Proof of Lemma ~\ref{lem:Bogoliubovbounds}] Consider 
	$$ \bH(t) + \d\Gamma(\Delta) = \dGamma(h(t)+\Delta+K_1)+\frac{1}{2}\iint \Big(K_2(x,y) a_x^* a_y^* + \overline{K_2(x,y)} \Big) a_x a_y \d x \d y.$$
	
	    By definition of $h(t)$ we have that
		\begin{align*}
		h(t) + \Delta = \iint |u(t,y)|^2 V_N(x-y,x-z)|u(t,z)|^2 \d y \d z 
		\end{align*}
		which is a multiplication operator. By Theorem~\ref{thm:Hartree} the corresponding function can be bounded by
		\begin{align} \label{convbound}
		\left|\iint |u(t,y)|^2 V_N(x-y,x-z)|u(t,z)|^2 \d y \d z\right| \le \|u(t,\cdot)\|^4_{L^\infty(\R^3)} \|V_N\|_{L^1(\R^6)} \le C_t.
		\end{align}
		
		Moreover,  the operator $K_1$ in \eqref{eq:Bog-Hamiltonian} satisfies 
		\begin{equation} \label{eq:K1 bound}
		\| K_1\|_{op} \le C_t
		\end{equation} 
		(this can be proved similarly to the bound $\| K_2\|_{op} \le C_t$ in Lemma \ref{lem:Hilbertschmidtbound}). 
		Thus we have proved that $\pm (h(t)+\Delta +K_1) \le C_t$, and hence
		\begin{align} \label{eq:oneparticlebound}
		\pm\d\Gamma(h+\Delta +K_1) \le C_t \cN.
		\end{align}
		
		Next, by applying the paring term estimate in Lemma \ref{lem:Pairtermbound} with 
		$$H= \eta (1-\Delta) + \|K_2\|_{op}, \quad \eta>0,$$ 
		and the kernel estimate in Lemma \ref{lem:Hilbertschmidtbound} we find that 
		\begin{align}
		&\pm\frac{1}{2}\iint \Big( K_2(t,x,y) a_x^* a_y^* + \overline{K_2(t,x,y)}a_x a_y \Big)  \nn\\
		&\le   \eta \d\Gamma(1-\Delta) + \|K_2\|_{op}\, \cN + \eta^{-1} \|(1-\Delta_x)^{-1/2} K_2\|^2_{L^2(\R^6)}\nn \\
		&\le \eta \d\Gamma(1-\Delta) + C_t \cN + C_{t,\eps} \eta^{-1} N^{\beta +\eps}. \label{eq:Pairbound}
		\end{align}
		Combining this with \eqref{eq:oneparticlebound} we conclude
		\begin{align*}
		\pm \Big( \bH(t) + \d\Gamma(\Delta) \Big) \le \eta\d\Gamma(1-\Delta) + C_t \cN + C_{t,\eps} \eta^{-1} N^{\beta +\eps} 
		\end{align*}
		which finishes the proof of \eqref{eq:Bogbound1}. The bound on $\partial_t\bH(t)$ can be proven similarly. Moreover, we see that \begin{align*}
		i[\bH(t),\cN] = - \iint \Big( iK_2(t,x,y) a_x^* a_y^* + \overline{iK_2(t,x,y)}a_x a_y \Big) \d x \d y
		\end{align*}
		and \eqref{eq:Bogbound3} also follows from  the same argument. This ends the proof.
	\end{proof}

Finally we conclude

\begin{proof}[Proof of Theorem~\ref{thm:Bog-equa}] Using the bound in Lemma \ref{lem:Bogoliubovbounds}, the existence and uniqueness of the solution $\Phi(t)\in \cF(\{u(t)\}^{\bot})$ to the Bogoliubov equation \eqref{eq:Bog} follow from the abstract results in \cite[Theorems 7, 8]{LewNamSch-15}. 

It remains to prove the kinetic bound \eqref{eq:Bog-dyn-N-bounds-kinetic0}. By Lemma \ref{lem:Bogoliubovbounds} we have 
		\begin{align*}
		A(t) = \bH(t) + C_{t,\eps} (\cN + N^{\beta+\eps})  \ge  \frac{1}{2}\dGamma(1-\Delta),
		\end{align*}
		if we choose $C_{t,\eps}$ large enough. Using the equation for $\Phi(t)$, we see that
		\begin{align*}
		\frac{d}{dt} \big\langle \Phi(t), A(t) \Phi(t)  \big\rangle &= \big\langle \Phi(t), \partial_t A(t)  \Phi(t)  \big\rangle + \big\langle \Phi(t), i[\bH(t),A(t)]  \Phi(t)  \big\rangle \\
		&= \big\langle \Phi(t), \partial_t (\bH(t)+\partial_t C_{t,\eps}(\cN + N^{\beta+\eps}) ) \Phi(t)  \big\rangle + \big\langle \Phi(t), i[\bH(t),\cN]  \Phi(t)  \big\rangle\\
		&\le C_{t,\eps} \big\langle \Phi(t), A(t)  \Phi(t) \rangle.
		\end{align*}
		Thus, using Gronwall's inequality we get
		\begin{equation*}
		\big\langle \Phi(t), A(t) \Phi(t)  \big\rangle \le e^{C{t,\eps}}\big\langle \Phi(0), A(0) \Phi(0)  \big\rangle.
		\end{equation*}
		Since 
		\begin{equation*}
		A(0)\le C_\eps( \dGamma(1-\Delta) + N^{\beta+\eps}),
		\end{equation*}
		we obtain that
		\begin{align*}
		\big\langle \Phi(t), \dGamma(1-\Delta) \Phi(t)  \big\rangle &\le2\big\langle \Phi(t), A(t) \Phi(t)  \big\rangle \le  C_\eps e^{C{t,\eps}} \Big( \big\langle \Phi(0), \dGamma(1-\Delta) \Phi(0)  \big\rangle \Big) + N^{\beta+\eps} \Big).
		\end{align*}
		This finishes the proof of the kinetic estimate \eqref{eq:Bog-dyn-N-bounds-kinetic0}.		
	\end{proof}

\section{Transformation of the many-body dynamics}\label{sec:gen-stra}

Our general strategy to derive effective equations from the many-body Schr\"odinger equation \eqref{eq:MBSch} is similar to that in the pair-interaction case \cite{LewNamSch-15,NamNap-15,NamNap-17,BreNamNapSch-17}. Let $\{u(t)\}$ be the Hartree dynamics and recall from \cite[Section 2.3]{LewNamSerSol-15} the following operator
\begin{equation}
U_N(t)= \bigoplus_{k=0}^{N}Q(t)^{\otimes k} \frac{(a(u(t)))^{N-k}}{\sqrt{(N-k)!}}, \quad Q(t) = 1- |u(t)\rangle\langle u(t)|.
\end{equation}
It is a unitary operator from $L^2_s((\R^3)^N)$ to the truncated Fock space
$$\cF^{\le N}_+(t) = \1^{\le N} \cF_+(t) = \bigoplus_{n=0}^{N}  \Big( \{u(t)\}^\bot \Big)^{\otimes_s n},\quad \1^{\le N}= \1(\cN\le N)$$
with the inverse
\begin{equation}
U_N(t)^*= \bigoplus_{k=0}^{N} \frac{(a^*(u(t)))^{N-k}}{\sqrt{(N-k)!}}.
\end{equation}
Of course we can extend $U_N(t)^*$ to the whole Fock space $\cF_+(t)$ by setting value $0$ outside the truncated space $\cF_+^{\le N}(t)$ (in this way $U_N(t)$ is a partial unitary operator from $L^2_s((\R^3)^N)$ to $\cF_+(t)$).

As explained in \cite{LewNamSerSol-15}, $U_N(t)$ provides a rigorous implementation of the c-number substitution in Bogoliubov's heuristic argument \cite{Bogoliubov-47}, via the actions 
\begin{align} \label{eq:UN-action}
&U_N(t)a^*(u(t))a(u(t))U_N(t)^* = N -\cN, \nn\\
&U_N(t)a^*(f)a(u(t))U_N(t)^* = a^*(f)\sqrt{N-\cN}, \nn\\
&U_N(t)a^*(u(t))a(f)U_N(t)^* = \sqrt{N-\cN}a(f), \nn\\
&U_N(t)a^*(f)a(g)U_N(t)^* = a^*(f)a(g)
\end{align}
where $f, g \in \{u(t)\}^\bot$. When $\cN\ll N$ with $\cN$ the particle number operator, the quantity $\sqrt{N-\cN}$ is close to the scalar value $\sqrt{N}$, leading to a quantitative justification of Bogoliubov's approximation in the sector of few particles. 

The unitary operator $U_N(t)$ allows us to transform the Sch\"odinger equation \eqref{eq:MBSch} to an equation in Fock space. Recall the phase factor in \eqref{eq:chi}:
$$
\chi(t)=  \frac{2N+3}{6} \iiint V_N(x-y,x-z)|u(t,x)|^2|u(t,y)|^2|u(t,z)|^2 \d x\d y\d z.
$$

\begin{lemma}[Transformed many-body dynamics]\label{lem:eq-PhiN} Let $\Psi_N(t)$ be the solution of \eqref{eq:MBSch}. Then
\begin{equation} \label{eq:def-PhiN}
\Phi_N(t) := e^{-i\int_0^t \chi(s) \d s}\,U_N(t)\Psi_N(t)  \in \cF_+^{\le N}(t)
\end{equation}
solves 
\begin{align} \label{eq:eq-PhiN}
\boxed{ i\partial_t\Phi_N(t) =  \widetilde H_N(t) \Phi_N(t) = \left( \bH(t) +  \frac{1}{2}\sum_{j=0}^{6}(R_j+R_j^*)  \right) \Phi_N(t)} 
\end{align}
where
\begin{align*}
R_0 &= \frac{1}{6}\Big \langle u(t)^{\otimes 3}, V_N u(t)^{\otimes 3} \Big\rangle\left( \frac{3\cN^2+6\cN+2}{N}-\frac{\cN(\cN+1)(\cN+2)}{N^2}\right) \\
&\quad+\d\Gamma\left(Q(t)\left(\frac{1}{2}\iint|u(t,y)|^2V_N|u(t,z)|^2\d y \d z + K_1\right) Q(t)\right)\left(\frac{(N-\cN)(N-\cN-1)}{N^2} -1\right) \\
R_1 &=  \left( \frac{(N-\cN)(N-\cN -1)}{N^2}-1\right) \sqrt{N- \cN} \, a\left(Q(t)\iint|u(t,y)|^2V_N(\cdot, y,z))|u(t,z)|^2\d y \d z\,u(t,\cdot)\right), \\ 
R_2 &= \iint K_2(t,x,y) a_x^* a_y^*  \left( \frac{\sqrt{N-\cN -1} \sqrt{N - \cN}(N -\cN -2)}{N^2} - 1\right), \\ 
R_3 &= \frac{1}{3N^2} \idotsint(Q(t)\otimes Q(t)\otimes Q(t)V_N 1\otimes 1\otimes 1)(x,y,z;x',y',z')u(t,x')u(t,y')u(t,z') \\ 
& \quad\quad\quad\times a^*_x a^*_y a^*_z\d x \d y \d z\d x '\d y'\d z' \sqrt{N-\cN-2}\sqrt{N-\cN-1}\sqrt{N-\cN}  \\
&\quad+\frac{2}{N^2} \iiiint (Q(t)\otimes Q(t)\int|u(t,z)|^2V_N\d zQ(t)\otimes 1)(x,y;x',y') u(t,y') \times \\
&\quad\quad \quad \times a^*_xa^*_ya_{x'}\d x dy\d x '\d y'\sqrt{N-\cN}(N-\cN-1) + \\
&\quad+\frac{1}{N^2} \idotsint(Q(t)\otimes Q(t)\otimes 1 V_N 1\otimes 1 \otimes Q(t))(x,y,z;x',y',z') \overline{u(t,z)}u(t,x')u(t,z') \times \\
&\quad\quad\quad\times a^*_xa^*_ya_{z'}\d x \d y \d z\d x '\d y'\d z'\sqrt{N -\cN}(N-\cN-1)  \\
R_4&= \frac{1}{N^2}\idotsint(Q(t)\otimes Q(t)\otimes Q(t)V_NQ(t)\otimes 1\otimes 1)(x,y,z;x',y',z')u(t,y')u(t,z')\times \\
&\quad\quad\quad \times a^*_xa^*_ya^*_za_{x'}\sqrt{N-\cN -1}\sqrt{N-\cN}\d x \d y \d z\d x '\d y'\d z'  + \\
&\quad+ \frac{1}{2N^2} \iiiint (Q(t)\otimes Q(t) \int |u(t,z)|^2 V_N\d zQ(t)\otimes Q(t))(x,y;x',y') \times\\
&\quad \quad \quad a^*_xa^*_ya_{x'}a_{y'}\d x dy\d x '\d y'(N-\cN)\\
&\quad+\frac{1}{N^2}\idotsint (Q(t)\otimes Q(t)\otimes 1 V_N Q(t)\otimes 1\otimes Q(t))(x,y,z;x'y',z')\overline{u(t,z)}u(t,y')\times\\
&\quad\quad\quad\quad\quad\quad\quad\quad \times a^*_xa^*_ya_{x'}a_{z'}(N-\cN)\d x \d y \d z\d x '\d y'\d z'\\
R_5& = \frac{1}{N^2}\idotsint(Q(t)\otimes Q(t)\otimes Q(t)V_NQ(t)\otimes Q(t)\otimes 1)(x,y,z;x',y',z')u(t,z')\times \\
&\quad\quad\quad\quad\quad\quad\quad\quad \times a^*_xa^*_ya^*_za_{x'}a_{y'}\sqrt{N-\cN -2}\d x \d y \d z\d x '\d y'\d z'  \\
R_6 &= \frac{1}{6N^2}\idotsint(Q(t)\otimes Q(t)\otimes Q(t)V_NQ(t)\otimes Q(t)\otimes Q(t)(x,y,z;x',y',z')\times \\
&\quad\quad\quad\quad\quad\quad\quad\quad \times a^*_xa^*_ya^*_za_{x'}a_{y'}a_{z'}\d x \d y \d z\d x '\d y'\d z'.
\end{align*}
\end{lemma}

Note that the operator 
$$ \widetilde H_N(t)=\bH(t) +  \frac{1}{2}\sum_{j=0}^{6}(R_j+R_j^*)$$
 in \eqref{eq:eq-PhiN} really leaves invariant the truncated Fock space $ \cF_+^{\le N}(t)$.  For example, the sum of $R_2$ and the corresponding (creation) pairing term in the Bogoliubov Hamiltonian $\bH(t)$ is
$$
\iint K_2(t,x,y) a_x^* a_y^* \frac{\sqrt{N-\cN-1}(N-\cN)(N-\cN-2)}{N^2} 
$$
which never creates $N+2$ particles since $\sqrt{N-\cN-1}(N-\cN)(N-\cN-2)$ is zero on sectors $\cN \in \{N,N-1,N-2\}$.

\begin{proof} From the definition \eqref{eq:def-PhiN} and the Schr\"odinger equation \eqref{eq:MBSch} we find that 
\begin{align} \label{eq:eq-PhiN-00}
 i\partial_t\Phi_N(t) =  \Big( \chi(t) +  i\dot{U}_N(t) U^*_N(t)+ U_N(t)H_NU^*_N(t) \Big)  \Phi_N(t) .
\end{align}
The time-derivative of $U_N(t)$ has been computed in \cite[Lemma 6]{LewNamSch-15}
	\begin{align} 
	\label{eq:derivativunitary}
	i\dot{U}_N(t)U_N^*(t)&= a^*(u(t))a(Q(t)i \dot u(t))-\sqrt{N-\cN}\,a(Q(t)i\dot u(t))\nn\\
	&\quad-a^*(Q(t)i\dot u(t))\sqrt{N-\cN}-\langle i\dot u(t),u(t)\rangle(N-\cN))\nn\\
	&= a^*(u(t))a(Q(t)h(t) u(t))-\sqrt{N-\cN}\,a(Q(t) h(t)  u(t))\nn\\
	&\quad-a^*(Q(t)h(t) u(t))\sqrt{N-\cN}-\langle u(t), h(t) u(t)\rangle(N-\cN).
	\end{align}
	Here in the last identity we have used the Hartree equation $i\dot u(t) =h(t) u(t)$. 
	
	Then we compute the conjugation $U_N(t)H_NU^*_N(t)$ using \eqref{eq:HN-2nd-quan} and \eqref{eq:UN-action}:
	\begin{align}  \label{eq:UNHNUN}
	&U_N(t) H_N U_N^*(t) = \frac{1}{6N^2} \langle u(t)^{\otimes 3}, V_N u(t)^{\otimes 3}\rangle (N-\cN-1)(N-\cN-2)(N-\cN-3)\nn\\
	&\quad + \|\nabla u(t,\cdot)\|^2_{L^2} (N-\cN) +  \sqrt{N-\cN}a(Q(t)h(t) u(t)) + a^*(Q(t)h(t) u(t)) \sqrt{N-\cN} \nn\\
	&\quad +  \dGamma(Q(t) h(t) Q(t) + K_1) +  \frac{1}{2} \iint \Big( K_2(t,x,y) a_x^* a_y^* + \overline{K_2(t,x,y)}a_x a_y  \d x \d y \Big) \nn\\
	&\quad +  \frac{1}{2}\sum_{j=1}^6 (R_j+R_j^*).
	\end{align}
	It remains to sum \eqref{eq:derivativunitary} and \eqref{eq:UNHNUN}, then use the identity 
$$
a^*(u(t))a(Q(t)h(t)u(t)) + \dGamma(Q(t)h(t)Q(t)) = \dGamma(h(t)) - \dGamma(h(t) P(t))  = \dGamma(h(t))
$$
on the excited Fock space $\cF_+(t)$. All this leads to 
$$
i\dot{U}_N(t) U^*_N(t)+ U_N(t)H_NU^*_N(t) = \bH(t) + \frac{1}{2}\sum_{j=0}^6 (R_j+R_j^*) -\frac{2N+3}{6} \langle u(t)^{\otimes 3}, V_N u(t)^{\otimes 3}\rangle.
$$
The choice of the phase factor $\chi(t)$ in \eqref{eq:chi},
$$
\chi(t)=  \frac{2N+3}{6} \langle u(t)^{\otimes 3}, V_N u(t)^{\otimes 3}\rangle,
$$
exactly implies the equation \eqref{eq:eq-PhiN} from \eqref{eq:eq-PhiN-00}. 
\end{proof}

Formally, in the limit $N\to \infty$ all the error terms $R_j's$ are negligible and the transformed equation \eqref{eq:eq-PhiN} can be approximated by the Bogliubov equation \eqref{eq:Bog}. The task of estimating all $R_j's$ will be carried out in the next section.

\section{Operator bounds on Fock space} \label{sec:op-Fock} Now we collect various useful bounds on the the error terms $R_j$'s in \eqref{eq:eq-PhiN}. We will prove

\begin{lemma}[\textbf{Error bounds on truncated Fock spaces}] \label{lem:errorbounds} For $0\le m \le N$, we have the following quadratic form bounds on the truncated Fock space $\cF_+^{\le m}(t)$:
\begin{align}
\label{eq:errorR0-5}
\pm(R_j+R_j^*) &\le  \eta  \dGamma(1-\Delta) + \eta^{-1} C_t mN^{4\beta-1}, \\
	\pm i[R_j + R_j^*,\cN] & \le  \eta  \dGamma(1-\Delta)+ \eta^{-1} C_t  m N^{4\beta-1}   ,  \label{eq:errordRcom} \\
	\pm\partial_t( R_j + R_j^*) &\le  \eta  \dGamma(1-\Delta) + \eta^{-1} C_t m N^{4\beta-1} ,   \label{eq:errorRdt}
\end{align}
for all $j=0,1,2,3,4,5,6$ and for all 
$$\eta \ge C_t \max\{ \sqrt{m N^{2\beta-1}},  \sqrt{m^3 N^{5\beta-3}}, m^2 N^{4\beta-2} \}.$$
\end{lemma}

We will use the following well-known Sobolev type estimate (see \cite[Lemma 3.2]{NamRouSei-16}). 
\begin{lemma}[Kinetic bound for translation-invariant potentials] \label{lem:3body ope}
	For $0\le W\in L^{3/2}(\R^3)$, the multiplication operator $W(x-y)$ satisfies 
	\begin{align}
	0\le W(x-y) \le \|W\|_{L^{3/2}(\R^3)} (-\Delta_x)
	\end{align}
	as quadratic forms on $L^2(\R^3\times \R^3)$.
\end{lemma}
%

%
%
%

We will also need the following kernel estimate (c.f. Lemma \ref{lem:Hilbertschmidtbound}). 

\begin{lemma}[Kernel estimate]	\label{lem:R3kernelbound}
	For every $z\in\R^3$ fixed, let $k_z$ be the operator on $L^2(\R^3)$ with kernel 
	$$k_z(x,y) = u(t,x)u(t,y)V_N(x-y,x-z).$$
	Then we have
	\begin{equation}
	\|(1-\Delta_x)^{-1/2}k_z\|_{L^2(\R^6)}^2 \le C_t N^{4\beta}.
	\end{equation}	
\end{lemma}

\begin{proof}[Proof of Lemma~\ref{lem:R3kernelbound}.] 

	First we calculate the Fourier transform
	\begin{align*}
	\widehat k_z (p,q) &= \iint e^{-i2\pi(p\cdot x + q\cdot y)}u(t,x)u(t,y)V_N(x-y,x-z)\d x\d y \\
	&=e^{-i2\pi p\cdot z}\iint e^{-i2\pi N^{-\beta}(p\cdot x+q\cdot y)} u(t,N^{-\beta}x+z)u(t,N^{-\beta}y)V(x-y,x)\d x\d y \\
	&= e^{-i2\pi p\cdot z}\hat f(N^{-\beta}p,N^{-\beta}q),
	\end{align*}
	where we have defined $f(x,y) = u(t,N^{-\beta}x+z)u(t,N^{-\beta}y)V(x-y,x)$. From this and Hardy's inequality we get
	\begin{align*}
	\|(1-\Delta_x)^{-1/2}k_z\|_{L^2}^2 &= \iint (1+|p|^2)^{-1}|\hat f(N^{-\beta}p,N^{-\beta}q)|^2\d p\d q\\
	&= N^{4\beta}\iint(N^{-2\beta}+|p|^2)^{-1} |\hat f(p,q)|^2 \d p\d q \\
	&\le N^{4\beta} \iint \frac{|\hat f(p,q)|^2}{|p|^2}\d p\d q \le 4N^{4\beta}\iint |\nabla_p\hat f(p,q)|^2\d p\d q \\
	&=4 N^{4\beta} \iint |x|^2 |u(t,N^{-\beta}x)|^2|u(t,N^{-\beta}y)|^2|V(x-y,x)|\d x\d y  \\
	&\le C_t N^{4\beta}\iint |x|^2 V(x-y,x)\d x\d y \le C_t N^{4\beta}.
	\end{align*}
	In the last step we have used that $V$ has compact support. This finishes the proof.	
\end{proof}

\begin{proof}[Proof of Lemma \ref{lem:errorbounds}] {\bf Proof of \eqref{eq:errorR0-5}.} We proceed term by term. 

\bigskip

\noindent	$\boxed{j=0}$ Let us consider
\begin{align*}
R_0 &= \frac{1}{6}\Big \langle u(t)^{\otimes 3}, V_N u(t)^{\otimes 3} \Big\rangle\left( \frac{3\cN^2+6\cN+2}{N}-\frac{\cN(\cN+1)(\cN+2)}{N^2}\right), \\
&\quad+\d\Gamma\left(Q(t)\left(\frac{1}{2}\iint|u(t,y)|^2V_N|u(t,z)|^2\d y \d z + K_1\right) Q(t)\right)\left(\frac{(N-\cN)(N-\cN-1)}{N^2} -1\right) \\
&=: R_{0,1}+ R_{0,2}.
\end{align*}

For the first term $R_{0,1}$, using Theorem~\ref{thm:Hartree} we know that
	\begin{equation*}
	\iiint |u(t,x)|^2|u(t,y)|^2|u(t,z)|^2V_N(x-y,x-z)\d x\d y\d z \le C_t.
	\end{equation*}
	Moreover, 
	$$
	\left| \frac{3\cN^2+6\cN+2}{N}-\frac{\cN(\cN+1)(\cN+2)}{N^2}  \right| \le C\frac{(\cN+1)^2}{N}
	$$
	With that we obtain
	\begin{align*}
	\pm R_{0,1} \le C_t\frac{(\cN+1)^2}{N}.
	\end{align*}
	
	For the second term $R_{0,2}$, from the one particle operator bounds \eqref{convbound} and \eqref{eq:K1 bound} we get
	\begin{align*}
	\pm\d\Gamma\left(Q(t)\left(\frac{1}{2}\iint|u(t)|^2V_N|u(t)|^2\d y \d z + K_1\right) Q(t)\right) \le C_t \cN. 
	\end{align*}
	Moreover, we have the simple bound on $\cF^{\le N}_+(t)$ 
	\begin{align*}
	\left| \frac{(N-\cN)^2 - (N-\cN)}{N^2} - 1\right| \le C\frac{(\cN+1)}{N}.
	\end{align*}
	Putting both together (the relevant operators commute) we obtain
	\begin{align*}
	\pm R_{0,2}  \le C_t \frac{(\cN+1)^2}{N}.
	\end{align*}
	Thus in summary, as quadratic forms on $\cF^{\le m}_+(t)$, 
	$$
	\pm R_0 \le C_t \frac{(\cN+1)^2}{N} \le C_t \frac{m}{N} (\cN+1).$$
	Note that $\cN\le \dGamma(1-\Delta)$.
	
	\bigskip
	
	\noindent 
	$\boxed{j=1}$ We consider
	$$
	R_1=\left( \frac{(N-\cN)(N-\cN -1)}{N^2}-1\right) \sqrt{N- \cN} \, a\left(Q(t)\iint|u(t)|^2V_N|u(t)|^2\d y \d z\,u(t)\right).
	$$
	For any $\Phi \in \cF^{\le m}_+(t)$, by the Cauchy-Schwarz inequality we see that
	\begin{align*}
	|\langle\Phi,R_1\Phi\rangle| &\le 2\left\| \left( \frac{(N-\cN)(N-\cN -1)}{N^2}-1\right) \sqrt{N- \cN} \Phi \right\| \times \\
	&\quad \times \left\|a\left(Q(t)\frac{1}{2}\iint|u(t)|^2V_N|u(t)|^2\d y \d z\,u(t)\right)\Phi \right\|.
	\end{align*}
	For the second term we use the obvious inequality $a^*(v)a(v) \le \|v\|^2 \cN$ combined with
	\begin{align*}
	&\left\|Q(t)\frac{1}{2}\iint|u(t)|^2V_N|u(t)|^2\d y \d z\,u(t)\right\|^2_{L^2(\R^3)} \\
	&\le \left\|\frac{1}{2}\iint|u(t)|^2V_N|u(t)|^2\d y \d z\,u(t)\right\|^2_{L^2(\R^3)} \\
	&\le \|u(t,\cdot)\|^8_{L^\infty(\R^3)}\|V\|_{L^1(\R^3)} \|u(t,\cdot)\|^2_{L^2(\R^3)} = C_t.
	\end{align*}
	For the first term, we use the simple bound
 	\begin{align*}
 	\left| \left( \frac{(N-\cN)(N-\cN -1)}{N^2}-1\right) \sqrt{N- \cN} \right| \le C\frac{\cN + 1}{\sqrt{N}}.
 	\end{align*}
 	Putting together, and using $\Phi\in \cF^{\le m}_+(t)$ we obtain
	\begin{align*}
		|\langle\Phi,R_1\Phi\rangle| \le C_t \Big\langle \Phi, \cN\Phi \Big\rangle^{1/2} \Big\langle \Phi, \frac{(\cN+1)^2}{N}\Phi \Big\rangle^{1/2} \le C_t \sqrt{\frac{m}{N}} \langle \Phi, (\cN+1)\Phi\rangle. 
	\end{align*}
	Thus we have the quadratic form estimate on $\cF^{\le m}_+(t)$:
	$$
	\pm (R_1+R_1)\le C_t \sqrt{\frac{m}{N}}  (\cN+1) .
	$$
\bigskip

\noindent	
	$\boxed{j=2}$ We consider 
	$$ R_2 = \iint K_2(t,x,y) a_x^* a_y^*  \left( \frac{\sqrt{N-\cN -1} \sqrt{N - \cN}(N -\cN -2)}{N^2} - 1\right).$$
	For any $\Phi \in \cF^{\le m}_+(t)$, we have
	\begin{align} \label{eq:R2-000}
	|\langle\Phi,R_2\Phi\rangle| &=	\left|\iint K_2(t,x,y) \Big\langle\Phi,  a^*_xa^*_y\left( \frac{\sqrt{N-\cN -1} \sqrt{N - \cN}(N -\cN -2)}{N^2} - 1\right) \Phi \Big\rangle \d x dy\right| \nn\\ 
	&= \Big|\iiint V_N(x-y,x-z) |u(t,z)|^2 u(t,x) u(t,y) \times \nn\\
	&\quad \times \Big\langle a_x a_y \Phi,  \left( \frac{\sqrt{N-\cN -1} \sqrt{N - \cN}(N -\cN -2)}{N^2} - 1\right) \Phi \Big\rangle \d x dy\Big|  \nn\\
	  &\le \|u(t,.)\|_{L^\infty}^3  \iiint V_N(x-y,x-z) |u(t,z)| \times \nn\\
	  &\quad \times \|a_xa_y\Phi\|\left\|\left(\frac{\sqrt{N-\cN -1} \sqrt{N - \cN}(N -\cN -2)}{N^2} - 1\right)\Phi\right\| \d x \d y \d z.	
	 \end{align}
	 In the above we could replace $K_2(t,x,y) = Q(t)\otimes Q(t) \widetilde{K}_2(t,x,y)$ by $\widetilde{K}_2(t,x,y)$, namely could ignore the projection $Q(t)$, since $\Phi$ belongs to the excited Fock space $\cF_+(t)$. 
	
	In \eqref{eq:R2-000} we can use  again $\|u(t,.)\|_{L^\infty} \le C_t$. 	Moreover, using $\sqrt{1-s}=1+O(s)$ with $s>0$ small, it is straightforward to see that
	$$
	\left| \frac{\sqrt{N-\cN -1} \sqrt{N - \cN}(N -\cN -2)}{N^2} - 1 \right| \le C \frac{\cN+1}{N},
	$$
	and hence
	\begin{align*}
	&\left\|\left(\frac{\sqrt{N-\cN -1} \sqrt{N - \cN}(N -\cN -2)}{N^2} - 1\right)\Phi\right\| \\
	&\le C \left\langle\Phi,\frac{(\cN +1)^2}{N^2}\Phi \right\rangle^{\frac{1}{2}}\le \frac{C\sqrt{m}}{N} \left\langle\Phi, (\cN +1)\Phi \right\rangle^{1/2}
	\end{align*}
	because $\Phi\in \cF_+^{\le m}(t)$. 
	
	Next, by the Cauchy-Schwarz inequality we can estimate
\begin{align*}
 & \iiint V_N(x-y,x-z) |u(t,z)| \|a_xa_y\Phi\| \d x \d y \d z \\
 &\le \left( \iiint V_N(x-y,x-z) |u(t,z)|^2 \d x \d y \d z \right)^{1/2} \times\\
 &\quad \times  \left( \iiint V_N(x-y,x-z)  \|a_xa_y \Phi\|^2 \d x \d y \d z \right)^{1/2} \\
 &\le C \left( \iiint V_N(x-y,x-z)  \|a_xa_y \Phi\|^2 \d x \d y \d z \right)^{1/2}. 
\end{align*}

By Lemma \ref{lem:3body ope}, we can estimate 
	\begin{align*}
	\int V_N(x-y,x-z) \d z =  \int N^{3\beta}V(N^{\beta}(x-y),z) \d z \le C N^{\beta} (-\Delta_x).
	\end{align*}
	Therefore,
	\begin{equation} \label{eq:R4aaa}
	\iiint V_N(x-y,x-z)  a_x^* a_y^* a_x a_y \d z \le C N^{\beta} \dGamma(-\Delta) \cN,
	\end{equation}
	and hence,  since $\Phi\in \cF_+^{\le m}(t)$,
	\begin{equation} \label{eq:R4PPP}
	\iiint V_N(x-y,x-z) \|a_xa_y\Phi\|^2 \d x \d y \d z  \le C N^{\beta}m \langle \Phi, \dGamma(-\Delta) \Phi\rangle. 
	\end{equation}
	
	Thus we can conclude from \eqref{eq:R2-000} that
	\begin{align*} 
	|\langle\Phi,R_2\Phi\rangle| \le  C_t mN^{\beta/2-1} \langle\Phi,\dGamma(1-\Delta)\Phi\rangle.
	\end{align*}
	Therefore, we have the quadratic form bound on $\cF_+^{\le m}(t)$:
	\begin{align} 
	\label{A2bound}
	\pm (R_2+R_2^*) \le C_t  mN^{\beta/2-1}  \dGamma(1-\Delta). 
	\end{align}
	Note that $mN^{\beta/2-1} \le \sqrt{mN^{2\beta-1}}$ when $m\le N$. 
	\medskip
	
	\noindent
	$\boxed{j=3}$ Now we consider
	\begin{align*}
	R_3 &= \frac{1}{3N^2} \idotsint(Q(t)\otimes Q(t)\otimes Q(t)V_N 1\otimes 1\otimes 1)(x,y,z;x',y',z')u(t,x')u(t,y')u(t,z') \\ 
& \quad\quad\quad\times a^*_x a^*_y a^*_z\d x \d y \d z\d x '\d y'\d z' \sqrt{N-\cN-2}\sqrt{N-\cN-1}\sqrt{N-\cN}  \\
&\quad+\frac{2}{N^2} \iiiint (Q(t)\otimes Q(t)\int|u(t,z)|^2V_N\d zQ(t)\otimes 1)(x,y;x',y') u(t,y') \times \\
&\quad\quad \quad \times a^*_xa^*_ya_{x'}\d x dy\d x '\d y'\sqrt{N-\cN}(N-\cN-1)  \\
&\quad+\frac{1}{N^2} \idotsint(Q(t)\otimes Q(t)\otimes 1 V_N 1\otimes 1 \otimes Q(t))(x,y,z;x',y',z') \overline{u(t,z)}u(t,x')u(t,z') \times \\
&\quad\quad\quad\times a^*_xa^*_ya_{z'}\d x \d y \d z\d x '\d y'\d z'\sqrt{N -\cN}(N-\cN-1)  \\
&=: R_{3,1} +R_{3,2} + R_{3,3}.
\end{align*}
As before, we will estimate the expectation of $R_3$ against an arbitrary vector $\Phi \in \cF_+^{\le m}(t)$. Again, we can ignore the projection $Q(t)$, since $\Phi$ belongs to the excited Fock space $\cF_+(t)$.

\bigskip

\noindent
	$\boxed{R_{3,1}}$ This is the most complicated term (which does not appear in the pair-interaction case). We start with 
	\begin{align*}
	|\langle\Phi,R_{3,1}\Phi\rangle| &= \frac{1}{3N^2} \Big|\iiint V_N(x-y,x-z)u(t,x)u(t,y)u(t,z)\langle\Phi,a^*_xa^*_ya^*_z\Phi_1\rangle \d x \d y \d z\Big| \\
	&\le \frac{1}{3N^2} \int |u(t,z)| \left| \left\langle a_z \Phi, \Big( \iint k_z(x,y) a_x^*a_y^*  \d x \d y \Big) \Phi_1 \right\rangle \right| \d z
	\end{align*}
	with 
	\begin{align}
	\Phi_1&= \sqrt{N-\cN-2}\sqrt{N-\cN-1}\sqrt{N-\cN}\Phi, \nn \\
	k_z(x,y)&=V_N(x-y,x-z) u(t,x) u(t,y). \label{eq:def-kz}
	\end{align}
	
		For every $z\in \R^3$ fixed, we can think of $k_z(x,y)$ is the kernel of an operator $k_z$ on $L^2(\R^3)$. Let us prove that
		\begin{align} \label{eq:kz-op-bound}
		\pm (k_z+k_z^*) \le C_t N^\beta (-\Delta).
		\end{align}
		Indeed, for every $f \in L^2(\R^3)$ by the Cauchy-Schwarz inequality and Lemma \ref{lem:3body ope} we have
		\begin{align*}
		|\langle f, k_z f\rangle| &\le \iint |f(x)| |f(y)| V_N(x-y,x-z) |u(t,x)| |u(t,y)| \d x \d y \\
		&\le \|u(t,.)\|_{L^\infty}^2 \frac{|f(x)|^2+|f(y)|^2}{2} V_N(x-y,x-z) \d x \d y \\
		&\le C_t \int |f(x)|^2 N^{3\beta} \Big( \int V(y, N^\beta(x-z)) \d y\Big) \d x\\
		&\le C_t \Big\|  N^{3\beta} \int V(y, N^\beta x) \d y  \Big\|_{L^{3/2}(\R^3, \d x)} \int |\nabla f(x)|^2 \d x\\
		&\le C_t N^{\beta} \int |\nabla f(x)|^2 \d x.
		\end{align*}
		Thus for every $z\in \R^3$ fixed, we can use the pairing term estimate in Lemma \ref{lem:Pairtermbound} with
		$$
		K(x,y)=k_z(x,y),\quad H= \eta C_t N^\beta (1-\Delta), \quad \eta\ge 1.
		$$
		Combining with the kernel estimate in Lemma \ref{lem:R3kernelbound} we find that
	\begin{align} \label{eq:pairing-kz}
	 \pm \frac{1}{2} \left(  \iint k_z (x,y)a^*_xa^*_y \d x\d y + h.c.\right) \le  C_t \Big( \eta N^\beta  \d\Gamma(1-\Delta) + \eta^{-1} N^{3\beta} \Big), \quad \forall \eta\ge 1.
	\end{align}
		
	Using $\eqref{eq:pairing-kz}$ and  the Cauchy-Schwarz inequality we also have
	\begin{align*}
	&\left| \left\langle a_z \Phi, \Big( \iint k_z(x,y) a_x^*a_y^*  \d x \d y \Big) \Phi_1 \right\rangle \right| \\
	&\le C_t  \Big\langle \Phi_1, \left( \eta N^\beta  \d\Gamma(1-\Delta) + \eta^{-1} N^{3\beta}  \right)  \Phi_1\Big\rangle ^{1/2} 
	\Big\langle a_z \Phi, \left( \eta N^\beta  \d\Gamma(1-\Delta) + \eta^{-1} N^{3\beta}  \right) a_z \Phi \Big\rangle ^{1/2} .
	\end{align*}
	
	Integrating the above estimate against $|u(t,z)| \d z$ and using the Cauchy-Schwarz inequality 	 we get
\begin{align*}	
|\langle\Phi,R_{3,1}\Phi\rangle|  &\le \frac{1}{3N^2} \int |u(t,z)| \left| \iint k_z(x,y) a^*_xa^*_y \langle a_z \Phi, \Phi_1\rangle \d x \d y \right| \d z\\
&\le \frac{C_t}{N^2}  \Big\langle \Phi_1, \left( \eta N^\beta  \d\Gamma(1-\Delta) + \eta^{-1} N^{3\beta}  \right)  \Phi_1\Big\rangle ^{1/2}  \times \\
&\quad \times  \int |u(t,z)|  \Big\langle a_z \Phi, \left( \eta N^\beta  \d\Gamma(1-\Delta) + \eta^{-1} N^{3\beta}  \right) a_z \Phi \Big\rangle ^{1/2} \d z \\
&\le  \frac{C_t}{N^2}  \Big\langle \Phi_1, \left( \eta N^\beta  \d\Gamma(1-\Delta) + \eta^{-1} N^{3\beta}  \right)  \Phi_1 \Big\rangle ^{1/2} \\
&\quad \times \|u(t,.)\|_{L^2} \left( \int \Big\langle a_z \Phi, \left( \eta N^\beta  \d\Gamma(1-\Delta) + \eta^{-1} N^{3\beta}  \right)  a_z \Phi \Big\rangle \d z \right)^{1/2}.
\end{align*}
The term involving $\Phi_1$ can be estimated   using
	$$0\le \sqrt{N-\cN-2}\sqrt{N-\cN-1}\sqrt{N-\cN} \le N^{3/2}$$
	on $\cF^{\le m}_+(t)$ (and that $\cN$ commutes with $\dGamma(1-\Delta)$).  For the other term, we use
$$
\int a_z^* a_z \d z = \cN
$$
and
\begin{align*}
\int a_z^* \dGamma(1-\Delta) a_z \d z &= \iint a_z^* a_x^* (1-\Delta_x) a_x a_z \d x \d z\\
&= \iint a_x^* (1-\Delta_x) (a_z^* a_x) a_z \d x \d z\\
&=  \iint a_x^* (1-\Delta_x) (a_x a_z^*-\delta_{x=z}) a_z \d x \d z\\
&= \dGamma(1-\Delta) (\cN-1).
\end{align*}
Therefore, when $\Phi\in \cF_+^{\le m}(t)$, we have
\begin{align*}	
|\langle\Phi,R_{3,1}\Phi\rangle| \le C_t\sqrt{\frac{m}{N}} \Big\langle \Phi, \left( \eta N^\beta  \d\Gamma(1-\Delta) + \eta^{-1} N^{3\beta}  \right)   \Phi\Big\rangle, 
\quad \forall \eta\ge 1.
\end{align*}
Thus
\begin{align*}	
\pm (R_{3,1}+R_{3,1}^*) \le C_t \Big( \eta \sqrt{m N^{2\beta-1}}  \dGamma(1-\Delta) + \eta^{-1} \sqrt{m N^{6\beta-1}}),  \quad \forall \eta \ge 1.
\end{align*}
This is equivalent to 
\begin{align*}	
\pm (R_{3,1}+R_{3,1}^*) \le \eta  \dGamma(1-\Delta) + \eta^{-1} C_t mN^{4\beta-1},  \quad \forall \eta \ge C_t \sqrt{m N^{2\beta-1}}.
\end{align*}

	\bigskip
	
	\noindent
	$\boxed{R_{3,2}}$ By the Cauchy-Schwarz inequality we have that  
		\begin{align*}
	 &|\langle \Phi,R_{3,2}\Phi\rangle| = \frac{2}{N^2}\Big|\iiint V_N(x-y,x-z)|u(t,z)|^2u(t,y) \times \nn \\
	&\qquad\qquad\qquad\qquad  \times \Big\langle\Phi,a^*_xa^*_y a_x \sqrt{N-\cN}(N-\cN -1) \Phi \Big\rangle \d x \d y \d z\Big| \nn \\ 
	&\le \frac{2 \|u(t,.)\|_{L^\infty}^3}{N^2} \iiint V_N(x-y,x-z) \| a_x a_y \Phi \| \Big\| a_x \sqrt{N-\cN}(N-\cN -1) \Phi \Big\| \d x \d y \d z \nn\\
	&\le \frac{C_t}{N^2} \left(\iiint V_N(x-y,x-z)  \|a_x a_y \Phi\|^2\d x \d y \d z\right)^{\frac{1}{2}} \times \nn \\ 
	&\quad\quad\quad\times \left(\iiint V_N(x-y,x-z) \Big\| a_x \sqrt{N-\cN}(N-\cN -1) \Phi \Big\| ^2 \d x \d y \d z\right)^{\frac{1}{2}}.
	\end{align*}
	The first term has been estimated as in \eqref{eq:R4PPP},
	$$
	\iiint V_N(x-y,x-z) \|a_xa_y\Phi\|^2 \d x \d y \d z  \le C N^{\beta}m \langle \Phi, \dGamma(-\Delta) \Phi\rangle. 
	$$
	The second term can be computed explicitly (by doing the integration over $y,z$ first)
	\begin{align*}
	&\iiint V_N(x-y,x-z) \Big\| a_x \sqrt{N-\cN}(N-\cN -1) \Phi \Big\| ^2 \d x \d y \d z\\
	&= \left( \iint V(x,y) \d x \d y \right) \int \Big\| a_x \sqrt{N-\cN}(N-\cN -1) \Phi \Big\| ^2 \d x\\
	&= \left( \iint V(x,y) \d x \d y \right) \Big\langle \Phi, C \sqrt{N-\cN}(N-\cN -1) \cN \sqrt{N-\cN}(N-\cN -1) \Phi \Big\rangle \\
	&\le C N^{3}  \Big\langle \Phi, (\cN+1) \Phi \Big\rangle. 
	\end{align*}
	Putting all together, we obtain
	$$
	|\langle \Phi,R_{3,2}\Phi\rangle|  \le \frac{C_t}{N^2} \sqrt{N^{\beta}m \langle \Phi, \dGamma(-\Delta) \Phi\rangle} \sqrt{N^{3}  \Big\langle \Phi, (1+\cN) \Phi \Big\rangle} \le C_t \sqrt{mN^{\beta-1}}  \langle \Phi, \dGamma(1-\Delta) \Phi\rangle.
	$$
	Thus
	$$
	\pm (R_{3,2}+R_{3,2}^*) \le C_t \sqrt{mN^{\beta-1}}  \langle \Phi, \dGamma(1-\Delta) \Phi\rangle.
	$$

	\bigskip
	
	\noindent
	$\boxed{R_{3,3}}$ Using the Cauchy-Schwarz inequality we get for $\Phi \in \cF^{\le m}_+(t)$
	\begin{align*} 
	 &|\langle\Phi, R_{3,3}\Phi\rangle| =\frac{1}{N^2}\Big|\iiint V_N(x-y,x-z)  \overline{u(t,z)}u(t,x)u(t,y)\times\nn \\
	 &\qquad\qquad\qquad\qquad\times\Big\langle\Phi,a^*_xa^*_ya_z\sqrt{N-\cN}(N-\cN-1)\Phi\Big\rangle \d x \d y \d z\Big| \nn \\
	 &\le \frac{\|u(t,.)\|_{L^\infty}^3}{N^2}\Big|\iiint V_N(x-y,x-z)   \|a_x a_y \Phi \| \Big\|a_z\sqrt{N-\cN}(N-\cN-1)\Phi\Big\|  \d x \d y \d z\\
	 &\le \frac{C_t}{N^2} \left(\iiint V_N(x-y,x-z) \|a_x a_y \Phi\|^2 \d x \d y \d z\right)^{\frac{1}{2}} \times \nn\\
	 &\quad\quad\quad\quad\times \left(\iiint V_N(x-y,x-z) \Big\| a_z\sqrt{N-\cN}(N-\cN-1) \Phi\Big \|^2 \d x \d y \d z\right)^{\frac{1}{2}}.
	\end{align*}
	Then we can proceed exactly as for the previous term $R_{3,2}$ and obtain
	$$
	|\langle \Phi,R_{3,3}\Phi\rangle|  \le C_t \sqrt{mN^{\beta-1}}  \langle \Phi, \dGamma(1-\Delta) \Phi\rangle.
	$$
	Thus
	$$
	\pm (R_{3,3}+R_{3,3}) \le C_t \sqrt{mN^{\beta-1}} \dGamma(1-\Delta).
	$$

	\bigskip
	\noindent
	$\boxed{j=4}$ We consider
	\begin{align*}
	R_4&= \frac{1}{N^2}\idotsint(Q(t)\otimes Q(t)\otimes Q(t)V_NQ(t)\otimes 1\otimes 1)(x,y,z;x',y',z')u(t,y')u(t,z')\times \\
&\quad\quad\quad \times a^*_xa^*_ya^*_za_{x'}\sqrt{N-\cN -1}\sqrt{N-\cN}\d x \d y \d z\d x '\d y'\d z'  + \\
&\quad+ \frac{1}{2N^2} \iiiint (Q(t)\otimes Q(t) \int |u(t,z)|^2 V_N\d zQ(t)\otimes Q(t))(x,y;x',y') \times\\
&\quad \quad \quad a^*_xa^*_ya_{x'}a_{y'}\d x dy\d x '\d y'(N-\cN)\\
&\quad+\frac{1}{N^2}\idotsint (Q(t)\otimes Q(t)\otimes 1 V_N Q(t)\otimes 1\otimes Q(t))(x,y,z;x'y',z')\overline{u(t,z)}u(t,y')\times\\
&\quad\quad\quad\quad\quad\quad\quad\quad \times a^*_xa^*_ya_{x'}a_{z'}(N-\cN)\d x \d y \d z\d x '\d y'\d z'\\
&=:R_{4,1} + R_{4,2} + R_{4,3}.
\end{align*}	
Again, we will estimate the expectation of $R_4$ against $\Phi \in \cF_+^{\le m}(t)$, and for this purpose we ignore the projection $Q(t)$ in the computation because $\Phi$ belongs to the excited Fock space.

\bigskip

\noindent
	$\boxed{R_{4,1}}$ By the Cauchy-Schwarz inequality we can estimate
	
	\begin{align*}
	 |\langle \Phi,R_{4,1}\Phi\rangle|  & = \frac{1}{N^2} \Big|\iiint V_N(x-y,x-z)u(t,y)u(t,z) \times\\
	&\qquad \qquad \qquad \times \Big\langle\Phi, a^*_xa^*_ya^*_z a_x \sqrt{N -\cN -1}\sqrt{N-\cN}\Phi \Big \rangle \d x \d y \d z\Big|\\
	&\le \frac{1}{N^2} \iiint V_N(x-y,x-z) |u(t,y)u(t,z)| \times \\
	& \qquad \qquad \qquad \times \|a_x a_y a_z \Phi\| \| a_x \sqrt{N -\cN -1}\sqrt{N-\cN}\Phi \|  \d x \d y \d z \\
	&\le \frac{C_t}{N^2} \Big( \iiint V_N(x-y,x-z) \|a_x a_y a_z \Phi\|^2 \d x \d y \d z \Big)^{1/2} \times\\
	&\qquad \qquad \qquad \times \Big( \iiint  |u(t,y)|^2|u(t,z)|^2 \| a_x \sqrt{N -\cN -1}\sqrt{N-\cN}\Phi \|^2 \d x \d y \d z\Big)^{1/2}\\
	&\le C_t \langle \Phi, R_6 \Phi\rangle^{1/2} \langle \Phi, \cN \Phi\rangle^{1/2}. 
	\end{align*}
Then using the bound $R_6\le C m^2 N^{4\beta-2} \dGamma(-\Delta)$ (see \eqref{eq:bound-R6} below) we find that 
$$
|\langle \Phi,R_{4,1}\Phi\rangle| \le C_t mN^{2\beta-1} \dGamma(1-\Delta).
$$
Thus
$$
\pm (R_{4,1} +R_{4,1}^*) \le C_t mN^{2\beta-1} \dGamma(1-\Delta).
$$

	\bigskip 
	
	\noindent
	$\boxed{R_{4,2}}$
	By Lemma \ref{lem:3body ope} and the facts that $\|u(t,.)\|_{L^\infty}\le C_t$ and	\begin{align*}
	\left\|\int V_N(\cdot,z) \d z\right\|_{L^{3/2}(\R^3)} = \left(\int \left|N^{6\beta}\int V(N^\beta x,N^\beta z)\d z\right|^\frac{3}{2}\d x\right)^\frac{2}{3} \\
	= N^\beta \left(\int \left|\int V( x,z)\d z\right|^\frac{3}{2}\d x\right)^\frac{2}{3} \le C N^\beta
	\end{align*}
	we get 
	$$ 0\le R_{4,2} \le C_tN^{\beta-1}\d\Gamma(-\Delta)\cN.$$
	On the truncated Fock space $\cF_+^{\le m}(t)$, we obtain
	\begin{align} 0\le R_{4,2} \le C_t m N^{\beta-1}\d\Gamma(-\Delta). \label{eq:bound-R42}
	\end{align}

	\bigskip
	
	\noindent
	$\boxed{R_{4,3}}$ Using the Cauchy-Schwarz inequality we have for $\Phi \in \cF^{\le N}_+(t)$
	\begin{align*}
	|\langle\Phi, R_{4,1}\Phi\rangle| &=\frac{1}{N^2}\left|\iiint V_N(x-y,x-z)\overline{u(t,z)}u(t,y) \langle\Phi,a^*_xa^*_ya_xa_z(N-\cN)\Phi\rangle \d x \d y \d z\right| \\
	&\le \frac{1}{N^2}\left(\iiint V_N(x-y,x-z)|u(t,z)|^2\|a_xa_y\sqrt{N-\cN}\Phi\|^2\d x \d y \d z\right)^\frac{1}{2}\times \\
	&\quad\quad\quad\times\left(\iiint V_N(x-y,x-z)|u(t,y)|^2\|a_xa_z\sqrt{N-\cN}\Phi\|^2\d x \d y \d z\right)^\frac{1}{2}\\
	&= \langle\Phi,R_{4,2}\Phi\rangle.	
	\end{align*}
	Therefore, from the above bound of $R_{4,2}$ we have
	$$
	\pm (R_{4,3}+R_{4,3}) \le 2R_{4,2} \le C_t m N^{\beta-1}\d\Gamma(-\Delta)
	$$
	as quadratic forms on $\cF_+^{\le m}(t)$.

	\bigskip
	
	\noindent
$\boxed{j=5}$  We consider
\begin{align*}
R_5& = \frac{1}{N^2}\idotsint(Q(t)\otimes Q(t)\otimes Q(t)V_NQ(t)\otimes Q(t)\otimes 1)(x,y,z;x',y',z')u(t,z')\times \\
&\quad\quad\quad\quad\quad\quad\quad\quad \times a^*_xa^*_ya^*_za_{x'}a_{y'}\sqrt{N-\cN -2}\d x \d y \d z\d x '\d y'\d z' .
\end{align*}

For $\Phi\in \cF_+^{\le m}(t)$, by the Cauchy-Schwarz inequality and the previous bound on $R_{4,2}$, we have
\begin{align*}
|\langle \Phi, R_5\Phi\rangle| &= \frac{1}{N^2}\left|\iiint V_n(x-y,x-z) u(t,z)\langle \Phi,a^*_x a^*_y a^*_z a_x a_y\sqrt{N -\cN+2}\Phi\rangle \d x \d y \d z\right| \\
&\le \frac{1}{N^2} \iiint |V_N(x-y,x-z)| |u(t,z)|\|a_xa_ya_z \Phi\| \|a_xa_y\sqrt{N -\cN+2} \Phi\| dy\d y \d z \\
&\le \frac{1}{N^2}\left(\iiint |V_N(x-y,x-z)||u(t,z)|^2 \|a_xa_y\sqrt{N -\cN+2} \Phi\|^2 \d x \d y \d z\right)^{\frac{1}{2}}\times \\ &\quad\quad\times\left(\iiint |V_N(x-y,x-z)| \|a_xa_ya_z \Phi\|^2 \d x \d y \d z\right)^{\frac{1}{2}} \\
&\le \langle\Phi,R_{4,2}\Phi\rangle^{1/2} \langle\Phi,R_6\Phi\rangle^{1/2} \\
&\le C_t \sqrt{m^3 N^{5\beta-3}}\langle\Phi, \dGamma(1-\Delta) \Phi\rangle.
\end{align*}
Here in the last estimate we have used the bound on $R_{4,2}$ in \eqref{eq:bound-R42} above and the bound $0\le R_6\le Cm^2 N^{4\beta-2} \dGamma(-\Delta)$ in \eqref{eq:bound-R6} below. Thus
$$
\pm (R_{5}+ R_5^*) \le  C_t  \sqrt{m^3 N^{5\beta-3}} \dGamma(1-\Delta) .
$$

\bigskip

\noindent
$\boxed{j=6}$ We consider
\begin{align*}
R_6 &= \frac{1}{6N^2}\idotsint(Q(t)\otimes Q(t)\otimes Q(t)V_NQ(t)\otimes Q(t)\otimes Q(t)(x,y,z;x',y',z')\times \\
&\quad\quad\quad\quad\quad\quad\quad\quad \times a^*_xa^*_ya^*_za_{x'}a_{y'}a_{z'}\d x \d y \d z\d x '\d y'\d z'.
\end{align*}	
By Lemma \ref{lem:3body ope} we have
\begin{align*}
0 &\le V_N(x-y,x-z) \le \sup_{z} N^{6\beta} V(N^{\beta}(x-y),z) \\
&\le C N^{3\beta} \| \sup_{z} N^{3\beta} V (N^{\beta} \cdot ,z)\|_{L^{3/2}} (-\Delta_x) \le CN^{4\beta} (-\Delta_x).
\end{align*}
Therefore,  
$$
0\le R_6\le C N^{4\beta-2} \dGamma(-\Delta) \cN^2.
$$
In particular, on the truncated Fock space $\cF_+^{\le m}(t)$, 
\begin{align} \label{eq:bound-R6}
0\le R_6\le C m^2 N^{4\beta-2} \dGamma(-\Delta).
\end{align}
This finishes the proof of \eqref{eq:errorR0-5}.

\bigskip

%
%
%
%
	
\bigskip

\noindent	
{\bf Proof of the commutator bounds in \eqref{eq:errordRcom}.} These bounds follows from \eqref{eq:errorR0-5} with $\eta =1$ and the fact that
\begin{align*}
&[R_0,\cN]=[R_{4,2},\cN] = [R_{4,3},\cN] = [R_6,\cN] = 0, \\
&[R_1,\cN] = R_1,\quad [R_2,\cN] = -R_2 \\
&[R_{3,1},\cN] = -3R_{3,1}, \quad[R_{3,2},\cN] = R_{3,2}, \quad [R_{3,3},\cN] = R_{3,3}, \\
& [R_{4,1},\cN] = -2R_{4,1},
 \quad [R_5,\cN] = -R_5.
 \end{align*}

\bigskip

\noindent
{\bf Proof of the derivative bounds \eqref{eq:errorRdt}.}  Heuristically these bounds are similar to \eqref{eq:errorR0-5}. For the reader's convenience we will again go term by term.
 
 \bigskip

\noindent
	$\boxed{j=0} $ 
	Since 
	\begin{align*}
	&\Big|\partial_t\iiint |u(t,x)|^2|u(t,y)|^2|u(t,z)|^2V_N(x-y,x-z)\d x\d y\d z\Big| \\
	&\le 3\iiint |\partial_tu(t,x)||u(t,x)|u(t,y)|^2|u(t,z)|^2V_N(x-y,x-z)\d x\d y\d z \le C_t,
	\end{align*}
	we obtain
	\begin{align*}
	\pm\partial_tR_{0,1}= &\pm\frac{1}{6}\partial_t\iiint |u(x)|^2|u(y)|^2|u(z)|^2V_N(x-y,x-z)\d x\d y\d z\times\\
	&\times\left(\frac{3\cN^2+6\cN+2}{N}-\frac{\cN^3 +3\cN^2 +2\cN}{N^2}\right) \le C_t\frac{(\cN +1)^2}{N}.
	\end{align*}\\
	
	Moreover, using  $\|Q(t)\|_{op},\,\|\partial_tQ(t)\|_{op}\le C_t$, we have
 	\begin{align*}
 	\left\|\partial_t\left(Q(t)\left(\frac{1}{2}\iint|u(t)|^2V_N|u(t)|^2\d y \d z + K_1\right) Q(t)\right)\right\|_{op}\le C_t.
 	\end{align*}
 	Hence, as above
	\begin{align*}
	\pm\partial_t R_{0,2} &= \pm\d\Gamma\left(\partial_t\left(Q(t)\left(\frac{1}{2}\iint|u(t)|^2V_N|u(t)|^2\d y \d z + K_1\right) Q(t)\right)\right) \\
	&\quad \quad \times \left(\frac{(N-\cN)^2 - (N-\cN)}{N^2} -1\right) \le C_t \frac{(\cN+1)^2}{N}.
	\end{align*}
	In summary, 
	$$
	\pm \partial_t R_0 \le C_t\frac{(\cN+1)^2}{N}\le C_t \frac{m}{N} (\cN+1).
	$$
	
	\noindent 
	$\boxed{j=1}$ The term $\partial_t R_1$ can be estimated similarly to $R_1$ and we get
	$$
	\pm \partial_t (R_1+R_1^*) \le C_t\sqrt{\frac{m}{N}}(\cN+1).
	$$
	
	\bigskip
	
	\noindent
	$\boxed{j=2}$ Let $\Phi \in \cF^{\le m}_+(t)$. Then we have
	\begin{align*}
	\langle\Phi,\partial_tR_2\Phi\rangle & =
	\iint \left((\partial_tQ(t)\otimes 1  + 1\otimes\partial_tQ(t))\widetilde{K}_2(t,x,y) +\partial_t\widetilde{K}_2(t,x,y)\right)\times\\
	&\quad\quad\quad\times
	\Big\langle\Phi,	a^*_xa^*_y\left( \frac{\sqrt{N-\cN -1} \sqrt{N - \cN}(N -\cN -2)}{N^2} - 1\right)\Phi\Big\rangle \d x \d y,
	\end{align*}
	where we have used 
	\begin{align*}\partial_tK_2(t,x,y) &=\partial_tQ(t)\otimes Q(t)\widetilde{K}_2(t,x,y) + Q(t)\otimes\partial_tQ(t)\widetilde{K}_2(t,x,y) \\
	&\qquad +Q(t)\otimes Q(t)\partial_t\widetilde{K}_2(t,x,y)
	\end{align*}
	and then left out all $Q(t)$ since $\Phi \in \cF^{\le m}_+(t)$. 
	
	As in \eqref{A2bound} we have
	\begin{align*}
	&\left|\iint \partial_t\widetilde{K}_2(t,x,y)\langle\Phi,
	a^*_xa^*_y\left( \frac{\sqrt{N-\cN -1} \sqrt{N - \cN}(N -\cN -2)}{N^2} - 1\right)\Phi\rangle \d x \d y\right| \\
	&\le\iiint \left|\partial_t|u(t,z)|^2u(t,x)u(t,y) + 2|u(t,z)|^2\partial_tu(t,x)u(t,y)\right|V_N(x-y,x-z)\times\\
	&\quad\quad\quad\times\left|\langle\Phi,
	a^*_xa^*_y\left( \frac{\sqrt{N-\cN -1} \sqrt{N - \cN}(N -\cN -2)}{N^2} - 1\right)\Phi\rangle\right| \d x \d y\d z \\
	&\le C_t \left(\iiint |u(t,z)|^2|V_N(x-y,x-z)| \|a_xa_y\Phi\|^2 \d x \d y \d z\right)^\frac{1}{2} \times\\ &\quad\quad\quad \times\left\|\left(\frac{\sqrt{N-\cN -1} \sqrt{N - \cN}(N -\cN -2)}{N^2} - 1\right)\Phi\right\| \\
	&\le C_t \langle\Phi,R_{4,2}\Phi\rangle^{1/2} \langle\Phi,(\cN+1)\Phi\rangle^{1/2} \\
	& \le C_t \sqrt{mN^{\beta-1}} \langle\Phi,\d\Gamma(1-\Delta)\Phi\rangle .
	\end{align*}
	The other terms follow as
	\begin{align*}
	&\|(|\partial_tu\rangle\langle u|\otimes 1) \tilde{K_2}(t,\cdot,\cdot)\|^2_{L^2(\R^6)} = 
	\iint \left|\int \overline{u(t,\tilde{x})}\tilde{K_2}(t,\tilde{x},y)\d\tilde{x}\right|^2|\partial_tu(t,x)|^2\d x \d y \\
	&= \iint \left|\iint\overline{u(t,\tilde{x})}|u(t,z)|^2u(t,\tilde{x})	V_N(\tilde{x}-y,\tilde{x}-z)u(t,y)\d z\d\tilde{x}\right|^2|\partial_tu(t,x)|^2 \d x \d y \\
	&\le \|u(t,\cdot)\|^8_{L^\infty(\R^3)}\|V\|^2_{L^1(\R^6)}\|u(t,\cdot)\|^2_{L^2(\R^3)}\|\partial_t u(t,\cdot)\|^2_{L^2(\R^3)} \le C_t.
	\end{align*}
	Using this and similar estimates we obtain
	\begin{align*}
	&\iint \left((\partial_tQ(t)\otimes 1  + 1\otimes\partial_tQ(t))\tilde{K}_2(t,x,y)\right)\times\\
	&\quad\quad\quad\times\langle\Phi,
	a^*_xa^*_y\left( \frac{\sqrt{N-\cN -1} \sqrt{N - \cN}(N -\cN -2)}{N^2} - 1\right)\Phi\rangle \d x \d y \\ 
	&\quad\le \Big\|(\partial_tQ(t)\otimes 1  + 1\otimes\partial_tQ(t))\tilde{K}_2(t,\cdot,\cdot)\Big\|_{L^2(\R^6)} \left(\iint \|a_xa_y\Phi\|^2\d x \d y\right)^\frac{1}{2} \times\\
	&\quad\quad\quad\quad\times 	\left\|\left(\frac{\sqrt{N-\cN -1} 
	\sqrt{N - \cN}(N -\cN -2)}{N^2} - 1 \right)\Phi \right\| \\
	&\le \frac{C_t}{\sqrt{N}} \Big\langle\Phi,\cN^2\Phi\rangle^\frac{1}{2}\langle\Phi(\cN+1)\Phi \Big\rangle^\frac{1}{2} \le C_t \sqrt{mN^{-1}}\langle\Phi,(\cN+1)\Phi\rangle.
	\end{align*}
	Putting everything together we obtain
	\begin{align*}
	|\langle\Phi,\partial_tR_2\Phi \rangle| \le C_t  \sqrt{mN^{\beta-1}}  \langle \Phi,\d\Gamma(1-\Delta)\Phi\rangle.
	\end{align*}
	Thus
	$$
	\pm \partial_t(R_2+R_2^*) \le C_t  \sqrt{mN^{\beta-1}} \d\Gamma(1-\Delta).
	$$

	\bigskip
	
	\noindent
	\boxed{j=3} 
	$\boxed{R_{3,1}}$  Let $\Phi \in \cF^{\le m}_+(t)$. Then we have
	\begin{align*}
	&\langle\Phi,\partial_tR_{3,1}\Phi\rangle\\ &=  \frac{1}{3N^2}\idotsint\Big[(Q(t)\otimes Q(t)\otimes Q(t)V_N 1\otimes 1\otimes 1)(x,y,z;x',y',z')\partial_t\left(u(t,x')u(t,y')u(t,z')\right) + \\
	&\quad\quad\quad+ \partial_t(Q(t)\otimes Q(t)\otimes Q(t)V_N 1\otimes 1\otimes 1)(x,y,z;x',y',z')u(t,x')u(t,y')u(t,z') \Big] \\ 
	& \quad\quad\quad\quad\quad\times \langle\Phi,a^*_x a^*_y a^*_z \sqrt{N-\cN-2}\sqrt{N-\cN-1}\sqrt{N-\cN}\Phi\rangle \d x \d y \d z\d x '\d y'\d z'.
	\end{align*}
	The first term containing $\partial_t\left(u(t,x')u(t,y')u(t,z')\right)$ can be estimated exactly as for $R_{3,1}$. The second term containing $\partial_tQ(t)$ can be evaluated as
	\begin{align*}
	&\frac{1}{3N^2}\Big|\idotsint  (\partial_tQ(t)\otimes Q(t)\otimes Q(t)V_N 1\otimes 1\otimes 1)(x,y,z;x',y',z')u(t,x')u(t,y')u(t,z') \\ 
	& \quad\quad\quad\quad\quad\times \langle\Phi, a^*_x a^*_y a^*_z \sqrt{N-\cN-2}\sqrt{N-\cN-1}\sqrt{N-\cN}\Phi\rangle \d x \d y \d z\d x '\d y'\d z'\Big| \\
	&=\frac{1}{3N^2}\Big|\idotsint (\partial_tQ(t))(x;x') V_N(x'-y,x'-z)\delta(y-y')\delta(z-z')u(t,x')u(t,y')u(t,z') \\ 
	& \quad\quad\quad\quad\quad\times \langle\Phi, a^*_x a^*_y a^*_z \sqrt{N-\cN-2}\sqrt{N-\cN-1}\sqrt{N-\cN}\Phi\rangle \d x \d y \d z\d x '\d y'\d z'\Big| \\
	&=\frac{1}{3N^2}\Big|\iiiint (\partial_tQ(t))(x;x') V_N(x'-y,x'-z)u(t,x')u(t,y)u(t,z) \\ 
	& \quad\quad\quad\quad\quad\times \langle\Phi, a^*_x a^*_y a^*_z \sqrt{N-\cN-2}\sqrt{N-\cN-1}\sqrt{N-\cN}\Phi\rangle \d x \d y \d z\d x'\Big| \\
	&\le C_t \Big\langle \Phi, \Big( \eta  \d\Gamma(1-\Delta) + \eta^{-1} C_t mN^{4\beta-1} \Big)\Phi \Big\rangle, \quad \forall \eta\ge C_t \sqrt{mN^{2\beta-1}}
	\end{align*}
	where the last estimate follows similarly as for $R_{3,1}$ again. Hence, we have
	\begin{align*}
	\pm \partial_t (R_{3,1}+ R_{3,1}^*)\le C_t \left(\eta  \d\Gamma(1-\Delta) + \eta^{-1}C_t mN^{4\beta-1}\right), \quad \forall \eta\ge C_t \sqrt{mN^{2\beta-1}}
.
	\end{align*}
	
	\bigskip
	\noindent 
	
	$\boxed{R_{3,2}}$  Let $\Phi \in \cF^{\le m}_+(t)$. Then we have
	\begin{align*}
	\langle\Phi,\partial_tR_{3,2}\Phi\rangle &=
	\frac{2}{N^2} \idotsint \Big[(Q(t)\otimes Q(t)V_NQ(t)\otimes 1)(x,y;x',y') \partial_t\left(|u(t,z)|^2u(t,y')\right)+\\
	&\quad\quad\quad+\partial_t(Q(t)\otimes Q(t)V_NQ(t)\otimes 1)(x,y;x',y')|u(t,z)|^2u(t,y')\Big] \times \\
	&\quad\quad\quad\quad\times \langle\Phi,a^*_xa^*_ya_{x'}\sqrt{N-\cN}(N-\cN-1)\Phi\rangle \d x dy\d x '\d y'\d z.
	\end{align*}
	The first term containing $\partial_t\left(|u(t,z)|^2u(t,y')\right)$ can be estimated similarly to $R_{3,2}$, which is
	\begin{align*}
	\frac{2}{N^2}&\left|\iiint V_N(x-y,x-z)\partial_t\left(|u(t,z)|^2u(t,y)\right)\langle\Phi,a^*_xa^*_y a_x \sqrt{N-\cN}(N-\cN -1) \Phi\rangle \d x \d y \d z\right| \\
	&\quad\le \frac{4}{N^2} \left(\iiint |V_N(x-y,x-z)| |\partial_tu(t,y)|^2 |u(t,z)|^2 \|a_x(N-\cN -1)\Phi\|^2\d x \d y \d z\right)^{\frac{1}{2}} \times\\
	&\quad\quad\quad\quad\times \left(\iiint |V_N(x-y,x-z)| |u(t,z)|^2 \|a_xa_y \sqrt{N-\cN+1}\Phi\|^2 \d x \d y \d z\right)^{\frac{1}{2}}		\\
	&\quad\le \frac{C_t}{N^2}\left(\int\|a_x(N-\cN-1)\Phi\|^2\d x \right)^\frac{1}{2}\times\\
	&\quad\quad\quad\quad\times \left(\iiint |V_N(x-y,x-z)| |u(t,z)|^2 \|a_xa_y \sqrt{N-\cN+1}\Phi\|^2 \d x \d y \d z\right)^{\frac{1}{2}}\\ 
	&\le
	C_t \sqrt{mN^{\beta-1}} \Phi,\d\Gamma(-\Delta)\Phi\rangle.
	\end{align*}
	Here we have omitted all the $Q(t)$ since $\Phi \in \cF^{\le m}_+(t)$.
	For the other terms containing $\partial_tQ(t)$ we use the following kernel estimate
	\begin{align}
	|(\partial_tQ(t))(x;x')| = |\partial_tu(t,x)\overline{u(t,x')} + u(t,x)\partial_t\overline{u(t,x')}| \le q(t,x)q(t,x')
	\end{align}
	with $q(t,x) = |u(t,x)| + |\partial_tu(t,x)|$. Using Theorem~\eqref{thm:Hartree} we get
	\begin{align*}
	\|q(t,\cdot)\|_{L^2(\R^3)} \le C_t, \quad\quad \|q(t,\cdot)\|_{L^\infty(\R^3)} \le C_t
	\end{align*}
	We now decompose $\partial_t(Q(t)\otimes Q(t) V_NQ(t)\otimes 1)$ into three terms. The first one containing $(\partial_tQ(t)\otimes Q(t) V_NQ(t)\otimes 1)$ can be estimated as
	\begin{align*}
	&\frac{2}{N^2}\Big|\iiiint (\partial_tQ(t)\otimes Q(t)\int|u(t,z)|^2V_N\d zQ(t)\otimes 1)(x,y;x',y')u(t,y') \times
	\\
	 &\quad\quad\quad\quad\quad\times\langle\Phi,a^*_xa^*_ya_{x'}\sqrt{N-\cN}(N-\cN-1)\Phi\rangle \d x dy\d x '\d y'\Big| \\
	&=\frac{2}{N^2}\Big|\iiiint (\partial_tQ(t))(x;x')\int|u(t,z)|^2V_N(x'-y,x'-z)\d z\delta(y-y')u(t,y')\times
	\\
	&\quad\quad\quad\quad\quad\times\langle\Phi,a^*_xa^*_ya_{x'}\sqrt{N-\cN}(N-\cN-1)\Phi\rangle \d x dy\d x '\d y'\Big| \\
	&\le \frac{1}{N^2} \iiint q(t,x)q(t,x')\int|u(t,z)|^2V_N(x'-y,x'-z)\d z|u(t,y)|\times
	\\
	&\quad\quad\quad\quad\quad\times\|a_xa_y\sqrt{N-\cN+1}\Phi\|\|a_{x'}(N-\cN-1)\Phi\| \d x dy\d x ' \\
	&\le \frac{2}{N^2}\left(\iiiint |q(t,x)|^2|u(t,z)|^2|V_N(x'-y,x'-z)||u(t,y)|^2\|a_{x'}(N-\cN-1)\Phi\|^2\d x dy\d x '\d z\right)^\frac{1}{2}\times\\
	&\quad\quad\quad\times\left(\iiiint|q(t,x')|^2|u(t,z)|^2|V_N(x'-y,x'-z)|\|a_xa_y\sqrt{N-\cN+1}\Phi\|^2\d x dy\d x '\d z\right)^\frac{1}{2} \\
	&\le \frac{C_t}{N^2} \langle\Phi,\cN(N-\cN-1)^2\Phi\rangle^\frac{1}{2}\langle\Phi,\cN^2(N-\cN+1)\Phi\rangle^\frac{1}{2} \le C_t \sqrt{mN^{-1}}\langle\Phi,\cN\Phi\rangle.
	\end{align*}
	The second one $(Q(t)\otimes \partial_tQ(t) V_NQ(t)\otimes 1)$ follows as
	\begin{align*}
	&\frac{2}{N^2}\Big|\iiiint (Q(t)\otimes \partial_tQ(t)\int|u(t,z)|^2V_N\d zQ(t)\otimes 1)(x,y;x',y')u(t,y') \times
	\\
	&\quad\quad\quad\quad\quad\times\langle\Phi,a^*_xa^*_ya_{x'}\sqrt{N-\cN}(N-\cN-1)\Phi\rangle \d x dy\d x '\d y'\Big| \\
	&=\frac{2}{N^2}\Big|\iiiint (\partial_tQ(t))(y;y') \int|u(t,z)|^2V_N(x-y',x-z)\d z\delta(x-x')u(t,y')\times
	\\
	&\quad\quad\quad\quad\quad\times\langle\Phi,a^*_xa^*_ya_{x'}\sqrt{N-\cN}(N-\cN-1)\Phi\rangle \d x dy\d x '\d y'\Big| \\
	&\le\frac{2}{N^2} \iiint q(t,y)q(t,y')\int|u(t,z)|^2V_N(x-y',x-z)\d z|u(t,y')|\times
	\\
	&\quad\quad\quad\quad\quad\times\|a_xa_y\sqrt{N-\cN+1}\Phi\|\|a_{x}(N-\cN-1)\Phi\| \d x dy\d y'\Big| \\
	&\le \frac{2}{N^2} \|q(t,\cdot)\|_{L^\infty(\R^3)}\|u(t,\cdot)\|^3_{L^\infty(\R3)}\|V\|_{L^1(\R^6)}\times\\
	&\quad\times \left(\iint\|a_xa_y\sqrt{N-\cN+1}\Phi\|^2 \d x dy\right)^\frac{1}{2}\left(\iint |q(t,y)|^2\|a_{x}(N-\cN-1)\Phi\|^2\d x dy\right)^\frac{1}{2} \\
	&\le \frac{C_t}{N}\langle\Phi,\cN^2(N-\cN+1)\Phi\rangle^\frac{1}{2}\langle\Phi,\cN\Phi\rangle^\frac{1}{2} \le C_t \sqrt{mN^{-1}}\langle\Phi,\cN\Phi\rangle.
	\end{align*}
	For the third one $(Q(t)\otimes V_N\partial_tQ(t)Q(t)\otimes 1)$ we can estimate
	\begin{align*}
		&\frac{2}{N^2}\Big|\iiiint (Q(t)\otimes Q(t)\int|u(t,z)|^2V_N\d z\partial_tQ(t)\otimes 1)(x,y;x',y')u(t,y') \times
	\\
	&\quad\quad\quad\quad\quad\times\langle\Phi,a^*_xa^*_ya_{x'}\sqrt{N-\cN}(N-\cN-1)\Phi\rangle \d x dy\d x '\d y'\Big| \\
	&=\frac{2}{N^2}\Big|\iiiint \int|u(t,z)|^2V_N(x-y,x-z)\d z(\partial_tQ(t))(x;x')\delta(y-y')u(t,y')\times
	\\
	&\quad\quad\quad\quad\quad\times\langle\Phi,a^*_xa^*_ya_{x'}\sqrt{N-\cN}(N-\cN-1)\Phi\rangle \d x dy\d x '\d y'\Big| \\
	&\le \frac{2}{N^2} \iiint \int|u(t,z)|^2V_N(x-y,x-z)\d zq(t,x)q(t,x')|u(t,y)|\times
	\\
	&\quad\quad\quad\quad\quad\times\|a_xa_y\sqrt{N-\cN+1}\Phi\|\|a_{x'}(N-\cN-1)\Phi\| \d x dy\d x ' \\
	&\le \frac{2}{N^2}\left(\iiiint |q(t,x)|^2|u(t,z)|^2|V_N(x-y,x-z)||u(t,y)|^2\|a_{x'}(N-\cN-1)\Phi\|^2\d x dy\d x '\d z\right)^\frac{1}{2}\times\\
	&\quad\quad\quad\times\left(\iiiint|q(t,x')|^2|u(t,z)|^2|V_N(x-y,x-z)|\|a_xa_y\sqrt{N-\cN+1}\Phi\|^2\d x dy\d x '\d z\right)^\frac{1}{2} \\
	&\le C_t\langle\Phi,\cN\Phi\rangle^\frac{1}{2}\langle\Phi, R_{4,2}\Phi\rangle^\frac{1}{2}	 \le C_t \sqrt{mN^{\beta-1}} \langle\Phi,\d\Gamma(1-\Delta)\Phi\rangle .
	\end{align*}
	Here we have used the bound on $R_{4,2}$ in \eqref{eq:bound-R42}. Collecting all above estimates, we obtain
	\begin{align*}
	|\langle\Phi,\partial_tR_{3,2}\Phi\rangle| \le C_t \sqrt{mN^{\beta-1}} \langle\Phi,\d\Gamma(1-\Delta)\Phi\rangle.
	\end{align*}
	Thus
	\begin{align*}
	\pm \partial_t (R_{3,2} +R_{3,2}^*) \le C_t \sqrt{mN^{\beta-1}}  \d\Gamma(1-\Delta).
	\end{align*}

	\bigskip
	
	\noindent
	$\boxed{R_{3,3}}$ Let $\Phi \in \cF^{\le m}_+(t)$. We have that
	\begin{align*}
	&\langle\Phi,\partial_tR_{3,3}\Phi\rangle  \\
	&=\frac{1}{N^2} \idotsint\Big[(Q(t)\otimes Q(t)\otimes 1 V_N 1\otimes 1 \otimes Q(t))(x,y,z;x',y',z') \partial_t\left(\overline{u(t,z)}u(t,x')u(t,y')\right) + \\
	&\quad\quad\quad\quad+\Big(\partial_t(Q(t)\otimes Q(t)\otimes 1 V_N 1\otimes 1 \otimes Q(t))\Big)(x,y,z;x',y',z') \overline{u(t,z)}u(t,x')u(t,y')\Big] \times \\
	&\quad\quad\quad\quad\quad\quad\quad\times \langle\Phi,a^*_xa^*_ya_{z'}\sqrt{N -\cN}(N-\cN-1)\Phi\rangle \d x \d y \d z\d x '\d y'\d z'.
	\end{align*}
	The first term involving $\partial_t\left(\overline{u(t,z)}u(t,x')u(t,y')\right)$ can be estimated as in $R_{3,3}$:
	\begin{align*}
	&\frac{1}{N^2} \Big|\idotsint(Q(t)\otimes Q(t)\otimes 1 V_N 1\otimes 1 \otimes Q(t))(x,y,z;x',y',z') \partial_t\left(\overline{u(t,z)}u(t,x')u(t,y')\right)  \times \\
	&\quad\quad\quad\quad\quad\quad\quad\times \langle\Phi,a^*_xa^*_ya_{z'}\sqrt{N -\cN}(N-\cN-1)\Phi\rangle \d x \d y \d z\d x '\d y'\d z'\Big| \\
	&=\frac{1}{N^2} \Big|\iiint V_N(x-y,x-z)\Big[\overline{\partial_tu(t,z)}u(t,x)u(t,y) +2\overline{u(t,z)}u(t,x)\partial_tu(t,y)\Big] \times \\
	&\quad\quad\quad\quad\quad\quad\quad\times \langle\Phi,a^*_xa^*_ya_{z}\sqrt{N -\cN}(N-\cN-1)\Phi\rangle \d x \d y \d z\Big| \\
	&\le \frac{C_t}{N^2}\left(\iiint V_N(x-y,x-z)\|a_{z}(N-\cN-1)\Phi\|^2\d x\d y\d z\right)^\frac{1}{2}\times
	\\
	&\quad\quad\quad\quad\quad\times\left(\iiint V_N(x-y,x-z)\|a_xa_y\sqrt{N -\cN+1}\Phi\|^2\d x\d y\d z\right)^\frac{1}{2} \\
	&\le C_t \sqrt{mN^{\beta-1}} \langle\Phi,\d\Gamma(1-\Delta)\Phi\rangle .
	\end{align*}
	Let us now decompose $\partial_t(Q(t)\otimes Q(t)\otimes 1 V_N 1\otimes 1 \otimes Q(t))$ into three terms. For the first one containing $(\partial_tQ(t)\otimes Q(t)\otimes 1 V_N 1\otimes 1 \otimes Q(t))$ we get
	\begin{align*}
	&\frac{1}{N^2} \Big|\idotsint(\partial_tQ(t)\otimes Q(t)\otimes 1 V_N 1\otimes 1 \otimes Q(t))(x,y,z;x',y',z') \overline{u(t,z)}u(t,x')u(t,y')  \times \\
	&\quad\quad\quad\quad\quad\quad\quad\times \langle\Phi,a^*_xa^*_ya_{z'}\sqrt{N -\cN}(N-\cN-1)\Phi\rangle \d x \d y \d z\d x '\d y'\d z'\Big| \\
	&=\frac{1}{N^2}\Big|\idotsint(\partial_tQ(t))(x;x')V_N(x'-y,x'-z)\delta(y-y')\delta(z-z')\overline{u(t,z)}u(t,x')u(t,y')\times \\
	&\quad\quad\quad\quad\quad\quad\quad\times \langle\Phi,a^*_xa^*_ya_{z'}\sqrt{N -\cN}(N-\cN-1)\Phi\rangle \d x \d y \d z\d x '\d y'\d z'\Big| \\
	&\le \frac{1}{N^2} \iiiint q(t,x)q(t,x')|V_N(x'-y,x'-z)||u(t,z)||u(t,x')||u(t,y)|\times\\
	&\quad\quad\quad\quad\quad\quad\quad\times\|a_xa_y\sqrt{N-\cN+1}\Phi\|\|a_z(N-\cN-1)\Phi\| \d x \d y \d z\d x ' \\
	&\le \frac{C_t}{N^2}\left(\iiiint |V_N(x'-y,x'-z)| \|a_xa_y\sqrt{N-\cN+1}\Phi\|^2 \d x \d y \d z\d x '\right)^\frac{1}{2}\times\\ &\quad\quad\quad\quad\quad\times\left(\iiiint|q(t,x)|^2 |V_N(x'-y,x'-z)| \|a_z(N-\cN-1)\Phi\|^2\d x \d y \d z\d x '\right)^\frac{1}{2} \\
	&\le \frac{C_t}{N^2}\left(\iint\|a_xa_y\sqrt{N-\cN+1}\Phi\|^2\d x dy\right)^\frac{1}{2}\left(\iint \|a_z(N-\cN-1)\Phi\|^2\d z\right)^\frac{1}{2} \\
	&\le C_t \sqrt{mN^{-1}}\langle\Phi,(\cN+1)\Phi\rangle.
	\end{align*}
	The term involving $(Q(t)\otimes \partial_tQ(t)\otimes 1 V_N 1\otimes 1 \otimes Q(t))$ can be treated similarly.
	The third term $(Q(t)\otimes Q(t)\otimes 1 V_N 1\otimes 1 \otimes\partial_t Q(t))$ goes as follows
	\begin{align*}
	&\frac{1}{N^2} \Big|\idotsint(Q(t)\otimes Q(t)\otimes 1 V_N 1\otimes 1 \otimes \partial_tQ(t))(x,y,z;x',y',z') \overline{u(t,z)}u(t,x')u(t,y')  \times \\
	&\quad\quad\quad\quad\quad\quad\quad\times \langle\Phi,a^*_xa^*_ya_{z'}\sqrt{N -\cN}(N-\cN-1)\Phi\rangle \d x \d y \d z\d x '\d y'\d z'\Big| \\
	&=\frac{1}{N^2}\Big|\idotsint V_N(x-y,x-z)(\partial_tQ(t))(z;z')\delta(x-x')\delta(y-y')\overline{u(t,z)}u(t,x')u(t,y')\times \\
	&\quad\quad\quad\quad\quad\quad\quad\times \langle\Phi,a^*_xa^*_ya_{z'}\sqrt{N -\cN}(N-\cN-1)\Phi\rangle \d x \d y \d z\d x '\d y'\d z'\Big| \\
	&\le \frac{1}{N^2}\iiiint |V_N(x-y,x-z)|q(t,z)q(t,z')|u(t,z)||u(t,x)||u(t,y)|\times\\
	&\quad\quad\quad\quad\quad\quad\quad\times\|a_xa_y\sqrt{N-\cN+1}\|\|a_{z'}(N-\cN-1)\Phi\|\d x \d y \d z\d z' \\
	&\le \frac{C_t}{N^2}\left(\iiiint|q(t,z')|^2|V_N(x-y,x-z)||u(t,z)|^2\|a_xa_y\sqrt{N-\cN+1}\|^2\d x \d y \d z\d z'\right)^\frac{1}{2} \times \\
	&\quad\quad\quad\times\left(\iiiint |V_N(x-y,x-z)||u(t,x)|^2\|a_{z'}(N-\cN-1)\Phi\|^2\d x \d y \d z\d z'\right)^\frac{1}{2} \\
	&\le\frac{C_t}{N^2}\left(\iiint|V_N(x-y,x-z)||u(t,z)|^2\|a_xa_y\sqrt{N-\cN+1}\|^2\d x \d y \d z\right)^\frac{1}{2}\times\\
	&\quad\quad\quad\quad\times\langle\Phi,\cN(N-\cN-1)^2\Phi\rangle^\frac{1}{2} \\
	&\le C_t \sqrt{mN^{\beta-1}} \langle \Phi,\d\Gamma(1-\Delta)\Phi\rangle.
	\end{align*}
	Putting everything together we obtain
	\begin{align*}
	|\langle\Phi,\partial_tR_{3,3}\Phi\rangle| \le C_t \sqrt{mN^{\beta-1}} \langle \Phi,\d\Gamma(1-\Delta)\Phi\rangle.
	\end{align*}
	Thus
	$$
	\pm \partial_t (R_{3,3}+R_{3,3}^*) \le C_t \sqrt{mN^{\beta-1}}  \d\Gamma(1-\Delta). 
	$$
	
	\bigskip
	
	\noindent
	$\boxed{j=4}$ $\boxed{R_{4,1}}$ Let $\Phi \in \cF^{\le m}_+(t)$. Then we have
	\begin{align*}
	&\langle\Phi,\partial_tR_{4,1} \Phi\rangle \\
	 &= \frac{1}{N^2}\idotsint\Big[(Q(t)\otimes Q(t)\otimes Q(t)V_NQ(t)\otimes 1\otimes 1)(x,y,z;x',y',z')\partial_t\left(u(t,y')u(t,z')\right) \\
	&\quad\quad\quad\quad+\Big(\partial_t(Q(t)\otimes Q(t)\otimes Q(t)V_NQ(t)\otimes 1\otimes 1)\Big)(x,y,z;x',y',z')u(t,y')u(t,z')\Big]\times \\
	&\quad\quad\quad\quad\quad\quad\quad\quad \times \langle\Phi, a^*_xa^*_ya^*_za_{x'}\sqrt{N-\cN -1}\sqrt{N-\cN}\Phi\rangle \d x \d y \d z\d x '\d y'\d z'.
	\end{align*}
	The  term containing $\partial_t\left(u(t,y')u(t,z')\right)$ can be estimated similarly to $R_{4,1}$. 
For the other terms we decompose $\partial_t(Q(t)\otimes Q(t)\otimes Q(t)V_NQ(t)\otimes 1\otimes 1)$ into four terms. For the first term $(\partial_tQ(t)\otimes Q(t)\otimes Q(t)V_NQ(t)\otimes 1\otimes 1)$ we have
	\begin{align*}
	&\frac{1}{N^2}\Big|\idotsint(\partial_tQ(t)\otimes Q(t)\otimes Q(t)V_NQ(t)\otimes 1\otimes 1)(x,y,z;x',y',z')u(t,y')u(t,z')\times \\
	&\quad\quad\quad\quad\quad\quad\quad\quad \times \langle\Phi, a^*_xa^*_ya^*_za_{x'}\sqrt{N-\cN -1}\sqrt{N-\cN}\Phi\rangle \d x \d y \d z\d x '\d y'\d z'\Big| \\
	&=\frac{1}{N^2} \Big|\idotsint (\partial_tQ(t))(x;x')V_N(x'-y,x'-z)\delta(y-y')\delta(z-z')u(t,y')u(t,z')\times \\
	&\quad\quad\quad\quad\quad\quad\quad\quad \times \langle\Phi, a^*_xa^*_ya^*_za_{x'}\sqrt{N-\cN -1}\sqrt{N-\cN}\Phi\rangle \d x \d y \d z\d x '\d y'\d z'\Big| \\
	&\le\frac{1}{N^2}\iiiint q(t,x)q(t,x')|V_N(x'-y,x'-z)||u(t,y)||u(t,z)|\times \\
	&\quad\quad\quad\quad\quad\quad\quad\quad \times\|a_xa_ya_z\Phi\|\|a_{x'}\sqrt{N-\cN-1}\sqrt{N-\cN}\Phi\|\d x \d y \d z\d x ' \\ 
	&\le \frac{1}{N^2} \|u(t,\cdot)\|^2_{L^\infty}\left(\iiiint |q(t,x')|^2V_N(x'-y,x'-z)\|a_xa_ya_z\Phi\|^2\d x \d y \d z\d x '\right)^\frac{1}{2}\times \\
	&\quad\quad\quad\quad \times\left(\iiiint |q(t,x)|^2|V_N(x'-y,x'-z)|\|a_{x'}\sqrt{N-\cN-1}\sqrt{N-\cN}\Phi\|^2\d x \d y \d z\d x '\right)^\frac{1}{2} \\
	&\le C_tmN^{\beta-1} \langle\Phi,\d\Gamma(-\Delta)\Phi\rangle^\frac{1}{2}\langle\Phi,\cN\Phi\rangle^\frac{1}{2} \\
	&\le C_t mN^{\beta-1} \langle\Phi,\d\Gamma(1-\Delta)\Phi\rangle.
	\end{align*}
	The second one containing $(Q(t)\otimes \partial_tQ(t)\otimes Q(t)V_NQ(t)\otimes 1\otimes 1)$ can be bounded as
	\begin{align*}
	&\frac{1}{N^2}\Big|\idotsint(Q(t)\otimes \partial_tQ(t)\otimes Q(t)V_NQ(t)\otimes 1\otimes 1)(x,y,z;x',y',z')u(t,y')u(t,z')\times \\
	&\quad\quad\quad\quad\quad\quad\quad\quad \times \langle\Phi, a^*_xa^*_ya^*_za_{x'}\sqrt{N-\cN -1}\sqrt{N-\cN}\Phi\rangle \d x \d y \d z\d x '\d y'\d z'\Big| \\
	&=\frac{1}{N^2} \Big|\idotsint (\partial_tQ(t))(y;y')V_N(x-y',x-z)\delta(x-x')\delta(z-z')u(t,y')u(t,z')\times \\
	&\quad\quad\quad\quad\quad\quad\quad\quad \times \langle\Phi, a^*_xa^*_ya^*_za_{x'}\sqrt{N-\cN -1}\sqrt{N-\cN}\Phi\rangle \d x \d y \d z\d x '\d y'\d z'\Big| \\
	&\le\frac{1}{N^2}\iiiint q(t,y)q(t,y')|V_N(x-y',x-z)||u(t,y')||u(t,z)|\times \\
	&\quad\quad\quad\quad\quad\quad\quad\quad \times\|a_xa_ya_z\Phi\|\|a_x\sqrt{N-\cN-1}\sqrt{N-\cN}\Phi\|\d x \d y \d z\d y' \\ 
	&\le \frac{C_t}{N^2}\left(\iiiint |u(t,y')|^2V_N(x-y',x-z)\|a_xa_ya_z\Phi\|^2\d x \d y \d z\d y '\right)^\frac{1}{2}\times \\
	&\quad\quad\quad\quad \times\left(\iiiint |q(t,y)|^2|V_N(x-y',x-z)|\|a_x\sqrt{N-\cN-1}\sqrt{N-\cN}\Phi\|^2\d x \d y \d z\d y '\right)^\frac{1}{2} \\
	&\le C_t mN^{\beta-1} \langle\Phi,\d\Gamma(-\Delta)\Phi\rangle^\frac{1}{2}\langle\Phi,\cN\Phi\rangle^\frac{1}{2} \\
	&\le C_tmN^{\beta-1} \langle\Phi,\d\Gamma(1-\Delta)\Phi\rangle.
	\end{align*}
	The third term $(Q(t)\otimes Q(t)\otimes \partial_tQ(t)V_NQ(t)\otimes 1\otimes 1)$ can be treated similarly. For the fourth term $(Q(t)\otimes Q(t)\otimes Q(t)V_N\partial_tQ(t)\otimes 1\otimes 1)$ we estimate
	\begin{align*}
	&\frac{1}{N^2}\Big|\idotsint(Q(t)\otimes Q(t)\otimes Q(t)V_N\partial_tQ(t)\otimes 1\otimes 1)(x,y,z;x',y',z')u(t,y')u(t,z')\times \\
	&\quad\quad\quad\quad\quad\quad\quad\quad \times \langle\Phi, a^*_xa^*_ya^*_za_{x'}\sqrt{N-\cN -1}\sqrt{N-\cN}\Phi\rangle \d x \d y \d z\d x '\d y'\d z'\Big| \\
	&=\frac{1}{N^2} \Big|\idotsint V_N(x-y,x-z)(\partial_tQ(t))(x;x')\delta(y-y')\delta(z-z')u(t,y')u(t,z')\times \\
	&\quad\quad\quad\quad\quad\quad\quad\quad \times \langle\Phi, a^*_xa^*_ya^*_za_{x'}\sqrt{N-\cN -1}\sqrt{N-\cN}\Phi\rangle \d x \d y \d z\d x '\d y'\d z'\Big| \\
	&\le\frac{1}{N^2}\iiiint V_N(x-y,x-z)|q(t,x)q(t,x')u(t,y)u(t,z)|\times \\
	&\quad\quad\quad\quad\quad\quad\quad\quad \times\|a_xa_ya_z\Phi\|\|a_{x'}\sqrt{N-\cN-1}\sqrt{N-\cN}\Phi\|\d x \d y \d z\d x ' \\ 
	&\le \frac{1}{N^2} \|u(t,\cdot)\|^2_{L^\infty} \|q(t,\cdot)\|_{L^\infty}^2 \left(\iiiint |q(t,x')|^2V_N(x-y,x-z)\|a_xa_ya_z\Phi\|^2\d x \d y \d z\d x '\right)^\frac{1}{2}\times \\
	&\quad\quad\quad\quad \times\left(\iiiint |q(t,x)|^2|V_N(x-y,x-z)|\|a_{x'}\sqrt{N-\cN-1}\sqrt{N-\cN}\Phi\|\d x \d y \d z\d x '\right)^\frac{1}{2} \\
	&\le \frac{C_t}{N^2} \left(\iiint V_N(x-y,x-z)\|a_xa_ya_z\Phi\|^2 \right)^\frac{1}{2}\langle\Phi,\cN(N-\cN-1)(N-\cN)\Phi\rangle^\frac{1}{2}\\
	&\le C_t \langle\Phi,R_6\Phi\rangle^{1/2} \langle\Phi,\cN\Phi\rangle^{1/2} \le C_t mN^{2\beta-1}  \langle\Phi,\dGamma(1-\Delta) \Phi\rangle .
	\end{align*}
	In the last estimate, we have used the bound on $R_6$ in \eqref{eq:bound-R6}. Hence, we conclude that
	\begin{align*}
	|\langle\Phi,\partial_tR_{4,1}\Phi\rangle| \le C_t mN^{ 2\beta-1}\langle \Phi,  \d\Gamma(1-\Delta) \Phi\rangle.	\end{align*}
	This means
	$$
	\pm \partial_t (R_{4,1}+R_{4,1}^*) \le C_t mN^{2\beta-1} \d\Gamma(1-\Delta).
	$$
	
	\bigskip
	\noindent
	$\boxed{R_{4,2}}$ Let $\Phi \in \cF^{\le m}_+(t)$. Then we have
	\begin{align*}
	\langle\Phi,\partial_tR_{4,2}\Phi\rangle &= \frac{1}{2N^2} \idotsint \Big[ (Q(t)\otimes Q(t)  V_NQ(t)\otimes Q(t))(x,y;x',y')\partial_t|u(t,z)|^2 + \\
	&\quad\quad\quad\quad\quad\quad\quad+ \partial_t(Q(t)\otimes Q(t) V_NQ(t)\otimes Q(t))(x,y,;x',y')|u(t,z)|^2\Big]\times \\
	&\quad\quad\quad\quad\quad\quad\quad\quad\times \langle\Phi, a^*_xa^*_ya_{x'}a_{y'}(N-\cN)\Phi\rangle \d x dy\d x '\d y'\d z.
	\end{align*}
	The first term containing $\partial_t|u(t,z)|^2$ can be estimated as 
	\begin{align*}
	&\frac{1}{2N^2} \Big|\iiint \partial_t|u(t,z)|^2 V_N(x-y,x-z) \langle\Phi, a^*_xa^*_ya_{x}a_{y}(N-\cN)\Phi\rangle \d x \d y \d z\Big| \\ 
	&\le C_tN^{\beta-1}{m}\d\Gamma(-\Delta),
	\end{align*}
	where we have used the Sobolev bound. 	For the other terms we decompose $\partial_t(Q(t)\otimes Q(t)V_NQ(t)\otimes Q(t))$ into four terms. The first one containing $(\partial_tQ(t)\otimes Q(t)V_NQ(t)\otimes Q(t))$ can be handled as
	\begin{align*}
	&\frac{1}{2N^2}\Big|\iiiint (\partial_tQ(t)\otimes Q(t) \int |u(t)|^2 V_N \d z Q(t)\otimes Q(t))(x,y;x',y')\times\\
	&\quad\quad\quad\quad\quad\times\langle\Phi, a^*_xa^*_ya_{x'}a_{y'}(N-\cN)\Phi\rangle \d x dy\d x '\d y'\Big| \\
	&=\frac{1}{2N^2} \iiiint (\partial_tQ(t))(x;x') \int |u(t,z)|^2 V_N(x'-y,x'-z)\d z \delta(y-y')\times\\
	&\quad\quad\quad\quad\quad\times \langle\Phi, a^*_xa^*_ya_{x'}a_{y'}(N-\cN)\Phi\rangle \d x dy\d x '\d y'\Big| \\
	&\le \frac{1}{2N^2} \iiiint q(t,x)q(t,x') |u(t,z)|^2 V_N(x'-y,x'-z)\times\\
	&\quad\quad\quad\quad\quad\times \|a_xa_y \sqrt{N-\cN}\|\|a_{x'}a_y\sqrt{N-\cN}\| \d x dy\d z \d x '\\
	&\le \frac{C_t}{N^2} \left(\iiiint  V_N(x'-y,x'-z) \|a_xa_y\sqrt{N-\cN}\|^2\d x \d y \d z\d x '\right)^\frac{1}{2}\times \\
	&\quad\quad\quad\quad\quad \times\left(\iiiint |q(t,x)|^2|u(t,z)|^2 V_N(x'-y,x'-z) \|a_{x'}a_y\sqrt{N-\cN}\|^2\d x\d y \d z\d x '\right)^\frac{1}{2}\\
	&\le \frac{C_t}{N}\langle\Phi \cN^2(N-\cN)\Phi\rangle^\frac{1}{2}\langle\Phi,R_{4,2}\Phi\rangle \\
	&\le C_t \sqrt{mN^{\beta-1}} \langle\Phi,\d\Gamma(1-\Delta)\Phi\rangle.
	\end{align*}
	Here we have used the previous bound on $R_{4,2}$. 	The other three terms can be estimated in the same way. With this we can conclude
	\begin{align*}
	|\langle\Phi,\partial_tR_{4,2}\Phi\rangle| \le C_t \sqrt{mN^{\beta-1}} \langle\Phi,\d\Gamma(1-\Delta)\Phi\rangle.
	\end{align*}
	Thus
	$$
	\pm \partial_t (R_{4,2}+R_{4,2}^*) \le  C_t \sqrt{mN^{\beta-1}} \d\Gamma(1-\Delta).
	$$
		
	\bigskip
	
	\noindent
	$\boxed{R_{4,3}}$ Let $\Phi \in \cF^{\le m}_+(t)$. Then we have
	\begin{align*}
	&\langle\Phi,\partial_tR_{4,3}\Phi\rangle \\
	&= \frac{1}{N^2}\idotsint\Big[ (Q(t)\otimes Q(t)\otimes 1 V_N Q(t)\otimes 1\otimes Q(t))(x,y,z;x'y',z')\partial_t\left(\overline{u(t,z)}u(t,y')\right)+ \\
	&\quad\quad\quad\quad\quad\quad\quad+\partial_t(Q(t)\otimes Q(t)\otimes 1 V_N Q(t)\otimes 1\otimes Q(t))(x,y,z;x'y',z')\overline{u(t,z)}u(t,y')\Big]\times\\
	&\quad\quad\quad\quad\quad\quad\quad\quad \times \langle\Phi,a^*_xa^*_ya_{x'}a_{z'}(N-\cN)\Phi\rangle \d x \d y \d z\d x '\d y'\d z'.
	\end{align*}
	For the first term, which contains $\partial_t\left(\overline{u(t,z)}u(t,y')\right)$, we get
	\begin{align*}
	&\frac{1}{N^2}\Big|\idotsint(Q(t)\otimes Q(t)\otimes 1V_NQ(t)\otimes 1\otimes Q(t)(x,y,z;x',y',z')\partial_t\left(\overline{u(t,z)}u(t,y')\right)\times \\
	&\quad\quad\quad\quad\quad\quad\quad\quad \times \langle\Phi,a^*_xa^*_ya_{x'}a_{z'}(N-\cN)\Phi\rangle \d x \d y \d z\d x '\d y'\d z'\Big|\\
	&\le \frac{2}{N^2} \left(\iiint|\partial_t u(t,z)|^2V_N(x-y,x-z)\|a_xa_y\sqrt{N-\cN}\Phi\|^2\d x \d y \d z\right)^\frac{1}{2}\times \\
	&\times \left(\iiint| u(t,y)|^2V_N(x-y,x-z)\|a_xa_z\sqrt{N-\cN}\Phi\|^2\d x \d y \d z\right)^\frac{1}{2} \\ 
	&\le C_t mN^{\beta-1}\langle\Phi,\d\Gamma(-\Delta)\Phi\rangle.
	\end{align*}
	Let us now decompose $\partial_t(Q(t)\otimes Q(t)\otimes 1 V_N Q(t)\otimes 1\otimes Q(t))$ into four terms. For the first term $\left(\partial_tQ(t)\otimes  Q(t) \otimes 1 V_NQ(t) \otimes 1 \otimes Q(t)\right)$ we get
	\begin{align*}
	&\frac{1}{N^2}\Big|\idotsint \left(\partial_tQ(t)\otimes  Q(t) \otimes 1 V_NQ(t) \otimes 1 \otimes Q(t)\right)(x,y,z;x',y',z')\overline{u(t,z)}u(t,y') \times \\
	&\quad\quad\quad \times \langle\Phi,a^*_xa^*_ya_{x'}a_{z'}(N-\cN)\Phi\rangle \d x \d y \d z\d x '\d y'\d z'\Big| \\
	&=\frac{1}{N^2}\Big|\idotsint (\partial_tQ(t))(x;x')V_N(x'-y,x'-z)\delta(y-y')\delta(z-z')\overline{u(t,z)}u(t,y')\times \\
	&\quad\quad\quad \times \langle\Phi,a^*_xa^*_ya_{x'}a_{z'}(N-\cN)\Phi\rangle \d x \d y \d z\d x '\d y'\d z'\Big| \\
	&\le \frac{1}{N^2} \iiiint q(t,x)q(t,x')V_N(x'-y,x'-z)|u(t,z)||u(t,y)|\times\\
	&\quad\quad\quad\quad\quad\quad\times\|a_xa_y\sqrt{N-\cN}\Phi\|\|a_{x'}a_z\sqrt{N-\cN}\Phi\| \d x \d y \d z\d x '\\
	& \le 	\frac{C_t}{N^2}\langle\Phi,\cN^2(N-\cN)\Phi\rangle^\frac{1}{2}\left(\iiint V_N(x'-y,x'-z)|u(t,y)|^2\|a_{x'}a_z\sqrt{N-\cN}\Phi\|^2 \d x '\d y \d z\right)^\frac{1}{2} \\
	&\le C_t \sqrt{mN^{\beta-1}} \langle\Phi,\cN\Phi\rangle^{1/2} \langle\Phi,\d\Gamma(-\Delta)\Phi\rangle^{1/2}\le C_t \sqrt{mN^{\beta-1}} \langle\Phi,\d\Gamma(1-\Delta)\Phi\rangle.
	\end{align*}
	The second term $\left(Q(t)\otimes  \partial_tQ(t) \otimes 1 V_NQ(t) \otimes 1 \otimes Q(t)\right)$ goes as 
	\begin{align*}
		&\frac{1}{N^2}\Big|\idotsint \left(Q(t)\otimes  \partial_tQ(t) \otimes 1 V_NQ(t) \otimes 1 \otimes Q(t)\right)(x,y,z;x',y',z')\overline{u(t,z)}u(t,y') \times \\
	&\quad\quad\quad \times \langle\Phi,a^*_xa^*_ya_{x'}a_{z'}(N-\cN)\Phi\rangle \d x \d y \d z\d x '\d y'\d z'\Big| \\
	&=\frac{1}{N^2}\Big|\idotsint (\partial_tQ(t))(y;y')V_N(x-y',x-z)\delta(x-x')\delta(z-z')\overline{u(t,z)}u(t,y')\times \\
	&\quad\quad\quad \times \langle\Phi,a^*_xa^*_ya_{x'}a_{z'}(N-\cN)\Phi\rangle \d x \d y \d z\d x '\d y'\d z'\Big| \\
	&\le \frac{1}{N^2} \iiiint q(t,y)q(t,y')V_N(x-y',x-z)|u(t,z)||u(t,y')|\times\\
	&\quad\quad\quad\quad\quad\quad\times\|a_xa_y\sqrt{N-\cN}\Phi\|\|a_xa_z\sqrt{N-\cN}\Phi\| \d x \d y \d z\d y'\\
	& \le
	\frac{C_t}{N^2}\langle\Phi,\cN^2(N-\cN)\Phi\rangle^\frac{1}{2}\left(\iiint V_N(x-y',x-z)|u(t,y')|^2\|a_xa_z\sqrt{N-\cN}\Phi\|^2 \d x \d y'\d z\right)^\frac{1}{2}
	\\
	&\le C_t \sqrt{mN^{\beta-1}} \langle\Phi,\d\Gamma(1-\Delta)\Phi\rangle.
	\end{align*}
	For the third term $\left(Q(t)\otimes  Q(t) \otimes 1 V_N\partial_tQ(t) \otimes 1 \otimes Q(t)\right)$ we see that
	\begin{align*}
	&\frac{1}{N^2}\Big|\idotsint \left(Q(t)\otimes  Q(t) \otimes 1 V_N\partial_tQ(t) \otimes 1 \otimes Q(t)\right)(x,y,z;x',y',z')\overline{u(t,z)}u(t,y') \times \\
	&\quad\quad\quad \times \langle\Phi,a^*_xa^*_ya_{x'}a_{z'}(N-\cN)\Phi\rangle \d x \d y \d z\d x '\d y'\d z'\Big| \\
	&=\frac{1}{N^2}\Big|\idotsint V_N(x-y,x-z)(\partial_tQ(t))(x;x')\delta(y-y')\delta(z-z')\overline{u(t,z)}u(t,y')\times \\
	&\quad\quad\quad \times \langle\Phi,a^*_xa^*_ya_{x'}a_{z'}(N-\cN)\Phi\rangle \d x \d y \d z\d x '\d y'\d z'\Big| \\
	&\le \frac{1}{N^2} \iiiint V_N(x-y,x-z)q(t,x)q(t,x')|u(t,z)||u(t,y)|\times\\
	&\quad\quad\quad\quad\quad\quad\times\|a_xa_y\sqrt{N-\cN}\Phi\|\|a_{x'}a_z\sqrt{N-\cN}\Phi\| \d x \d y \d z\d x '\\
	& \le
	\frac{C_t}{N^2}\left(\iiint V_N(x-y,x-z)|u(t,z)|^2\|a_xa_y\sqrt{N-\cN}\Phi\|^2 \d x \d y \d z\right)^\frac{1}{2}\langle\Phi,\cN^2(N-\cN)\Phi\rangle^\frac{1}{2} 	\\
	&\le C_t \sqrt{mN^{\beta-1}} \langle\Phi,\d\Gamma(1-\Delta)\Phi\rangle.	
	\end{align*}
	The fourth term containing $\left(Q(t)\otimes  Q(t) \otimes 1 V_N\partial_tQ(t) \otimes 1 \otimes Q(t)\right)$ can be estimated similarly.
	Hence, we conclude
	\begin{align*}
	|\langle\Phi,\partial_tR_{4,3}\Phi\rangle| \le C_t \sqrt{mN^{\beta-1}} \langle\Phi,\d\Gamma(1-\Delta)\Phi\rangle.
	\end{align*}
	Thus
	\begin{align*}
	\pm \partial_t (R_{4,3} +R_{4,3}^*)  \le C_t \sqrt{mN^{\beta-1}} \d\Gamma(1-\Delta).
	\end{align*}

	\bigskip
	
	\noindent
	$\boxed{j=5}$ For $\Phi \in \cF^{\le m}_+(t)$ we have that
	\begin{align*}
	\langle\Phi,\partial_tR_5\Phi\rangle &=
	\frac{1}{N^2}\idotsint\Big[(Q(t)\otimes Q(t)\otimes Q(t)V_NQ(t)\otimes Q(t)\otimes 1)(x,y,z;x',y',z')\partial_tu(t,z')+\\
	&\quad\quad\quad+\Big(\partial_t(Q(t)\otimes Q(t)\otimes Q(t)V_NQ(t)\otimes Q(t)\otimes 1)\Big)(x,y,z;x',y',z')u(t,z')\Big]\times \\
	&\quad\quad\quad\quad\quad\quad\quad\quad \times \langle\Phi,a^*_xa^*_ya^*_za_{x'}a_{y'}\sqrt{N-\cN -2}\Phi\rangle\d x \d y \d z\d x '\d y'\d z'.
	\end{align*}
	The term containing $\partial_tu(t,z')$ can be estimated as
		\begin{align*}
	&\frac{1}{N^2}\Big|\idotsint(Q(t)\otimes Q(t)\otimes Q(t)V_NQ(t)\otimes Q(t)\otimes 1)(x,y,z;x',y',z')\partial_tu(t,z')\times \\
	&\quad\quad\quad\quad\quad\quad\quad\quad \times \langle\Phi,a^*_xa^*_ya^*_za_{x'}a_{y'}\sqrt{N-\cN-2}\Phi\rangle \d x \d y \d z\d x '\d y'\d z'\Big| \\
	&\le\frac{1}{N^2}\left(\iiint V_N(x-y,x-z)\|a_xa_ya_z\Phi\|^2\d x\d y\d z\right)^\frac{1}{2}\times\\
	&\quad\quad\times\left(\iiint V_N(x-y,x-z)|\partial_tu(t,z)|^2\|a_xa_y\sqrt{N-\cN-2}\Phi\|^2\d x\d y\d z\right)^\frac{1}{2} \\
	&\le C_t \sqrt{mN^{\beta-1}} \langle\Phi,R_6\Phi\rangle^{1/2} \langle\Phi,\d\Gamma(1-\Delta)\Phi\rangle^{1/2}\\
	&\le C_t \sqrt{m^3 N^{5\beta-3}} \langle\Phi,\d\Gamma(1-\Delta)\Phi\rangle. 
	\end{align*}
	Here we have used the bound on $R_6$ in \eqref{eq:bound-R6} in the last estimate.
	
	For the other terms we decompose  $\partial_t(Q(t)\otimes Q(t)\otimes Q(t)V_NQ(t)\otimes Q(t)\otimes 1)$ into five terms. The first one containing $(\partial_tQ(t)\otimes Q(t)\otimes Q(t)V_NQ(t)\otimes Q(t)\otimes 1)$ can be bounded as
	\begin{align*}
	&\frac{1}{N^2}\Big|\idotsint(\partial_tQ(t)\otimes Q(t)\otimes Q(t)V_NQ(t)\otimes Q(t)\otimes 1)(x,y,z;x',y',z')u(t,z')\times \\
	&\quad\quad\quad\quad\quad\quad\quad\quad \times \langle\Phi, a^*_xa^*_ya^*_za_{x'}a_{y'}\sqrt{N-\cN -2}\Phi\rangle \d x \d y \d z\d x '\d y'\d z'\Big| \\
	&=\frac{1}{N^2} \Big|\idotsint(\partial_tQ(t))(x;x') V_N(x'-y,x'-z)\delta(y-y')\delta(z-z')u(t,z')\times \\
	&\quad\quad\quad\quad\quad\quad\quad\quad \times \langle\Phi, a^*_xa^*_ya^*_za_{x'}a_{y'}\sqrt{N-\cN -2}\Phi\rangle \d x \d y \d z\d x '\d y'\d z'\Big| \\
	&\le\frac{1}{N^2}\iiiint q(t,x)q(t,x')|V_N(x'-y,x'-z)||u(t,z)|\times \\
	&\quad\quad\quad\quad\quad\quad\quad\quad \times\|a_xa_ya_z\Phi\|\|a_{x'}a_y\sqrt{N-\cN-2}\Phi\|\d x \d y \d z\d x ' \\
	&\le \frac{C_t}{N^2}\left(\iiiint V_N(x'-y,x'-z)\|a_xa_ya_z\Phi\|^2\d x \d y \d z\d x '\right)^\frac{1}{2}\times\\
	&\quad\quad\quad\times\left(\iiiint |q(t,x)|^2V_N(x'-y,x'-z)\|a_{x'}a_y\sqrt{N-\cN-2}\Phi\|\d x \d y \d z\d x '\right)^\frac{1}{2} \\
	&\le C_tN^{\beta-1}m\langle\Phi,\d\Gamma(1-\Delta)\Phi\rangle.	\end{align*}
	The term with $(Q(t)\otimes \partial_tQ(t)\otimes Q(t)V_NQ(t)\otimes Q(t)\otimes 1)$ and $(Q(t)\otimes Q(t)\otimes \partial_tQ(t)V_NQ(t)\otimes Q(t)\otimes 1)$ can be treated similarly. The fourth term $(Q(t)\otimes Q(t)\otimes Q(t)V_N\partial_tQ(t)\otimes Q(t)\otimes 1)$ can be bounded as 
	\begin{align*}
	&\frac{1}{N^2}\Big|\idotsint(Q(t)\otimes Q(t)\otimes Q(t)V_N\partial_tQ(t)\otimes Q(t)\otimes 1)(x,y,z;x',y',z')u(t,z')\times \\
	&\quad\quad\quad\quad\quad\quad\quad\quad \times \langle\Phi, a^*_xa^*_ya^*_za_{x'}a_{y'}\sqrt{N-\cN -2}\Phi\rangle \d x \d y \d z\d x '\d y'\d z'\Big| \\
	&\le\frac{1}{N^2}\iiiint V_N(x-y,x-z)|q(t,x)q(t,x')|u(t,z)|\times \\
	&\quad\quad\quad\quad\quad\quad\quad\quad \times\|a_xa_ya_z\Phi\|\|a_{x'}a_y\sqrt{N-\cN-2}\Phi\|\d x \d y \d z\d x ' \\
	&\le \frac{C_t}{N^2}\left(\int|q(t,x')|^2\d x '\right)^\frac{1}{2}\left(\iiint V_N(x-y,x-z)\|a_xa_ya_z\Phi\|^2\d x \d y \d z\right)^\frac{1}{2} \times \\
	&\quad\quad\quad\quad\times\langle\Phi,\cN^2(N-\cN-2)\Phi\rangle^\frac{1}{2} \\
	&\le C_t \sqrt{mN^{-1}} \langle\Phi,R_6\Phi\rangle^{1/2} \langle\Phi,\cN\Phi\rangle^\frac{1}{2}  \le  C_t \sqrt{m^3 N^{4\beta-3}} \langle\Phi, \dGamma(1-\Delta)\Phi\rangle.
	\end{align*}
	The other term containing $Q(t)\otimes Q(t)\otimes Q(t)V_NQ(t)\otimes \partial_tQ(t)\otimes 1$ can be bounded in the same way. Putting everything together we get
	\begin{align*}
	|\langle\Phi,\partial_tR_5\Phi\rangle| \le C_t (mN^{\beta-1}+\sqrt{m^3 N^{5\beta-3}}) \langle\Phi,\d\Gamma(1-\Delta)\Phi\rangle.
	\end{align*}
	This means
	$$
	\partial_t (R_5+R_5^*) \le C_t (mN^{\beta-1}+\sqrt{m^3 N^{5\beta-3}})  \d\Gamma(1-\Delta).
	$$
	
	\bigskip
	
	\noindent
$\boxed{j=6}$ Let $\Phi \in \cF^{\le m}_+(t)$. Then we have
\begin{align*}
\langle\Phi,\partial_tR_6\Phi\rangle &= \frac{1}{6N^2}\idotsint\left(\partial_t(Q(t)\otimes Q(t)\otimes Q(t)V_NQ(t)\otimes Q(t)\otimes Q(t))\right)(x,y,z;x',y',z')\times \\
&\quad\quad\quad\quad\quad\quad\quad\quad \times \langle\Phi,a^*_xa^*_ya^*_za_{x'}a_{y'}a_{z'}\Phi\rangle \d x \d y \d z\d x '\d y'\d z'
\end{align*}
We can decompose $\partial_t(Q(t)\otimes Q(t)\otimes Q(t)V_NQ(t)\otimes Q(t)\otimes Q(t))$ into six terms which can be estimated all the same. Therefore, we only need to prove the bound for $\partial_tQ(t)\otimes Q(t)\otimes Q(t)V_NQ(t)\otimes Q(t)\otimes Q(t)$. We have
\begin{align*}
&\frac{1}{6N^2}\Big|\idotsint\left(\partial_tQ(t)\otimes Q(t)\otimes Q(t)V_NQ(t)\otimes Q(t)\otimes Q(t))\right)(x,y,z;x',y',z')\times \\
&\quad\quad\quad\quad\quad\quad\quad\quad \times \langle\Phi,a^*_xa^*_ya^*_za_{x'}a_{y'}a_{z'}\Phi\rangle \d x \d y \d z\d x '\d y'\d z'\Big| \\
&= \frac{1}{6N^2}\Big|\idotsint(\partial_tQ(t))(x;x')V_N(x'-y,x'-z)\delta(y-y')\delta(z-z') \times\\
&\quad\quad\quad\quad\quad\quad\quad\quad \times \langle\Phi,a^*_xa^*_ya^*_za_{x'}a_{y'}a_{z'}\Phi\rangle \d x \d y \d z\d x '\d y'\d z'\Big| \\
&\le\frac{1}{6N^2}\iiiint q(t,x)q(t,x'))V_N(x'-y,x'-z)\|a_xa_ya_z\Phi\|\|a_{x'}a_ya_z\Phi\|\d x \d y \d z\d x ' \\
&\le \frac{1}{6N^2}\|q(t,\cdot)\|_{L^\infty}\left(\iiiint V_N(x'-y,x'-z)\|a_xa_ya_z\Phi\|^2\d x \d y \d z\d x '\right)^\frac{1}{2} \times \\
&\quad\quad\quad\times\left(\iiiint |q(t,x)|^2V_N(x'-y,x'-z)\|a_{x'}a_ya_z\Phi\|^2\d x \d y \d z\d x '\right)^\frac{1}{2} \\
&\le  C_t \sqrt{mN^{\beta-1}} \langle\Phi,\d\Gamma(1-\Delta)\Phi\rangle^{1/2} \langle\Phi, R_6\Phi\rangle^{1/2} \le C_t \sqrt{m^3 N^{5\beta-3}} \langle\Phi,\d\Gamma(1-\Delta)\Phi\rangle.
\end{align*}
Here again we used the bound on $R_6$ in \eqref{eq:bound-R6}. Thus
$$
\pm \partial_t R_6 \le  C_t \sqrt{m^3 N^{5\beta-3}} \langle\Phi,\d\Gamma(1-\Delta)\Phi\rangle.
$$
This finishes the proof of \eqref{eq:errorRdt} and by that the proof of Lemma~\ref{lem:errorbounds}.
\end{proof}

\section{Proof of Theorem~\ref{thm:main1}} \label{sec:thm1}
In this section we will prove the leading order convergence stated in Theorem~\ref{thm:main1}. As in Section \ref{sec:gen-stra}, we consider  
$$\Phi_N(t)=e^{-i\int_0^t \chi(s) \d s}\,U_N(t)\Psi_N(t)=e^{-i\int_0^t \chi(s) \d s}\,U_N(t) e^{-itH_N} U_N^*(0)\Phi_{N}(0)$$
which solves \eqref{eq:eq-PhiN},
$$
i\partial_t \Phi_N(t) = \widetilde H_N(t) \Phi_N(t)= \left( \bH(t) + \frac{1}{2}\sum_{j=0}^6 (R_j + R_j^*) \right) \Phi_N(t).
$$
Note that thanks to  \eqref{eq:UN-action},   the condition  \eqref{eq:leading-order} is equivalent to 
\begin{align} \label{eq:Phi0-assume-extra}
\langle \Phi_N(0), \d\Gamma(1-\Delta)  \Phi_N(0)\rangle \le C. 
\end{align}

Let us denote by $\1^{\le m}=\1(\cN\le m)$ the projection on the truncated Fock space. From Lemmas \ref{lem:Bogoliubovbounds} and  \ref{lem:errorbounds},  for all $t>0$, $\eps>0$, $1\le m\le N$, and 
$$\eta \ge C_t \max\{ \sqrt{m N^{2\beta-1}}, \sqrt{m^3 N^{5\beta-3}}, m^2 N^{4\beta-2} \}, $$
we have
\begin{align}
\pm \1^{\le m}(\widetilde H_N(t) - \dGamma(-\Delta)) \1^{\le m} &\le  \eta  \dGamma(1-\Delta) +  \eta^{-1} C_t mN^{4\beta-1}+C_{t,\eps} (\cN+N^{\beta+\eps}), \label{eq:wHN-1}\\
\pm \1^{\le m}  [\widetilde H_N(t),\cN ] \1^{\le m} & \le \eta  \dGamma(1-\Delta) +  \eta^{-1} C_t mN^{4\beta-1}+C_{t,\eps} (\cN+N^{\beta+\eps}),  \label{eq:wHN-2}\\
\pm \1^{\le m} \partial_t \widetilde H_N(t) &\le \eta  \dGamma(1-\Delta) +  \eta^{-1} C_t mN^{4\beta-1}+C_{t,\eps} (\cN+N^{\beta+\eps}). \label{eq:wHN-3}
\end{align}
We will use the estimates \eqref{eq:wHN-1}-\eqref{eq:wHN-3} to show that $\Phi_N(t)$ essentially localizes in the low-particle sectors, which will imply the Bose-Einstein condensation on $u(t)$ for the original wave function $\Psi_N=U_N^*(t) \Phi_N(t)$.

\subsection{Step 1: Round kinetic bound for full dynamics} 
\begin{lemma} [Round kinetic bound for full dynamics] 
	\label{roughbounds} Let $\beta<1/4$ and let $\Phi_N(0)$ satisfy \eqref{eq:Phi0-assume-extra}. Then 
	\begin{align}
	\label{roughkineticbound}\langle\Phi_N(t),\d\Gamma(1-\Delta)\Phi_N(t)\rangle \le C_tN. 
	\end{align}
\end{lemma}
\begin{proof}
	First of all by the energy conservation of the Schr\"odinger dynamics, we have
	\begin{align}\label{eq:rba-0}
	\langle \Psi_N(t),  H_N(t)  \Psi_N(t)\rangle = \langle \Psi_N(0), H_N(t)  \Psi_N(0)\rangle.
	\end{align}
	Moreover, from the explicit computation \eqref{eq:UNHNUN} we have the simple estimate
	\begin{align}\label{eq:rba-1}
	\langle \Psi_N(t), H_N \Psi_N(t)\rangle = \langle \Phi_N(t), U_N H_N U_N^* \Phi_N(t)\rangle =  \langle \Phi_N(t), \widetilde H_N(t)  \Phi_N(t)\rangle + O(C_t  N).
	\end{align}
	
	Next, using \eqref{eq:wHN-1} with $m=N$ and $\eta=C_t N^{4\beta}$ we find that 
	$$
        \pm \widetilde H_N(t) \le C_t N^{4\beta} \dGamma(1-\Delta) +  C_t N.  
	$$
	Therefore, the assumption \eqref{eq:Phi0-assume-extra} ensures that
	$$
	\pm \langle \Phi_N(0), \widetilde H_N(0)  \Phi_N(0)\rangle \le CN.
	$$ 
	
	Combining the latter estimate with \eqref{eq:rba-0} and \eqref{eq:rba-1} we deduce that 
\begin{align}\label{eq:rba-2}
	 \langle \Psi_N(t),  H_N  \Psi_N(t)\rangle \le C_t N
\end{align}
	Since $V_N\ge 0$, we then obtain
	\begin{equation}
	\label{roughstart}
	\langle\Psi_N(t), \dGamma(-\Delta) \Psi_N(t)\rangle  \le C_t N.
	\end{equation}
	
	Next, by decomposing $1=P(t)+Q(t)$, with $P(t)=|u(t)\rangle \langle u(t)|$ we have the Cauchy-Schwarz inequality
	\begin{align*}
	-\Delta &= P(t) (-\Delta)P(t) + Q(t) (-\Delta)Q(t) + P(t) (-\Delta)Q(t) + Q(t) (-\Delta) P(t) \\
	&\ge (1-\eta^{-1}) P(t) (-\Delta)P(t) + (1-\eta) Q(t) (-\Delta)Q(t)  , \quad \forall \eta>0.
		\end{align*}
Taking $\eta=1/2$ and using $P(t) (-\Delta) P(t)=\|\nabla u(t,\cdot)\|_{L^2}^2 P(t)$ we obtain
$$
Q(t) (1-\Delta)Q(t)  \le C_t (1-\Delta).
$$
Consequently,
$$
\dGamma(Q(t) (1-\Delta)Q(t) ) \le 2 \dGamma (-\Delta) + C_t \cN.
$$
Thus from \eqref{roughstart} we deduce \eqref{roughkineticbound}:
$$
\langle \Phi_N, \dGamma(1-\Delta) \Phi_N\rangle = \langle \Psi_N,\dGamma(Q(t) (1-\Delta)Q(t) ) \Psi_N\rangle \le   C_t \langle \Psi_N,  \dGamma (1-\Delta) \Psi_N\rangle \le C_t N. 
$$

\end{proof}

\subsection{Step 2: Improved kinetic bound for truncated dynamics} \label{sec:trunc-dyn} For every $M\ge 1$, we introduce an intermediate dynamics on truncated Fock space $\cF^{\le M}_+(t)$. 
\begin{align} \label{eq:eq-PhiNM}
i\partial_t\Phi_{N,M}(t) = \1^{\le M}\widetilde H_N(t)\1^{\le M}\Phi_{N,M}(t), \quad\quad \Phi_{N,M}(0) = \1^{\le M}\Phi_N(0).
\end{align} 

Our idea is that if $M$ is significantly smaller than $N$, then we will have a better control on the kinetic energy of the truncated dynamics $\Phi_{N,M}$. 

\begin{lemma}[Refined kinetic bound for truncated dynamics] 	\label{lem:truncated-dynamics} Let $\Phi_N(0)$ satisfy \eqref{eq:Phi0-assume-extra}. Then for all $1 \ll M \ll N^{1-2\beta}$, and for all $t>0$, $\eps>0$ we have 
	\begin{equation} \label{eq:truncated-dynamics-kinetic}
	\left\langle\Phi_{N,M}(t),\d\Gamma(1-\Delta) \Phi_{N,M}(t)\right\rangle \le C_{t,\eps} \left( MN^{4\beta-1} + N^{\beta+\eps}\right).
	\end{equation}
\end{lemma}

\begin{proof} When $M\ll N^{1-2\beta}$, we have
$$
\max\{ \sqrt{MN^{2\beta-1}}, \sqrt{M^3 N^{5\beta-3}}, M^2 N^{4\beta-2}\} = \sqrt{MN^{2\beta-1}} \ll 1. 
$$
Therefore, we can apply \eqref{eq:wHN-1}-\eqref{eq:wHN-3} with $\eta=1/2$ and obtain 
	\begin{align*}
	\pm\1^{\le M}\left(\widetilde H_N(t) - \d\Gamma(-\Delta) \right)\1^{\le M} &\le \frac{1}{2} \d\Gamma(1-\Delta) + C_{t,\eps}\left( \cN + MN^{4\beta-1} + N^{\beta + \eps}\right), \\
	\pm\1^{\le M}\partial_t\widetilde H_N(t)\1^{\le M} &\le \frac{1}{2} \d\Gamma(1-\Delta) + C_{t,\eps}\left( \cN + MN^{4\beta-1} + N^{\beta + \eps}\right),\\
		\pm\1^{\le M} i[\widetilde H_N(t),\cN]\1^{\le M} &\le \frac{1}{2} \d\Gamma(1-\Delta) + C_{t,\eps}\left( \cN + MN^{4\beta-1} + N^{\beta + \eps}\right).
	\end{align*}

	Consequently, if we denote
	\begin{align*}
	A(t) := \1^{\le M}\widetilde H_N(t)\1^{\le M} + C_{t,\eps}\left(\cN + MN^{4\beta-1}+ N^{\beta + \eps}\right), 
	\end{align*}
	then, with $C_{t,\eps}$ large enough, we have the quadratic form estimates
\begin{align} 
	\frac{1}{2} \d\Gamma(1-\Delta) &\le A(t)  \le C_{t,\eps} \left(\d\Gamma(1-\Delta) +  MN^{4\beta-1} + N^{\beta + \eps}\right) ,\label{eq:up-lo-At} \\
	\pm \partial_t A(t) &\le C_{t,\eps }A(t), \quad  \pm i[A(t), \1^{\le M}\widetilde H_N(t)\1^{\le M} ] \le A(t). \label{eq:up-lo-dtAt} 
	\end{align}
	
 Using now the equation \eqref{eq:eq-PhiNM} for $\Phi_{N,M}(t)$ and   \eqref{eq:up-lo-dtAt}  we obtain
	\begin{align*}
	&\frac{\d}{\d t}\big\langle\Phi_{N,M}(t),A(t)\Phi_{N,M}(t)\big\rangle \\&= 
	\big\langle\Phi_{N,M}(t),i[A(t), \1^{\le M}\widetilde H_N(t)\1^{\le M} ]\Phi_{N,M}(t)\big\rangle + \big\langle\Phi_{N,M}(t),\partial_tA(t)\Phi_{N,M}(t)\big\rangle \\
	&\le C_{t,\eps} \big\langle\Phi_{N,M}(t),A(t)\Phi_{N,M}(t)\big\rangle.
	\end{align*}
	Therefore, Gronwall's inequality implies that
	\begin{equation}
	\nn\big\langle\Phi_{N,M}(t),A(t)\Phi_{N,M}(t)\big\rangle \le C_{t,\eps} \big\langle\Phi_{N,M}(0),A(0)\Phi_{N,M}(0)\big\rangle.
	\end{equation}
	Here recall that $\Phi_{N,M}(0) = \1^{\le M}\Phi(0)$. 	Combining with \eqref{eq:up-lo-At} we find that
	\begin{align} \label{eq:PhiNM<=Phi0}
 &\left\langle \Phi_{N,M}(t), \d\Gamma(1-\Delta)  \Phi_{N,M}(t)\right\rangle\nn\\
 &\le C_{t,\eps} \left\langle \1^{\le M} \Phi_N(0),\left(\d\Gamma(1-\Delta) + MN^{4\beta-1} + N^{\beta + \eps}\right) \1^{\le M} \Phi_N(0) \right\rangle.
	\end{align}
	The kinetic bound \eqref{eq:truncated-dynamics-kinetic} then follows from the technical assumption \eqref{eq:Phi0-assume-extra}.	
	\bigskip
\end{proof}

\subsection{Step 3: Truncated vs. full dynamics} \label{sec:trunc-full-dyn} Now we show that if $M$ is sufficiently large, then the truncated dynamics $\Phi_{N,M}$ in \eqref{eq:eq-PhiNM} is close to the full dynamics $\Phi_N$. 

\begin{lemma}[Norm approximation to full dynamics] 	\label{lem:truncated-vs-full-dynamics} Let $\beta<1/4$ and let $\Phi_N(0)$ satisfy \eqref{eq:Phi0-assume-extra}. Then for all $1 \ll M \ll N^{1-2\beta}$, and for all $t>0$, $\eps>0$ we have 
	\begin{equation} \label{eq:truncated-dynamics-norm}
	\|\Phi_N(t) -\Phi_{N,M}(t)\|^2 \le C_{t,\eps} \left( N^{(4\beta -1)/2}+ \frac{N^{2\beta}}{\sqrt{M}} + \frac{N^{(1+\beta+\eps)/2}}{M} \right).
	\end{equation}
\end{lemma}

\begin{proof} Using the mass conservation 
$$\|\Phi_N(t)\| = \|\Phi_N(0)\| = \|\1^{\le N}\Phi(0)\|\le 1, \quad \|\Phi_{N,M}(t)\| = \|\Phi_{N,M}(0)\| = \|\1^{\le M}\Phi(0)\|\le1$$
we have
\begin{align} \label{eq:norm-innner}
\|\Phi_N(t) -\Phi_{N,M}(t)\|^2 \le 2(1-\Re\langle\Phi_N(t),\Phi_{N,M}(t)\rangle).
\end{align}
To bound this expression we will use a method from \cite{NamNap-17/2,BreNamNapSch-17} where we will separate the sectors with large and small particle numbers. Take $M/2 \le m \le M-3$ and write
\begin{align} \label{eq:decom-M-m}
\langle\Phi_N(t),\Phi_{N,M}(t)\rangle = \langle\Phi_N(t),\1^{\le m}\Phi_{N,M}(t)\rangle + \langle\Phi_N(t),\1^{> m}\Phi_{N,M}(t)\rangle.
\end{align}
The term containing large particle numbers can be bounded using the Cauchy-Schwarz inequality and the kinetic bound \eqref{eq:truncated-dynamics-kinetic}, which gives
\begin{align}
|\langle \Phi_N(t),\1^{>m}\Phi_{N,M}(t)\rangle| &\le 
\|\Phi_N(t)\| \| \1^{>m} \Phi_{N,M}\| \nn\\
&\le \left\langle\Phi_{N,M}(t),\frac{\cN}{m}\Phi_{N,M}(t)\right\rangle^\frac{1}{2} \nn\\
&\le C_{t,\eps}  \sqrt{\frac{1}{m} \left( MN^{4\beta-1} + N^{\beta+\eps}\right)} \nn\\ \label{manyparticlesector}
&\le C_{t,\eps} \Big( N^{(4\beta -1)/2} + \frac{N^{(\beta+\eps)/2}}{\sqrt{M}}  \Big).
\end{align}
For the few particle sector we see that
\begin{align} \label{eq:fewparticlesector}
\frac{d}{dt}\langle\Phi_N(t),\1^{\le m}\Phi_{N,M}(t)\rangle &= i\Big\langle\Phi_N(t),\left(\widetilde H_N(t)\1^{\le m}- \1^{\le m}\1^{\le M}\widetilde H_N(t)\1^{\le M}\right)\Phi_{N,M}(t)\Big\rangle \nn\\
&=\Big\langle\Phi_{N}(t),i[\widetilde H_N(t),\1^{\le m}]\Phi_{N,M}(t)\Big\rangle.
\end{align}
Here we have used the fact that
\begin{equation}
\label{truncatedHamiltonianidentity}
\1^{\le m} \1^{\le M}\widetilde H_N(t)\1^{\le M} = \1^{\le m}\widetilde H_N(t)
\end{equation}
which is true because $m\le M - 3 $ and $\widetilde H_N(t)$ contains at most 3 creation and at most 3 annihilation operators.
To bound this expression we will average over $m\in[M/2, M-3]$. The key point of this argument is summarized in the following lemma.

\begin{lemma} \label{averagebound}
	For $1\le M\ll N^{1-2\beta}$ we have the operator bound
\begin{equation}
\pm \frac{1}{M/2 - 2}\sum_{m=M/2}^{M-3} i[\widetilde H_N(t),\1^{\le m}] \le \frac{C_{t,\eps}}{M}\left(\d\Gamma(1-\Delta) + MN^{4\beta-1} + N^{\beta+\eps}\right).
\end{equation}
\end{lemma} 

Let us postpone the proof of Lemma \ref{averagebound} and proceed to conclude Lemma \ref{lem:truncated-dynamics}. Using \eqref{eq:fewparticlesector}, Lemma \ref{averagebound} and the kinetic estimates in Lemmas \ref{roughbounds}, \eqref{lem:truncated-dynamics}, we obtain
	\begin{align*}
	&\Big|\frac{1}{M/2 -1}\sum_{m=M/2}^{M-3}\frac{d}{dt}\langle\Phi_N(t),\1^{\le m}\Phi_{N,M}(t)\rangle\Big| \\
	& = \Big|\frac{1}{M/2 -1}\sum_{m=M/2}^{M-3}\Big\langle\Phi_{N}(t),i[\widetilde H_N(t),\1^{\le m}]\Phi_{N,M}(t)\Big\rangle\Big| \\
	&\le \frac{C_{t,\eps}}{M}\Big\langle\Phi_N(t),\left(\d\Gamma(1-\Delta) +  MN^{4\beta-1} + N^{\beta+\eps}\right)\Phi_N(t)\Big\rangle^\frac{1}{2}\times\\
	&\quad \times\Big\langle\Phi_{N,M}(t),\left(\d\Gamma(1-\Delta) +   MN^{4\beta-1} + N^{\beta+\eps}\right)\Phi_{N,M}(t)\Big\rangle^\frac{1}{2} \\
	&\le  \frac{C_{t,\eps}}{M} \sqrt{ N +  MN^{4\beta-1} + N^{\beta+\eps}} \sqrt{ MN^{4\beta-1} + N^{\beta+\eps}}\\
	&\le C_{t,\eps} \left( N^{4\beta-1} + \frac{N^{2\beta}}{\sqrt{M}} + \frac{N^{(1+\beta+\eps)/2}}{M} \right). 
	\end{align*}
	Furthermore, the assumption \eqref{eq:Phi0-assume-extra} ensures that 
	\begin{align*}
	\langle\Phi_N(0),\1^{\le m}\Phi_{N,M}(0)\rangle &= \langle\Phi_N(0),\1^{\le m} \Phi_N (0)\rangle =  1 - \langle\Phi(0),\1^{> m}\Phi(0)\rangle  \\
	&\ge 1 - \left\langle\Phi(0),\frac{\cN}{m}\Phi(0)\right\rangle \ge 1 -\frac{C}{M}.
	\end{align*}
	Hence, we get
	\begin{align}
	&\Re\frac{1}{M/2 -1}\sum_{m=M/2}^{M-3}\langle\Phi_N(t),\1^{\le m}\Phi_{N,M}(t)\rangle  \nn\\
	&= \langle\Phi_N(0),\1^{\le m}\Phi_{N,M}(0)\rangle + \int_0^t \d s \Re \frac{1}{M/2 -1}\sum_{m=M/2}^{M-3}\frac{d}{ds}\langle\Phi_N(s),\1^{\le m}\Phi_{N,M}(s)\rangle\nn\\
	&\ge 1 - \frac{C}{M} - C_{t,\eps} \left( N^{4\beta-1} + \frac{N^{2\beta}}{\sqrt{M}} + \frac{N^{(1+\beta+\eps)/2}}{M} \right).   \label{lowparticlebound}
	\end{align}
	Moreover, from \eqref{manyparticlesector} we immidiately see that 
	\begin{align*}
	\Re\frac{1}{M/2 -1}\sum_{m=M/2}^{M-3}\langle\Phi_N(t),\1^{> m}\Phi_{N,M}(t)\rangle \ge -C_{t,\eps} \Big( N^{(4\beta -1)/2} + \frac{N^{(\beta+\eps)/2}}{\sqrt{M}}  \Big).
	\end{align*}
	Summing these two estimates and using the decomposition \eqref{eq:decom-M-m}, we get 
	\begin{align*}
	\Re\langle\Phi_N(t),\Phi_{N,M}(t)\rangle \ge 1 - C_{t,\eps} \left( N^{(4\beta -1)/2}+ \frac{N^{2\beta}}{\sqrt{M}} + \frac{N^{(1+\beta+\eps)/2}}{M} \right). 
	\end{align*}
	Thus in conclusion, \eqref{eq:norm-innner} implies that 
	\begin{align*}
	\|\Phi_N(t) -\Phi_{N,M}(t)\|^2 \le C_{t,\eps} \left( N^{(4\beta -1)/2}+ \frac{N^{2\beta}}{\sqrt{M}} + \frac{N^{(1+\beta+\eps)/2}}{M} \right).
	\end{align*}
	This finishes the proof of Lemma \ref{lem:truncated-dynamics}. 
\end{proof}

It remains to give 

\begin{proof}[Proof of Lemma \ref{averagebound}]
	We write
	\begin{align*}
	[\widetilde H_N(t),\1^{\le m}] = \1^{>m}\widetilde H_N(t)\1^{\le m} - \1^{\le m}\widetilde H_N(t)^{> m}.
	\end{align*}
	Since both terms can be treated similarly we will only handle the first one. Let
	\begin{align*}
	&\cE_1 = a^*\left(Q(t)\frac{1}{2}\iint|u(t)|^2V_N|u(t)|^2\d y \d z\,u(t)\right) \sqrt{N- \cN} \frac{(N-\cN -1)^2 - (N-\cN -1) -N^2}{N^2} \\
	&\quad\quad\quad +\frac{1}{N^2} \iiiint (Q(t)\otimes Q(t)\int|u(t,z)|^2V_N\d zQ(t)\otimes 1)(x,y;x',y') u(t,y')\times\\ &\quad\quad\quad\quad\quad\quad\quad\quad\times a^*_xa^*_ya_{x'}\d x dy\d x '\d y'\sqrt{N-\cN}(N-\cN-1) \\
	&\quad\quad\quad+\frac{1}{2N^2} \idotsint(Q(t)\otimes Q(t)\otimes 1 V_N 1\otimes 1 \otimes Q(t))(x,y,z;x',y',z') \overline{u(t,z)}u(t,x')u(t,z') \times \\
	&\quad\quad\quad\quad\quad\quad\quad\quad\times a^*_xa^*_ya_{z'}\d x \d y \d z\d x '\d y'\d z'\sqrt{N -\cN}(N-\cN-1) + h.c. \\
	&\quad\quad\quad+\frac{1}{2N^2}\idotsint(Q(t)\otimes Q(t)\otimes Q(t)V_NQ(t)\otimes 1\otimes 1)(x,y,z;x',y',z')u(t,y')u(t,z')\times \\
	&\quad\quad\quad\quad\quad\quad\quad\quad \times a^*_xa^*_ya^*_za_{x'}\sqrt{N-\cN -1}\sqrt{N-\cN}\d x \d y \d z\d x '\d y'\d z' + h.c. + \\
	&\cE_2 = \frac{1}{2} \iint K(x,y) a_x^* a_y^* \d x \d y\frac{\sqrt{N-\cN -1} \sqrt{N - \cN}(N -\cN -2)}{N^2} \\
	&\quad\quad\quad+\frac{1}{2N^2}\idotsint(Q(t)\otimes Q(t)\otimes Q(t)V_NQ(t)\otimes Q(t)\otimes 1)(x,y,z;x',y',z')u(t,z')\times \\
	&\quad\quad\quad\quad\quad\quad\quad\quad \times a^*_xa^*_ya^*_za_{x'}a_{y'}\sqrt{N-\cN -2}\d x \d y \d z\d x '\d y'\d z' + h.c. \\
    &\cE_3 = \frac{1}{6N^2}\idotsint(Q(t)\otimes Q(t)\otimes Q(t)V_N 1\otimes 1\otimes 1)(x,y,z;x',y',z')u(t,x')u(t,y')u(t,z') \\ 
    & \quad\quad\quad\quad\quad\times a^*_x a^*_y a^*_z\d x \d y \d z\d x '\d y'\d z' \sqrt{N-\cN-2}\sqrt{N-\cN-1}\sqrt{N-\cN}
	\end{align*} 
	
	With this we can write 
	\begin{align*}
	\1^{> m}\widetilde H_N(t)\1^{\le m} = \cE_1\1(\cN = m) + \cE_2\1(m-1\le\cN\le m) + \cE_3\1(m-2\le\cN\le m).
	\end{align*}
	From this we get
	\begin{align*}
	&\sum_{m=M/2}^{M-3} \1^{> m}\widetilde H_N(t)\1^{\le m} = \cE_1 \1(M/2\le\cN\le M-3) \\
	&+ \cE_2\Big[\1(M/2 - 1\le\cN\le M-4) + \1(M/2\le\cN\le M-3)\Big] \\& + \cE_3\Big[\1(M/2 - 2\le\cN\le M-5) +\1(M/2 - 1\le\cN\le M-4) +\1(M/2\le\cN\le M-3)\Big]
	\end{align*}
	
	Let us concentrate on the most complicated term $\cE_3$. We can bound $\cE_3\1(M/2\le\cN\le M-3)$ using \eqref{eq:errorR0-5}.  For every $X \in \cF_+(t)$, we have
	\begin{align*}
	&\left|\left\langle X, \cE_3\1(M/2\le\cN\le M-3)X\right\rangle\right| \le \left\langle X, \d\Gamma(1-\Delta) + MN^{4\beta-1}X\right\rangle^{\frac{1}{2}}\times\\
	&\quad\times\left\langle \1(M/2\le\cN\le M-3)X, \d\Gamma(1-\Delta) + MN^{4\beta-1}\1(M/2\le\cN\le M-3)X\right\rangle^{\frac{1}{2}} \\
	&\quad\le \left\langle X, \d\Gamma(1-\Delta) + MN^{4\beta-1}X\right\rangle.
	\end{align*}
	Thus we have the quadratic form estimate
	$$
	\pm \Big( \cE_3\1(M/2\le\cN\le M-3) + h.c. \Big)  \le \d\Gamma(1-\Delta) + MN^{4\beta-1}.
	$$
	All the other terms can be bound similarly using the Lemmas \ref{lem:Bogoliubovbounds} and \ref{lem:errorbounds}. For example, we have
	$$
	\pm \Big( \cE_2\1(M/2\le\cN\le M-3) + h.c. \Big)  \le \d\Gamma(1-\Delta) + N^{\beta+\eps},
	$$
	and similarly,
	\begin{equation} \label{eq:Bog-com-aver}
	 \pm \frac{1}{M/2 - 2}\sum_{m=M/2}^{M-3} i[ \bH(t),\1^{\le m}] \le \frac{C_{t,\eps}}{M}\left(\d\Gamma(1-\Delta) + N^{\beta+\eps}\right).
	\end{equation}
	This ends the proof of Lemma \ref{averagebound}.
\end{proof}

	\subsection{Step 3: Conclusion of Theorem \ref{thm:main1}}  
\begin{proof}[Proof of Theorem \ref{thm:main1}] Let $0<\beta<1/6$. Let $u(t)$ and $\varphi(t)$ be the solution to the quintic Hartree equation \eqref{eq:Hartree}  and the quintic NLS \eqref{eq:Hartree-NLS} with the same initial condition $\varphi_0=u_0 \in H^4(\R^3)$. We define 
$$
\Phi_N(t)=e^{-i\int_0^t \chi(s) \d s}\,U_N(t) \Psi_N(t) = e^{-i\int_0^t \chi(s) \d s}\,U_N(t) e^{-itH_N} U_N^*(0)\Phi_{N}(0).
$$ 
The actions \eqref{eq:UN-action} lead to the key identity
\begin{align} \label{eq:key-gamma1-PsiN-PhiN}
Q(t) \gamma_{\Psi_N(t)}^{(1)} Q(t) = \gamma_{\Phi_N(t)}^{(1)}. 
\end{align}
In particular, recall that the asusmption \eqref{eq:leading-order} in Theorem  \ref{thm:main1}  is equivalent to \eqref{eq:Phi0-assume-extra},
$$
\langle \Phi_N(0), \d\Gamma(1-\Delta)  \Phi_N(0)\rangle =  \Tr \Big((1-\Delta) \gamma_{\Phi_N(0)}^{(1)}\Big) = \Tr \Big((1-\Delta) Q(0) \gamma_{\Psi_N(0)}^{(1)}Q(0)\Big) \le C. 
$$

By applying Lemma \ref{lem:truncated-vs-full-dynamics}  we find that for all 
$$1 \ll M \ll N^{1-2\beta}$$
the truncated dynamics $\Phi_{N,M}(t)\in \cF_+^{\le M}(t)$ defined by
\begin{align*} 
i\partial_t\Phi_{N,M}(t) = \1^{\le M}\widetilde H_N(t)\1^{\le M}\Phi_{N,M}(t), \quad\quad \Phi_{N,M}(0) = \1^{\le M} \Phi_N(0) 
\end{align*} 
satisfies the norm approximation
\begin{align} \label{eq:com-PhiNM-PhiN-inter}
\|\Phi_N - \Phi_{N,M}\|^2 
\le C_{t,\eps} \left( N^{(4\beta -1)/2}+ \frac{N^{2\beta}}{\sqrt{M}} + \frac{N^{(1+\beta+\eps)/2}}{M} \right).
\end{align}	

Now we transfer \eqref{eq:com-PhiNM-PhiN-inter} to an estimate on the one-body density matrix. Since $\Phi_{N,M}(t)\in \cF_+^{\le M}(t)$, we have the obvious estimate
$$
\|\Phi_N(t) -\Phi_{N,M}(t)\|^2 = \|\1^{>M} \Phi_N(t)\|^2 + \| \1^{\le M} \Phi_N(t) - \Phi_{N,M}(t)\|^2 \ge \|\1^{>M} \Phi_N(t)\|^2.
$$
Thus \eqref{eq:com-PhiNM-PhiN-inter} gives 
\begin{align} \label{eq:conclusion-leading}
\|\1^{>M} \Phi_N(t)\|^2 \le C_{t,\eps} \left( N^{(4\beta -1)/2}+ \frac{N^{2\beta}}{\sqrt{M}} + \frac{N^{(1+\beta+\eps)/2}}{M} \right).
\end{align}

In viewing of \eqref{eq:key-gamma1-PsiN-PhiN}, we have
\begin{align*} N^{-1} \Tr\Big(Q(t)\gamma_{\Psi_N(t)}^{(1)} Q(t) \Big) &= N^{-1}\Tr\Big(\gamma_{\Phi_N(t)}^{(1)}\Big) = N^{-1}\langle \Phi_N(t), \cN \Phi_N(t)\rangle \nn\\
&=  N^{-1}\langle \Phi_N(t), \1^{\le M} \cN \Phi_N(t)\rangle + N^{-1}\langle \Phi_N(t), \1^{> M} \cN \Phi_N(t)\rangle\nn\\
&\le N^{-1}M \| \1^{\le M}\Phi_N(t) \|^{2} + \|\1^{> M} \Phi_N(t)\|^2 \nn\\
&\le M N^{-1} +  C_{t,\eps} \left(  N^{(4\beta -1)/2}+ \frac{N^{2\beta}}{\sqrt{M}} + \frac{N^{(1+\beta+\eps)/2}}{M} \right). 
\end{align*}
Then by the Cauchy-Schwarz inequality and the fact that $1-Q(t)=|u(t)\rangle \langle u(t)|$ is a rank-one projection, we obtain
\begin{align} \label{eq:QgQ}
&\Tr \left|  N^{-1}\gamma_{\Psi_N(t)}^{(1)} - |u(t)\rangle \langle u(t)| \right| \le  \sqrt{N^{-1} \Tr\Big(Q(t)\gamma_{\Psi_N(t)}^{(1)} Q(t) \Big)} \nn\\
&\le \sqrt{M N^{-1} +  C_{t,\eps} \left(    \frac{N^{2\beta}}{\sqrt{M}} + \frac{N^{(1+\beta+\eps)/2}}{M} \right)}. 
\end{align} 

On the other hand, from \eqref{eq:H-to-NLS} in Theorem \ref{thm:Hartree} we find that
\begin{equation} \label{eq:QgQ-phi}
		\tr\Big||u(t)\rangle\langle u(t)| - |\phi(t)\rangle\langle\phi(t)|\Big| \le 2 \|u(t) - \phi(t)\|_{L^2(\R^3)} \le C_t N^{-\beta/2}.
		\end{equation}
Thus in summary, from \eqref{eq:QgQ} and \eqref{eq:QgQ-phi} we conclude by the triangle inequality that
\begin{align} \label{eq:QgQ-final}
&\Tr \left| N^{-1} \gamma_{\Psi_N(t)}^{(1)} - |\varphi(t)\rangle \langle \varphi(t)| \right| \nn\\
&\le C_t N^{-\beta/2} +  \sqrt{M N^{-1} +  C_{t,\eps} \left( N^{(4\beta -1)/2}+ \frac{N^{2\beta}}{\sqrt{M}} + \frac{N^{(1+\beta+\eps)/2}}{M} \right)}. 
\end{align}

It remains to optimize the right side of \eqref{eq:QgQ-final} over $1 \ll M \ll N^{1-2\beta}$. Choosing
$$
M= N^{1-2\beta-\eps}
$$
for $\eps>0$ arbitrarily small (but independent of $N$), we deduce from \eqref{eq:QgQ-final} that
\begin{align}  
\Tr \left|  N^{-1}\gamma_{\Psi_N(t)}^{(1)} - |\phi(t)\rangle \langle \phi(t)| \right| \le C_{t,\eps} N^{-\alpha}
\end{align}
for any constant
$$
\alpha < \min\left\{\frac{\beta}{2}, \frac{1-6\beta}{4} \right\}.
$$
This completes the proof of Theorem \ref{thm:main1}.
\end{proof}

\section{Proof of Theorem~\ref{thm:main2}} \label{sec:thm2}
In this section we prove the norm approximation in Theorem \ref{thm:main2}. As explained in Section \ref{sec:gen-stra}, it boils down to prove that the transformed dynamics $\Phi_N(t)= e^{-i\int_0^t \chi(s) \d s}\,U_N(t)\Psi_N(t)$ in \eqref{eq:eq-PhiN} converges to the Bogoliubov dynamics $\Phi(t)$ in \eqref{eq:Bog}. 

We will follow the strategy of the previous section to consider the truncated dynamics $\Phi_{N,M}$ in \eqref{eq:eq-PhiNM},
\begin{align*} 
i\partial_t\Phi_{N,M}(t) = \1^{\le M}\widetilde H_N(t)\1^{\le M}\Phi_{N,M}(t), \quad\quad \Phi_{N,M}(0) = \1^{\le M}\Phi_N(0).
\end{align*} 
The norm approximation between $\Phi_{N,M}$ and $\Phi_N$ was already provided in Lemma \ref{lem:truncated-vs-full-dynamics}. Therefore, it is natural to compare $\Phi_{N,M}$ with the Bogoliubov dynamics $\Phi(t)$.



\subsection{Step 1: Truncated vs. Bogoliubov dynamics}

\begin{lemma}[Truncated vs. Bogoliubov dynamics] \label{lem:normtrunc-to-Bog}
Let $1 \ll M \ll N^{1-2\beta}$. Then
\begin{equation} \label{eq:norm-PhiNM-Phi}
\|\Phi_{N,M}(t)	- \Phi(t)\|^2  \le C_{t,\eps} \left( \frac{ N^{\beta+\eps}}{M} + \sqrt{M N^{5\beta-1+2\eps}}  \right).
\end{equation}
\end{lemma}

\begin{proof} Similarly to \eqref{eq:norm-innner}, we have 
	\begin{align*}
	\|\Phi_{N,M}(t)	- \Phi(t)\|^2 \le 2\left(1-\Re\left\langle\Phi_{N,M}(t),\Phi(t)\right\rangle\right).
	\end{align*}
Now picking $M/2 \le m \le M -3$ we write 
	\begin{align} \label{eq:decom-PhiNM-Phi}
	\langle\Phi_{N,M}(t),\Phi(t)\rangle =  \langle\Phi_{N,M}(t),\1^{> m}\Phi(t)\rangle + \langle\Phi_{N,M}(t),\1^{\le m}\Phi(t)\rangle.
	\end{align}
	The first term containing the many particle sector can be bounded using \eqref{eq:Bog-dyn-N-bounds-kinetic0}:
	\begin{align}
	\label{lastnormmanyparticle}
	|\langle\Phi_{N,M}(t),\1^{> m}\Phi(t)\rangle| \le \|\1^{> m}\Phi(t) \| \le \langle \Phi(t), (\cN/m) \Phi(t) \rangle^{1/2}\le C_{t ,\eps}\frac{N^{(\beta+\eps)/2}}{\sqrt{M}}.
	\end{align}
	For the second term in \eqref{eq:decom-PhiNM-Phi} we calculate the derivative
	\begin{align}
	\nn\frac{\d}{\d t}\langle\Phi_{N,M}(t),\1^{\le m}\Phi(t)\rangle &= i \Big\langle\Phi_{N,M}(t),\left(\1^{\le M}\widetilde H_N(t)\1^{\le M}\1^{\le m} - \1^{\le m}\bH(t)\right)\Phi(t)\Big\rangle \\ 
	&\label{derivativetermlastnorm}=i \Big\langle\Phi_{N,M}(t), [ \bH(t),\1^{\le m}] + \left((\widetilde H_N(t)- \bH(t))\1^{\le m} \right)\Phi(t)\Big\rangle,
	\end{align} 
	where we have used, as in \eqref{truncatedHamiltonianidentity}, 
$$
\1^{\le M}\widetilde H_N(t)\1^{\le M}\1^{\le m} = \widetilde H_N(t)\1^{\le m}.
$$
	
	To bound the first term in \eqref{derivativetermlastnorm} containing the commutator $[ \bH(t),\1^{\le m}]$, we average over $m\in[M/2,M -3]$, then use \eqref{eq:Bog-com-aver} and the kinetic bounds in Theorem \ref{thm:Bog-equa} and Lemma \ref{lem:truncated-dynamics}. We obtain
	\begin{align} \label{eq:first-term-derivativetermlastnorm}
	&\Big|\frac{1}{M/2 -1}\sum_{m=M/2}^{M-3}\Big\langle\Phi_{N,M}(t),i[\bH(t),\1^{\le m}]\Phi(t)\Big\rangle\Big| \nn\\
	&\le \frac{C_{t,\eps}}{M}\Big\langle\Phi_{N,M}(t),\left(\d\Gamma(1-\Delta) + N^{\beta+\eps}\right)\Phi_{N, M}(t)\Big\rangle^\frac{1}{2}
	\Big\langle\Phi(t),\left(\d\Gamma(1-\Delta) +  N^{\beta+\eps}\right)\Phi(t)\Big\rangle^\frac{1}{2} \nn\\
	&\le \frac{C_{t,\eps}}{M} \sqrt{MN^{4\beta-1} + N^{\beta+\eps}} \cdot \sqrt{N^{\beta+\eps}}	\le  \frac{C_{t,\eps}}{M} \left( \sqrt{M N^{5\beta-1+\eps}} +N^{\beta+\eps}\right).
	\end{align}

	In order to estimate the second term in \eqref{derivativetermlastnorm}, we use \eqref{eq:errorR0-5} in Lemma \ref{lem:errorbounds}:  when $m\ll N^{1-2\beta}$,
	$$
	\1^{\le m} (\widetilde H_N(t) -\bH(t)) \1^{\le m} \le  \eta \dGamma(1-\Delta) +  \eta^{-1}C_t m N^{4\beta-1}, \quad \forall \eta \ge C_t \sqrt{mN^{2\beta-1}}.
	$$
	Then for all $M/2 \le m \le M -3$, by the Cauchy-Schwarz inequality and the kinetic bounds in Theorem \ref{thm:Bog-equa} and Lemma \ref{lem:truncated-dynamics},  we can estimate
	\begin{align*}
	&\Big| \Big\langle\Phi_{N,M}(t),(\widetilde H_N(t)- \bH(t))\1^{\le m}\Phi(t)\Big\rangle\Big| \nn\\
	&=\Big | \Big\langle  \Phi_{N,M}(t), \1^{\le m+3}(\widetilde H_N(t)- \bH(t))\1^{\le m+3} \1^{\le m}\Phi(t)\Big\rangle\Big| \nn\\
	&\le C_{t} \left\langle \Phi_{N,M}(t), \left(  \eta \dGamma(1-\Delta) + \eta^{-1} MN^{4\beta-1}\right) \Phi_{N,M}(t)\right\rangle^\frac{1}{2}\times\\
	&\quad\times\left\langle\Phi(t), \left(   \eta \dGamma(1-\Delta) + \eta^{-1} MN^{4\beta-1}\right)  \Phi(t)\right\rangle^\frac{1}{2} \nn\\
	&\le C_{t,\eps} \sqrt{\eta (MN^{4\beta-1}+N^{\beta+\eps})+ \eta^{-1} MN^{4\beta-1} } \cdot \sqrt{\eta N^{\beta+\eps} + \eta^{-1} MN^{4\beta-1}}\nn\\
	&\le C_{t,\eps} \sqrt{ \eta^2 N^{2(\beta+\eps)} + (1+\eta^2) MN^{5\beta-1}+(1+\eta^{-2})(MN^{4\beta-1})^2  }
	\end{align*}
for all $\eta \ge C_t \sqrt{MN^{2\beta-1}}$. By choosing
$$
\eta= C_t \sqrt{MN^{3\beta-1}} 
$$
we obtain
\begin{align} \label{eq:PhiNM-HN-H-cs-final-2}
\Big| \Big\langle\Phi_{N,M}(t),(\widetilde H_N(t)- \bH(t))\1^{\le m}\Phi(t)\Big\rangle\Big|  \le C_{t,\eps} \sqrt{M N^{5\beta-1+2\eps} }.
\end{align}
	Averaging the latter bound over $M/2\le m\le M-3$, and combining with \eqref{eq:first-term-derivativetermlastnorm}, we deduce from \eqref{derivativetermlastnorm} that
$$\left| \frac{1}{M/2 -1}\sum_{m=M/2}^{M-3} \frac{d}{dt} \langle\Phi_{N, M}(t),\1^{\le m}\Phi(t)\rangle \right| \le C_{t,\eps} \left( \frac{ N^{\beta+\eps}}{M} + \sqrt{M N^{5\beta-1+2\eps}}  \right).$$
	Then using the the bound
\begin{align*}
	\langle\Phi_{N,M}(0),\1^{\le m} \Phi(0) \rangle =  1 - \langle\Phi_N(0),\1^{> m}\Phi_N(0)\rangle \ge 1 - \frac{C}{M},
	\end{align*}
	which follows from assumption $\langle\Phi_N(0),\cN\Phi_N(0)\rangle \le C$, we obtain as in \eqref{lowparticlebound}, 
$$\Re \frac{1}{M/2 -1}\sum_{m=M/2}^{M-3} \langle\Phi_{N, M}(t),\1^{\le m}\Phi(t)\rangle  \ge 1 - C_{t,\eps} \left( \frac{ N^{\beta+\eps}}{M} + \sqrt{M N^{5\beta-1+2\eps}}  \right).$$

Putting the latter estimate together with  \eqref{lastnormmanyparticle}  (after averaging over $M/2\le m\le M-3$), we conclude that from  \eqref{eq:decom-PhiNM-Phi}  that
	\begin{equation*}
	\|\Phi_{N,M}(t)	- \Phi(t)\|^2  \le C_{t,\eps} \left( \frac{ N^{\beta+\eps}}{M} + \sqrt{M N^{5\beta-1+2\eps}}  \right).
	\end{equation*}
	This is the desired estimate. 
\end{proof}

\subsection{Step 2: A further truncated dynamics} Recall that from Lemma \ref{lem:truncated-vs-full-dynamics} and Lemma \ref{lem:normtrunc-to-Bog}, we have proved that for all $1 \ll M \ll N^{1-2\beta}$,
	\begin{align*} 
	\|\Phi_N(t) -\Phi_{N,M}(t)\|^2 &\le C_{t,\eps} \left( N^{(4\beta -1)/2}+ \frac{N^{2\beta}}{\sqrt{M}} + \frac{N^{(1+\beta+\eps)/2}}{M} \right),\\
\|\Phi_{N,M}(t)	- \Phi(t)\|^2  &\le C_{t,\eps} \left( \frac{ N^{\beta+\eps}}{M} + \sqrt{ M N^{5\beta-1+2\eps}} \right).
\end{align*}

Obviously, we can deduce $\|\Phi_N(t)-\Phi(t)\|\to 0$ if we can control all the error terms in  the above estimates. However, to make both error terms 
$$
\frac{N^{2\beta}}{\sqrt{M}} , \quad \sqrt{M N^{5\beta-1+2\eps} }
$$
small simultaneously, we would need 
$$
\frac{N^{2\beta}}{\sqrt{M}} \cdot  \sqrt{M N^{5\beta-1+2\eps} } =  \sqrt{N^{9\beta-1+2\eps}} \to 0
$$
which requires $\beta<1/9$. 

To extend the norm approximation to all $\beta<1/6$, we need a further step. We will introduce another truncated dynamics $\Phi_{N,\tilde M}$ with $\widetilde M \ll M$ and apply Lemma \ref{lem:normtrunc-to-Bog} to the new one. Of course, to make this strategy work we need to show that $\Phi_{N,\tilde M}$  is sufficiently close to $\Phi_{N,M}$. This is the content of the following 
\begin{lemma} [Intermediate norm approximation] \label{lem:normtrunc-to-trunc}
	Let $1\ll \tilde M \le M \ll N^{1-2\beta}$. Then
\begin{equation*}
\|\Phi_{N,M}(t)	- \Phi_{N,\tilde M}(t)\|^2 \le  C_{t,\eps} \left( N^{(4\beta -1)/2} + \frac{N^{(\beta+\eps)/2}}{\sqrt{\tilde M}} + \sqrt{\frac{M}{\tilde M}}N^{4\beta-1} +\frac{\sqrt{MN^{5\beta+\eps-1}}}{\tilde M}  \right).
\end{equation*}
\end{lemma}
\begin{proof}
	The proof strategy follows similarly as in Lemma~\ref{lem:truncated-vs-full-dynamics}. We use again
	\begin{align} \label{eq:nor-PhiNM-PhintM-norm-inner}
	\|\Phi_{N,M}(t)	- \Phi_{N,\tilde M}(t)\|^2 \le 2\left(1-\Re\left\langle\Phi_{N,M}(t),\Phi_{N,\tilde M}(t)\right\rangle\right). 
	\end{align}
	For any $\tilde M/2 \le m \le \tilde M -3$ we split the right  side as 
	\begin{align*}
	\left\langle\Phi_{N,M}(t),\Phi_{N,\tilde M}(t)\right\rangle = \left\langle\Phi_{N,M}(t),\1^{\le m}\Phi_{N,\tilde M}(t)\right\rangle + \left\langle\Phi_{N,M}(t),\1^{> m}\Phi_{N,\tilde M}(t)\right\rangle.
	\end{align*}
	The second term can be bounded similarly to \eqref{manyparticlesector}:
	\begin{align} 	\label{manyparticletrunc} 
	|\langle\Phi_{N,M}(t),\1^{> m}\Phi_{N,\tilde M}(t)\rangle|
	&\le  \| \Phi_{N,M}(t)\| \cdot \| \1^{>m}\Phi_{N,\tilde M}(t)\| \\ &\le C_{t,\eps}\Big( N^{(4\beta -1)/2} + \frac{N^{(\beta+\eps)/2}}{\sqrt{\tilde M}}  \Big).\nn
	\end{align}
	The bound for the few particle sectors follows again by averaging over $m\in[\tilde M/2, \tilde M -3]$, then using Lemma~\ref{averagebound} and  the kinetic bound in Lemma~\ref{lem:truncated-dynamics} (for both $\Phi_{N,M}(t)$ and $\Phi_{N,\tilde M}(t)$). This gives
	\begin{align*}
&\Big|\frac{1}{\tilde M/2 -1}\sum_{m=\tilde M/2}^{\tilde M-3}\frac{d}{dt}\langle\Phi_{N,M}(t),\1^{\le m}\Phi_{N,\tilde M}(t)\rangle\Big| \\
&= \Big|\frac{1}{\tilde M/2 -1}\sum_{m=\tilde M/2}^{\tilde M-3}\Big\langle\Phi_{N,M}(t),i[\widetilde H_N(t),\1^{\le m}]\Phi_{N,\tilde M}(t)\Big\rangle\Big| \\
&\le \frac{C_t}{\tilde M}\Big\langle\Phi_{N,M}(t),\left(\d\Gamma(1-\Delta) + \tilde M N^{4\beta-1} + N^{\beta+\eps}\right)\Phi_{N, M}(t)\Big\rangle^\frac{1}{2}\times\\
&\times\Big\langle\Phi_{N,\tilde M}(t),\left(\d\Gamma(1-\Delta) + \tilde MN^{4\beta-1} + N^{\beta+\eps}\right)\Phi_{N,\tilde M}(t)\Big\rangle^\frac{1}{2} \\
&\le  \frac{C_{t,\eps}}{\tilde M} \sqrt{  MN^{4\beta-1} + N^{\beta+\eps}} \sqrt{ \tilde M N^{4\beta-1} + N^{\beta+\eps}}\\
&\le C_{t,\eps} \left( \sqrt{\frac{M}{\tilde M}}N^{4\beta-1} +\frac{\sqrt{MN^{5\beta+\eps-1}}}{\tilde M} + \frac{N^{\beta+\eps}}{\tilde M}\right).
\end{align*}

Combining with the bound
\begin{align*}
	\langle\Phi_{N,M}(0),\1^{\le m}\Phi_{N,\tilde M}(0)\rangle =  1 - \langle\Phi(0),\1^{> m}\Phi(0)\rangle \ge 1 - \frac{C}{\tilde M},
	\end{align*}
	which follows from assumption $\langle\Phi_N(0),\cN\Phi_N(0)\rangle \le C$, we obtain as in \eqref{lowparticlebound}, 
	\begin{align*}
&\Re\frac{1}{\tilde M/2 -1}\sum_{m=\tilde M/2}^{\tilde M-3}\langle\Phi_{N, M}(t),\1^{\le m}\Phi_{N,\tilde M}(t)\rangle  \\
&\ge 1- C_{t,\eps} \left( \sqrt{\frac{M}{\tilde M}}N^{4\beta-1} +\frac{\sqrt{MN^{5\beta+\eps-1}}}{\tilde M} + \frac{N^{\beta+\eps}}{\tilde M}\right).
\end{align*}
Putting this together with \eqref{manyparticletrunc} we obtain
\begin{align*}
\Re\langle\Phi_{N,M}(t),\Phi_{N,\tilde M}(t)\rangle \ge 1 - C_{t,\eps} \left( N^{(4\beta -1)/2} + \frac{N^{(\beta+\eps)/2}}{\sqrt{\tilde M}} + \sqrt{\frac{M}{\tilde M}}N^{4\beta-1} +\frac{\sqrt{MN^{5\beta+\eps-1}}}{\tilde M}  \right).
\end{align*}
Therefore, \eqref{eq:nor-PhiNM-PhintM-norm-inner} gives the desired inequality. 
\end{proof}

\subsection{Step 3: Conclusion of Theorem~\ref{thm:main2}}

\begin{proof}[Proof of Theorem~\ref{thm:main2}] Let $\Phi_N(t)=U_N(t)\Psi_N(t)$ as in \eqref{eq:eq-PhiN}. For every $1\le M\le N$, let $\Phi_{N,M}(t)$ be the truncated dynamics defined as in \eqref{eq:eq-PhiNM}. 

For every $0<\beta<1/6$ and $0<\alpha<(1-6\beta)/4$, we can find a finite number $K>0$ and a decreasing sequence $\{M_k\}_{k=1}^K$ such that   
\begin{align} 
&N^{1-2\beta}\gg M_1 \ge \max \{ N^{4\beta + 2\alpha}, N^{(1+\beta)/2 +\alpha}\},\\
& M_k \ge M_{k+1} \ge \max \{ \sqrt{M_k}, M_k N^{-\beta}, N^{\beta+2\alpha}\}, \quad \forall k=1,2,...,K-1,\\
& M_K = N^{\beta+2\alpha}.
\end{align}
From Lemmas \ref{lem:truncated-vs-full-dynamics}, \ref{lem:normtrunc-to-trunc} and \ref{lem:normtrunc-to-Bog} and the above choice of $\{M_k\}_{k=1}^K$, we have the norm approximations
\begin{align*}
\|\Phi_N(t) -\Phi_{N,M_1}(t)\|^2 &\le C_{t,\eps} \left( N^{(4\beta -1)/2}+ \frac{N^{2\beta}}{\sqrt{M_1}} + \frac{N^{(1+\beta+\eps)/2}}{M_1} \right) \le C_{t,\eps} N^{-\alpha+\eps},\\
\|\Phi_{N,M_k}(t)	- \Phi_{N, M_{k+1}}(t)\|^2 &\le  C_{t,\eps} \left( N^{(4\beta -1)/2} + \sqrt{\frac{N^{\beta+\eps}}{M_{k+1}}} + \sqrt{\frac{M_k}{M_{k+1}}}N^{4\beta-1} +\frac{\sqrt{M_k N^{5\beta+\eps-1}}}{M_{k+1}}  \right),\\
&\le C_{t,\eps} N^{-\alpha+\eps}, \quad \forall  k=1,2,...,K-1, \\
\|\Phi_{N,M_K}(t)	- \Phi(t)\|^2 &\le C_{t,\eps} \left( \frac{ N^{\beta+\eps}}{M_K} + \sqrt{  M_K N^{5\beta-1+2\eps}} \right) \le C_{t,\eps} N^{-\alpha+\eps}.
\end{align*}

By the triangle inequality we conclude that 
\begin{align*}
 \|\Phi_N(t) - \Phi(t)\| &\le \|\Phi_N(t) - \Phi_{N,M_1}(t)\| + \sum_{k=1}^{K-1} \|\Phi_{N,M_k}(t) - \Phi_{N,M_{k+1}}(t)\|\nn\\
 &\qquad +\|\Phi_{N, M_K}(t) - \Phi(t)\| \le (K+1)  C_{t,\eps} N^{(-\alpha+\eps)/2}. 
 \end{align*}
Finally, since $U_N(t)$ is a (partial) unitary operator we have that
\begin{align*}
\|\Psi_N(t) - U_N^*(t)\1^{\le N}\Phi(t)\| = \|\Phi_N(t) - \1^{\le N}\Phi(t)\| \le \|\Phi_N(t) - \Phi(t)\|. 
\end{align*}
and this finishes the proof of Theorem \ref{thm:main2}.
\end{proof}

\appendix

\section{Extension to systems on torus}

Our main results in Theorem \ref{thm:main1} and \ref{thm:main2} remain valid when the configuration space $\R^3$ is replaced by the torus $\T^3=(\R/(2\pi\Z))^3$. In our proofs, the domain issue only emerges in the level of the effective theory. Therefore, the main concern is the well-posedness of the quintic Hartree equation on $\T^3$
\begin{equation}
\label{eq:Hartree-Torus}
\left\{
\begin{aligned}
i\partial_tu(t,x) &=   \Big( - \Delta +  \frac{1}{2}\iint |u(t,y)|^2 V_N(x-y,x-z)|u(t,z)|^2\d y \d z\Big) u(t,x) , \quad x\in \T^3, \, t>0 \\
u(0,x) &= u_0(x). 
\end{aligned}
\right.
\end{equation}

We will prove the following analogue of Theorem \ref{thm:Hartree} for the torus case. 
\begin{theorem}
	\label{thm:Torusmain}
	Let $u_0 \in H^4(\T^3)$ be an initial state. Then for every time $T>0$,  there exists a solution to equation~\eqref{eq:Hartree-Torus} on $[0,T]$ which satisfies
	\begin{align}
	\label{eq:H4bound1Torus}
	\|u(t,\cdot)\|_{H^4(\T^3)} \le C_t, \quad \|\partial_t u(t,\cdot)\|_{H^2(\T^3)} \le C_t.
	\end{align}
	Here the constant $C_t$ is dependent on $\|u_0\|_{H^4}$ and $t$, but independent of $N$. 
\end{theorem}

Again the result is proven by considering the Hartree equation as a perturbation of the quintic NLS. Even though the basic idea is very similar to the  $\R^3$ case, the details in the torus case is more involved because some Strichartz's estimates are no longer available. In the following we will follow the analysis for quintic NLS on torus by Ionescu and Pausader \cite{IonPau-12,IonPau-12b} and will mainly focus on the places when nontrivial modifications from the $\R^3$ case in Section \ref{sec:Hartree} are needed.  

Let us start by briefly recalling the definition of the Littlewood-Paley projections. Take $\psi:\mathbb{R}\to[0,1]$ a smooth even function with $\psi(y)=1$ if $|y|\leq 1$ and $\psi(y)=0$ if $|y|\geq 2$. We will use a decomposition into dyadic integers $$M=2^j,\quad j\in\Z.$$ The Littlewood-Paley projectors $P_{\le M}$ and $P_M$ are then defined by
\begin{equation*}
\begin{split}
&\widehat{P_{\le M}f}\left(\xi\right):=\psi(|\xi|/M)\hat{f}(\xi),\qquad\xi\in \Z^3,\\
&P_1f:=P_{\le 1}f,\qquad P_Mf:=P_{\le M}f-P_{\le M/2}f\quad\hbox{if}\quad M\ge 2.
\end{split}
\end{equation*}
The following strong functional spaces are introduced by Herr-Tataru-Tzvetkov \cite{HerrTatTzv-11}
$$\|u\|_{X^s(\R)}^2 = \sum_{\xi\in \Z^3} (1+|\xi|^2)^s \|e^{it |\xi|^2} \widehat{u}(t,\xi) \|^2_{U_t^2}, $$
$$\|u\|_{Y^s(\R)}^2 = \sum_{\xi\in \Z^3} (1+|\xi|^2)^s \|e^{it |\xi|^2} \widehat{u}(t,\xi) \|^2_{V_t^2},$$
with $U^p$ and $V^p$ are defined by Hadac-Herr-Koch \cite{HadHerrKoch-09}. For any bounded time  interval $I \subset [0,\infty)$, we denote $X^s(I)$ and $Y^s(I)$ in the usual way as restriction norms. As suggested by Ionescu-Pausder \cite{IonPau-12}, we will also use the following spacetime norm
\begin{equation}
\begin{split}
&\| u\|_{Z(I)}:=\sum_{p\in\{p_0,p_1\}}\sup_{J\subseteq I,\,|J|\le 1}\Big(\sum_MM^{5-p/2}\| P_Mu(t)\|_{L^{p}_{t,x}(J\times\T^3)}^{p}\Big)^{1/p},\\
&p_0=4+1/10,\qquad p_1=100,
\end{split}
\end{equation}
which satisfies $\| u\|_{Z(I)}\lesssim \| u\|_{X^1(I)}$. 
Finally we will also use a norm interpolating between the $Z(I)$ and $X^1(I)$ norm
\begin{align*}
\|u\|_{Z'(I)} = \|u\|^{1/2}_{Z(I)}\|u\|^{1/2}_{X^1(I)}.
\end{align*} 
\subsection{Estimate of the nonlinear term}
In order to prove Theorem~\ref{thm:Torusmain} it is important to have good control on the nonlinear term in \eqref{eq:Hartree-Torus}. For that we will use the following Lemma. This will be proved similarly to \cite[Lemma 3.2]{IonPau-12} (see also \cite[Proposition 4.1]{HerrTatTzv-11}, \cite[Lemma 3.2]{IonPau-12b} and \cite[Lemma 5.5]{CheHol-18} for related results).
\begin{lemma}	\label{lem:nonlinear-estimate}
	Let $s\ge 1$, and $u_j\in X^s(I)$, $j=1\dots 5$. Then we have on a time interval $I =(a,b)$, $|I|\le 1$ the following estimate

	\begin{align} \label{eq:nonlinear-estimate-torus}
	&\nn\left\|\int_a^t e^{i(t-s)\Delta}\iint \tilde u_1(s,y)\tilde u_2(s,y)V_N(x-y,x-z)\tilde u_3(s,z)\tilde u_4(s,z)\d y\d z \,\,\tilde u_5(s)\d s\right\|_{X^s(I)} \\ &\lesssim \|u_{\sigma(1)}\|_{X^s(I)} \prod^5_{j=2}\|u_{\sigma(j)}\|_{Z'(I)},
	\end{align}	
	for any permutation $\sigma$ on $\{1,\cdots,5\}$. Here $\tilde u_j$ can either be $u_j$ or $\overline u_j$.	
\end{lemma}
\begin{proof}
We will just proof the case where $\sigma$ is the identity since the case for all other permutations follows in the same way.
By \cite[Proposition 2.11]{HerrTatTzv-11} we know that $$\int_0^te^{i(t-s)\Delta}P_{\le M}\left(\iint \tilde u_1\tilde u_2V_N\tilde u_3\tilde u_4\d y\d z \,\,\tilde u_5\right)\d s \in X^s(I)$$ with
\begin{align*}
&\left\|\int_0^te^{i(t-s)\Delta}P_{\le M}\left(\iint \tilde u_1\tilde u_2V_N\tilde u_3\tilde u_4\d y\d z \,\,\tilde u_5\right)\d s\right\|_{X^s(I)} \\ &\le \sup \left|\int_{I\times\T^3}P_{\le M}\left(\iint \tilde u_1(t,y)\tilde u_2(t,y)V_N(x-y,x-z)\tilde u_3(t,z)\tilde u_4(t,z)\d y\d z \,\,\tilde u_5(t,x)\right) \overline{v(t,x)}\d x\d t\right|,
\end{align*}
where the supremum is taken over all $v\in Y^{-s}(I)$ with $ \|v\|_{Y^{-s}}=1$.
Define $u_0 = P_{\le M} v$. With this notation we need to prove
\begin{align}
&\label{eq:Xs-estimate}\nn\left|\int_{I\times\T^3}\iint \tilde u_1(t,y)\tilde u_2(t,y)V_N(x-y,x-z)\tilde u_3(t,z)\tilde u_4(t,z)\d y\d z \,\,\tilde u_5(t,x) \tilde u_0(t,x)\d x\d t\right|\\ &\le \|u_0\|_{Y^{-s}(I)} \|u_1\|_{X^s(I)}\prod_{j=2}^5\|u_j\|_{Z'(I)}.
\end{align}
The desired result will then follow by taking the limit $M\to\infty$. 

We now use the dyadic decompositon $u_k = \sum_{M_k}P_{M_k}u_k$. With that and the Cauchy-Schwarz inequality we see that in order to bound the left hand side of \eqref{eq:Xs-estimate} it suffices to bound
\begin{align*}
S=\sum_{\cM}\iint V_N(y,z)&\|P_{M_0}\tilde u_0P_{M_2}\tilde u_2(\cdot-y)P_{M_4}\tilde u_4(\cdot-z)\|_{L^2_{t,x}(I\times\T^3)}\times\\
&\times\|P_{M_1}\tilde u_1(\cdot-y)P_{M_3}\tilde u_3(\cdot-z)P_{M_5}\tilde u_5\|_{L^2_{t,x}(I\times\T^3)}\d y\d z,
\end{align*}
where $\cM$ is the set of tuples $(M_0,M_1,M_2,M_3,M_4,M_5)$ of dyadic numbers satisfying $M_i\ge1$ and $$M_5\le M_4 \le \cdots \le M_1,\quad\max\{M_0,M_2\}\sim M_1.$$
We first consider the subset $\cM_1 \subset \cM$ with $M_2 \le M_0 \sim M_1$:
Using \cite[Lemma 3.1]{IonPau-12} and translation invariance of the $Y^0$ and $Z'$ norms we get
\begin{align*}
S_1 \lesssim &\iint V_N(y,z)\d y\d z \sum_{\cM_1}\left(\frac{M_5}{M_1} + \frac{1}{M_3}\right)^{\delta}\left(\frac{M_4}{M_0} + \frac{1}{M_2}\right)^{\delta}\times\\
&\quad\quad\quad\times \|P_{M_1}u_1\|_{Y^0}\|P_{M_3}u_3\|_{Z'}\|P_{M_5}u_5\|_{Z'}\|P_{M_0}u_0\|_{Y^0}\|P_{M_2}u_2\|_{Z'}\|P_{M_4}u_4\|_{Z'}. \\
&\lesssim \sum_{\cM}\left(\frac{M_5}{M_1} + \frac{1}{M_3}\right)^{\delta}\left(\frac{M_4}{M_0} + \frac{1}{M_2}\right)^{\delta}\times\\
&\quad\quad\quad\times \|P_{M_1}u_1\|_{Y^0}\|P_{M_3}u_3\|_{Z'}\|P_{M_5}u_5\|_{Z'}\|P_{M_0}u_0\|_{Y^0}\|P_{M_2}u_2\|_{Z'}\|P_{M_4}u_4\|_{Z'},
\end{align*}
for some $\delta >0$.
By using the Cauchy-Schwarz inequality we can sum with respect to $M_2$,$M_3$,$M_4$, and $M_5$ which gives
\begin{align*}
S_1 \lesssim \prod_{j=2}^5\|u_j\|_{Z'}\sum_{M_0\sim M_1}\|P_{M_0}u_0\|_{Y^0}\|P_{M_1}u_1\|_{Y^0} \sim \prod_{j=2}^5\|u_j\|_{Z'} \sum_{M_0\sim M_1}M_0^{-s}\|P_{M_0}u_0\|_{Y^0}\,M_1^{s}\|P_{M_1}u_1\|_{Y^0}.
\end{align*}
Using again Cauchy-Schwarz in $M_1$ we get
\begin{align*}
S_1 \lesssim \|u_0\|_{Y^{-s}}\|u_1\|_{Y^s}\prod_{j=2}^5\|u_j\|_{Z'} \lesssim  \|u_0\|_{Y^{-s}}\|u_1\|_{X^s}\prod_{j=2}^5\|u_j\|_{Z'}.
\end{align*}
For the subset $\cM_2 \subset \cM$ with $M_0 \le M_2 \sim M_1$ we use both estimates in \cite[Lemma 3.1]{IonPau-12} and obtain
\begin{align*}
S_2 &\lesssim \sum_{\cM_2}\iint V_N(y,z)\d y\d z\left(\frac{M_5}{M_1} + \frac{1}{M_3}\right)^\delta M_0^{1/2 -5/p_0}M_2^{1/2 -5/p_0}M_4^{10/p_0 -2} \times \\
&\quad\quad\quad\quad\quad \times\|P_{M_1}u_1\|_{Y^0}\|P_{M_3}u_3\|_{Z'}\|P_{M_5}u_5\|_{Z'}\|P_{M_0}u_0\|_{Z}\|P_{M_2}u_2\|_{Z}\|P_{M_4}u_4\|_{Z}.
\end{align*}
Now using Strichartz's estimate \cite[Corollary 2.2]{IonPau-12} (and the embedding $Y^0(I)\hookrightarrow U^{p_0}_\Delta(I,L^2)\hookrightarrow U_\Delta^{p_1}(I,L^2)$, with notation $U_\Delta$ in \cite{IonPau-12}) we have 
$$\|P_{M_0}u_0\|_{Z(I)} \lesssim M_0\left(\|u_0\|_{U^{p_0}_\Delta(I,L^2)} + \|u_0\|_{U^{p_1}_\Delta(I,L^2)}\right) \lesssim M_0 \|P_{M_0}u_0\|_{Y^0(I)}.$$
Putting this in the above we obtain
\begin{align*}
S_2 &\lesssim \sum_{\cM_2} \left(\frac{M_5}{M_1} + \frac{1}{M_3}\right)^\delta M_0^{3/2 -5/p_0}M_2^{1/2 -5/p_0}M_4^{10/p_0 -2} \times \\
&\quad\quad\quad\quad\quad \times\|P_{M_1}u_1\|_{Y^0}\|P_{M_3}u_3\|_{Z'}\|P_{M_5}u_5\|_{Z'}\|P_{M_0}u_0\|_{Y^0}\|P_{M_2}u_2\|_{Z}\|P_{M_4}u_4\|_{Z} \\
&\le \sum_{\cM_2} \left(\frac{M_5}{M_1} + \frac{1}{M_3}\right)^\delta \|P_{M_1}u_1\|_{Y^0}\|P_{M_3}u_3\|_{Z'}\|P_{M_5}u_5\|_{Z'}\|P_{M_0}u_0\|_{Y^0}\|P_{M_2}u_2\|_{Z}\|P_{M_4}u_4\|_{Z}.
\end{align*}
We can now use the Cauchy-Schwarz to sum over $M_5, M_4, M_3$ to get
\begin{align*}
S_2 &\lesssim \prod_{j=3}^5\|u_j\|_{Z'}\sum_{M_1 \sim M_2 \ge M_0}\|P_{M_1}u_1\|_{Y^0}\|P_{M_0}u_0\|_{Y^0}\|P_{M_2}u_2\|_{Z} \\
&\lesssim \prod_{j=2}^5\|u_j\|_{Z'}\sum_{M_1 \ge M_0}\|P_{M_1}u_1\|_{Y^0}\|P_{M_0}u_0\|_{Y^0} \\
&\lesssim  \prod_{j=2}^5\|u_j\|_{Z'}\sum_{M_1 \ge M_0}\left(\frac{M_0}{M_1}\right)^s\|P_{M_1}u_1\|_{Y^s}\|P_{M_0}u_0\|_{Y^{-s}}.
\end{align*}
Now using Schur's Lemma we obtain
\begin{align*}
S_2 \lesssim \|u_0\|_{Y^{-s}}\|u_1\|_{Y^s}\prod_{j=2}\|u_j\|_{Z'} \lesssim \|u_0\|_{Y^{-s}}\|u_1\|_{X^s}\prod_{j=2}\|u_j\|_{Z'}.
\end{align*}
This finishes the proof.

\end{proof}
\subsection{Local well-posedness}
In this section we will prove the local theory for the quintic Hartree equation on $\T^3$. We will proceed very similarly to \cite[Proposition 3.3]{IonPau-12} (see also \cite[Propositon 3.3]{IonPau-12b}). In fact the only real difference is the treatment of the nonlinear term which is provided by Lemma~\ref{lem:nonlinear-estimate}.
\begin{lemma}[Local well-posedness]
	For every $u_0\in H^1(\T^3)$ there exists a $\eps>0$ depending on $\|u_0\|_{H^1(\T^3)}$ such that on any time interval $I \ni 0$, $|I|\le 1$ with
	\begin{align}
	\label{eq:smallnesslocalTor}
	\|e^{it\Delta}u_0\|_{Z'(I)} \le \eps
	\end{align}
	 there exists a unique solution $u\in X^1(I)$ of \eqref{eq:Hartree-Torus}. \\
	If furthermore $u\in X^1(I)$ is a solution of \eqref{eq:Hartree-Torus} on some open interval $I$ satisfying 
	\begin{align}
	\label{eq:Zsmall}\|u\|_{Z(I)}<\infty
	\end{align} 
	then $u$ can be extended to some neighborhood of $\overline{I}$ and
	\begin{align}
	\|u\|_{X^1(I)} \le C(\|u_0\|_{H^1(\T^3)},\|u\|_{Z(I)}).
	\end{align}
\end{lemma}
\begin{proof}
	We will proceed in the standard way by using a fixed point argument.
	Let $E:=\|u_0\|_{H^1(\T^3)}$ and define the set
	\begin{align*}
	E(I,a) = \left\{v\in X^1(I): \|v\|_{X^1(I)}\le 2E,\quad \|v\|_{Z'(I)}\le a\right\},
	\end{align*} 
	which is closed in $X^1(I)$. 
	We will consider the following map 
	\begin{align*}
	\Phi(u)(t) = e^{it\Delta}u_0 - \frac{i}{2}\int_0^t e^{i(t-s)\Delta} \iint |u(s,y)|^2 V_N(x-y,x-z)|u(s,z)|^2\d y \d z \,u(s) \d s.
	\end{align*}
	Using \cite[Proposition 2.10]{HerrTatTzv-11} for the linear term and Lemma~\ref{lem:nonlinear-estimate} for the nonlinear term, we see that
	\begin{align*}
	\|\Phi(u)(t)\|_{X^1(I)}  &\le \|u_0\|_{H^1(\T^3)} + C\|u\|^4_{Z'(I)}\|u\|_{X^1(I)}	 \le E + Ca^4E \\
	\|\Phi(u)(t)\|_{Z'(I)}  &\le \eps + C\|u\|^4_{Z'(I)}\|u\|_{X^1(I)}	\le \eps + Ca^4E.
	\end{align*}  
	In the second line we have used assumption \eqref{eq:smallnesslocalTor} and that the $Z'(I)$ can be bounded by the $X^1(I)$ norm for the nonlinear term.
	 Hence by choosing $a=2\eps$ and then $\eps>0$ small enough we obtain that $\Phi(E(I,a)) \subseteq E(I,a)$.
	
	Moreover, using again Lemma~\ref{lem:nonlinear-estimate} we obtain that
	\begin{align*}
	\|\Phi(u)- \Phi(v)\|_{X^1(I)}  \lesssim \left(\|u\|^4_{Z'} + \|v\|^4_{Z'}\right)\|u-v\|_{X^1(I)} \lesssim a^4\|u-v\|_{X^1(I)}.
	\end{align*}
	For $a>0$ small enough this gives that $\Phi$ is a contraction map on $E(I,a)$. Using now the contraction mapping principle we see that there exists a unique $u\in E(I,a)$ which solves \eqref{eq:Hartree-Torus}.\\
	To see the uniqueness of the solution in the whole space $X^1(I)$ let us assume that there exist $u,v \in X^1(I)$ solution of \eqref{eq:Hartree-Torus}. If we choose an open subinterval $J \subseteq I$ containing $0$ we have that for $J$ small enough $u,v\in E(J,a)$. By uniqueness in $E(J,a)$ we know that $u_{|J} = v_{|J}$. Hence, the set $\{u=v\}$ is open and closed in $I$ and therefore equal to $I$.
	
	The extension result for finite $Z(I)$ norm \eqref{eq:Zsmall} follows in the same way as in \cite[Proposition 3.3]{IonPau-12} (see also \cite[Lemma 3.4]{IonPau-12b}).
\end{proof}
\subsection{Proof of Theorem \ref{thm:Torusmain}} In this section we will proof the global existence and regularity stated in Theorem~\ref{thm:Torusmain}. We will use the global well-posedness theory of the quintic nonlinear Schrödinger equation on $\T^3$
\begin{equation}
\label{eq:quinticNLSTorus}
i\partial_t \tilde{u} = -\Delta\tilde{u} +b_0|\tilde{u}|^4\tilde{u}
\end{equation}
provided in \cite{IonPau-12}. The quintic Hartree equation \eqref{eq:Hartree-Torus} will then be considered as a perturbation of \eqref{eq:quinticNLSTorus}. Using a stability result similar to \cite[Proposition 3.4]{IonPau-12} in the case of the quintic NLS, this will give the global theory of \eqref{eq:Hartree-Torus}.
\begin{lemma}[Stability]
	\label{lem:Stability-Torus}
	Assume $I$ is an open bounded interval and $\tilde u$ in $X^1(I)$ a solution of the pertubed equation
	\begin{equation}
	\label{perteq}
	i\partial_t\tilde{u} = - \Delta\tilde{u} + \frac{1}{2}\iint |\tilde{u}(y)|^2 V_N(x-y,x-z)|\tilde{u}(z)|^2\d y \d z\,\tilde{u} + e
	\end{equation}
	for some function $e$. Moreover, assume that
	\begin{align}
	\|\tilde u\|_{Z(I)} &\le M,\\
		\|\tilde{u}\|_{L^{\infty}_t \dot H^1_x (I \times \T^3)} &\le E, \\
	\|u_0 - \tilde u(0)\|_{H^1(\T^3)}&+\left\|\int_{0}^t e^{i(t-s)\Delta}e(s,\cdot)\d s\right\|_{X^1(I)}	\le \eps
	\end{align}
	for some $M,E\ge 0$ and some small enough $\eps>0$.
	
	Then there exists a solution $u\in X^1(I)$ of \eqref{eq:Hartree-Torus}
	with
	\begin{align}
	\|u\|_{X^1(I)} \le C(M,E).
	\end{align}
\end{lemma}
Using Lemma~\ref{lem:nonlinear-estimate} for the nonlinear term, the proof of Lemma~\ref{lem:Stability-Torus} follows similarly to \cite[Propositon 3.4]{IonPau-12} (consider also \cite[Proposition 3.5]{IonPau-12b}) and is therefore omitted. \\

Now we come to the proof of the main result Theorem~\ref{thm:Torusmain}.

\begin{proof}[Proof Theorem~\ref{thm:Torusmain}]
In order to prove Theorem~\ref{thm:Torusmain} we want to apply Lemma~\ref{lem:Stability-Torus} with $\tilde u$ being the solution of the pertubed equation given by the quintic NLS on $\T^3$
\begin{align*}
i\partial_t \tilde{u} &= -\Delta\tilde{u} +b_0|\tilde{u}|^4\tilde{u} \\
&=  - \Delta\tilde{u} + \frac{1}{2}\iint |\tilde{u}(y)|^2 V_N(x-y,x-z)|\tilde{u}(z)|^2\d y \d z\,\tilde{u} + e_N,
\end{align*}
with initial state $\tilde u(0,x) = u_0$.
Here we have defined the pertubation
\begin{equation*}
e_N = b_0|\tilde{u}|^4\,\tilde{u} - \frac{1}{2}\iint |\tilde{u}(y)|^2 V_N(x-y,x-z)|\tilde{u}(z)|^2\d y \d z\,\tilde{u}.
\end{equation*}
In order to use Lemma~\ref{lem:Stability-Torus} we want to show that $$\left\|\int_0^t e^{i(t-s)\Delta}e_N(s,\cdot)\d s\right\|_{X^1(I)}$$ is arbitrarily small for $N$ large. Without loss of generality we assume $|I|\le1$ in order to apply Lemma~\ref{lem:nonlinear-estimate}. 
By the triangle inequality we have
\begin{align*}
&\left\|\int_0^t e^{i(t-s)\Delta}e_N(s,\cdot)\d s\right\|_{X^1(I)} \\ 
&\le\left\|\int_0^te^{i(t-s)\Delta} \iint V_N(y,z)\left(|\tilde u(x-y)|^2 - |\tilde u(x)|^2\right)|\tilde u(x)|^2\d y\d z\,\tilde u(s,\cdot)\d s\right\|_{X^1(I)} \\
&+\left\|\int_0^te^{i(t-s)\Delta} \iint V_N(y,z)\left(|\tilde u(x-z)|^2 - |\tilde u(x)|^2\right)|\tilde u(x-y)|^2\d y\d z\,\tilde u(s,\cdot)\d s\right\|_{X^1(I)}.
\end{align*} 
We will only consider the first term since the second term will follow the same way.
From Lemma~\ref{lem:nonlinear-estimate} we have that
\begin{align*}
&\left\|\int_0^te^{i(t-s)\Delta} \iint V_N(y,z)\left(|\tilde u(x-y)|^2 - |\tilde u(x)|^2\right)|\tilde u(x)|^2\d y\d z\,\tilde u(s,\cdot)\d s\right\|_{X^1(I)} \\
&\lesssim \|\tilde u\|^4_{X^1(I)} \sup_{|y|\le CN^{-\beta}}\|\tilde u(\cdot-y) - \tilde u\|_{X^1(I)},
\end{align*}
where we have used that $V_N(y,\cdot) = 0$ for $|y|>CN^{-\beta}$ since $V$ has compact support.
Using the global well-posedness of the quintic NLS on $\T^3$ proven in \cite{IonPau-12} we see that this expression is arbitrarily small for $N$ large. By Lemma~\ref{lem:Stability-Torus} we now obtain that \eqref{eq:Hartree-Torus} has a solution $u\in X^1(I)$. \\ \\

To conclude the regularity result \eqref{eq:H4bound1Torus} we use that $\|u\|_{X^1(I)}$ is finite and the nonlinear estimate in Lemma~\ref{lem:nonlinear-estimate}. Here it is important that the right side of \eqref{eq:nonlinear-estimate-torus} includes the weaker norm of $Z'$ which can be made arbitrarily small by localizing in time (unlike the $X^1$-norm). To be precise, since $\|u\|_{Z'(I)}$ is finite, for every $\delta>0$ the time interval $I$ can be split up into finitely many subintervals $I_0,\cdots, I_K$ such that
$$
\|u\|_{Z'(I_j)} \le \delta $$
for each $j =1,\cdots, K$. We also assume $|I_j| \le 1$ for each $j=1,\cdots, K$. Now using Duhamel's formula and Lemma~\ref{lem:nonlinear-estimate} we obtain
\begin{align*}
\|u\|_{X^4(I_0)} &\lesssim \|u_0\|_{H^4(\T^3)} + \|u\|^4_{Z'(I_0)}\|u\|_{X^4(I_0)} \\
&\le \|u_0\|_{H^4(\T^3)} + \delta^4\|u\|_{X^4(I_0)}.
\end{align*} 
If we choose $\delta>0$ small enough, this gives
\begin{align*}
\|u\|_{X^4(I_0)} \lesssim \|u_0\|_{H^4(\T^3)}.
\end{align*}
 The embedding $X^4(I_0)\hookrightarrow L^\infty(I_0,H^4(\T^3))$ now gives
$$\|u(t,\cdot)\|_{H^4(\T^3)} \lesssim \|u_0\|_{H^4(\T^3)},$$
for each $t\in I_0$.
Using this we can iterate the procedure and obtain
$\|u(t,\cdot)\|_{H^4(\T^3)} \lesssim \|u_0\|_{H^4(\T^3)}$ for all $t\in I$.
\end{proof}


\begin{thebibliography}{10}

\bibitem{AdaGolTet-07}
{R.~Adami, F.~Golse, and A.~Teta}, 
{Rigorous derivation of the cubic NLS
in dimension one}, J.~Stat. Phys. {127} (6) (2007), 1193--1220.

\bibitem{AmmNie-08}
{Z.~Ammari and F.~Nier}, 
{\em Mean field limit for bosons and infinite
  dimensional phase-space analysis}, Ann. Henri Poincar\'e {9} (2008),
  1503--1574.
  
  \bibitem{BarGolMau-00}
{C.~Bardos, F.~Golse, and N.J.~Mauser}, 
{ Weak coupling limit of the N-particle Schr\"odinger equation}, 
Methods Appl. Anal. { 7} (2000), no.~2, 275--293.  

\bibitem{BenOliSch-15}
{N.~{Benedikter}, G.~{de Oliveira}, and B.~{Schlein}}, {Quantitative
  derivation of the Gross-Pitaevskii equation}, 
  Comm. Pure App. Math. {68}
  (2015), no.~8, 1399--1482.


\bibitem{BocCenSch-15}
{C.~{Boccato}, S.~{Cenatiempo}, and B.~{Schlein}}, { Quantum 
many-body fluctuations around nonlinear {S}chr\"odinger dynamics}, 
Ann. Henri Poincar\'e {18} (2017), no.~1, 113--191. 


\bibitem{Bogoliubov-47} N. N. Bogoliubov, On the theory of superfluidity, J. Phys. (USSR), 11 (1947), 23.

\bibitem{BocBreCenSch-18} C. Boccato, C. Brennecke, S. Cenatiempo, and B. Schlein,  Bogoliubov Theory in the Gross-Pitaevskii Limit, Acta Mathematica (to appear), arXiv:1801.01389 (Preprint 2018).

\bibitem{Bourgain-98} J. Bourgain, Scattering in the energy space and below for 3d NLS, Journal d'analyse Math\'ematiques 75 (1998), 267--297.

\bibitem{Bourgain-99}  J. Bourgain, Global wellposedness of defocusing critical nonlinear Schrödinger equation in the radial case, J. Amer. Math. Soc. 12 (1999), no. 1, pp. 145--171.

\bibitem{BreSch-17} C. Brennecke and B. Schlein. Gross-Pitaevskii dynamics for Bose-Einstein condensates. Preprint (2017) arxiv:1702.05625.



\bibitem{BreNamNapSch-17} C. Brennecke, P. T. Nam, M. Napi\'orkowski, and B. Schlein, Fluctuations of $N$-particle quantum dynamics around the nonlinear Schrödinger equation, Ann. Inst. Henri Poincar\'e Anal. Non Lin\'eaire (to appear), arXiv:1710.09743 (Preprint 2017) .

\bibitem{caz-03} T. Cazenave, Semilinear Schr\"odinger Equations, Courant Lecture Notes Vol. 10, American Mathematical Society 2003.


\bibitem{CheHaiPavSei-15}
{T.~{Chen}, C.~{Hainzl}, N.~{Pavlovi\'c}, and R.~{Seiringer}}, 
{Unconditional uni\-queness for the cubic Gross-Pitaevskii hierarchy via
  quantum de Finetti}, Commun. Pure Appl. Math. {68} (2015), no.~10,
  1845--1884.

\bibitem{ChePav-11} T. Chen and N. Pavlovi\'c, The quintic NLS as the mean field limit of a Boson gas with three-body interactions, J. Funct. Anal., 260 (2011), pp. 959--997. 

\bibitem{Chen-12} X. Chen, Second Order Corrections to Mean Field Evolution for Weakly Interacting Bosons in the Case of Three-body Interactions, Archive for Rational Mechanics and Analysis 203 (2012), pp. 455--497


\bibitem{CheHol-18} X. Chen and J. Holmer, The Derivation of the Energy-critical NLS from Quantum Many-body Dynamics, Invent. Math. (to appear), arXiv:1803.08082 (Preprint 2018).  


\bibitem{ColKeelStafTao-08} J. Colliander, M. Keel, G. Staffilani, H. Takaoka, and T. Tao, Global well-posedness and scattering for the energy-critical nonlinear Schrödinger equation in $\R^3$, Ann. of Math. 167 (2), (2008), pp. 767--865.

\bibitem{Derezinski-17} J. Derezi\'nski, Bosonic quadratic Hamiltonians, J. Math. Phys. 58 (2017), 121101.

\bibitem{Dodson-12} B. Dodson, Global well-posedness and scattering for the defocusing, $L^{2}$-critical nonlinear Schrödinger equation when $d\geq3$, J. Amer. Math. Soc. 25 (2012), pp. 429--463. 



\bibitem{ErdSchYau-07}
{L.~Erd{\"{o}}s, B.~Schlein, and H.-T. Yau}, { Derivation of the cubic
  non-linear {S}chr\"odinger equation from quantum dynamics of many-body
  systems}, Invent. Math. { 167} (2007), 515--614.


\bibitem{ErdSchYau-10}
{L.~Erd{\"{o}}s, B.~Schlein, and H.-T. Yau}, { Derivation of the
  {G}ross-{P}itaevskii equation for the dynamics of {B}ose-{E}instein
  condensate}, Ann. of Math. (2) { 172} (2010), 291--370.
  
  \bibitem{FroKnoSch-09}
{J.~Fr{\"o}hlich, A.~Knowles, and S.~Schwarz}, { On the mean-field limit
  of bosons with {C}oulomb two-body interaction}, Commun. Math. Phys. { 288}
  (2009), 1023--1059.

\bibitem{GinVel-79}
{J.~Ginibre and G.~Velo}, { The classical field limit of scattering
  theory for nonrelativistic many-boson systems. {I}}, Commun. Math. Phys. { 66}
  (1979), 37--76.


\bibitem{Grillakis-00} M. Grillakis, On nonlinear Schr\"odinger equations, Comm. Partial Differential Equations 25
(2000), pp. 1827--1844.


\bibitem{GriMac-13} M. Grillakis and M. Machedon, Pair excitations and the mean field approximation of interacting Bosons I, Commun. Math. Phys. 324 (2003), pp. 601--636.

\bibitem{GriMac-15}
{M.~{Grillakis} and M.~{Machedon}}, { Pair excitations and the mean
  field approximation of interacting {B}osons, II}, Comm. PDE. { 42} (2017), no.~1, 24--67.


\bibitem{GriMacMar-10}
{M.~G. Grillakis, M.~Machedon, and D.~Margetis}, { Second-order
  corrections to mean field evolution of weakly interacting bosons. {I}},
  Commun. Math. Phys. { 294} (2010), 273--301.
  
  \bibitem{HadHerrKoch-09} M. Hadac, S. Herr, and H. Koch, Well-posedness and scattering for the KP-II equation in a critical space, Ann. Inst. H. Poincar\'e, Anal. Non Lin\'eaire 26 (2009), no. 3, 917--941.
  
   \bibitem{HerrTatTzv-11} S. Herr, D. Tataru, and N. Tzvetkov, Global well-posedness of the energy critical nonlinear Schr\"odinger equation with small initial data in $H^1(\T^3)$, Duke Math. J. (2011) Vol. 159, no. 2, p. 329--349.
   
	 
  
  \bibitem{Hepp-74}
{K.~Hepp}, { The classical limit for quantum mechanical correlation
  functions}, Comm. Math. Phys. { 35} (1974), 265--277.
  
  \bibitem{IonPau-12} A. D. Ionescu and B. Pausader, The energy-critical defocusing NLS on $\T^3$, Duke Math. J. 161 (2012), 1581--1612. 

	 
	 \bibitem{IonPau-12b} A. D. Ionescu and B. Pausader, Global well-posedness of the energy-critical defocusing NLS on $\R\times\T^3$, Commun. Math. Phys. 312 (2012), 781--831.

  
  \bibitem{KlaMac-08}
{S.~Klainerman and M.~Machedon}, { On the uniqueness of solutions to the
  gross-pitaevskii hierarchy}, Commun. Math. Phys. { 279} (2008), 169--185.

\bibitem{KnoPic-10}
{A.~Knowles and P.~Pickl}, { Mean-field dynamics: singular potentials
  and rate of convergence}, Commun. Math. Phys. { 298} (2010), 101--138.
  
  \bibitem{Kuz-17}
{E.~Kuz}, { Exact Evolution versus Mean Field with Second-order correction 
for Bosons Interacting via Short-range Two-body Potential}, Differential Integral Equations { 30} (2017), no.~7/8, 587--630.




\bibitem{LewNamSerSol-15} M.~Lewin, P.~T. Nam, S.~Serfaty, and J.~P. Solovej, Bogoliubov spectrum of interacting Bose gases, Comm. Pure Appl. Math., 68 (2015), pp.~413--471.

\bibitem{LewNamSch-15} M.~Lewin, P.~T. Nam, and B.~Schlein, Fluctuations around Hartree states in the mean-field regime, Amer. J. Math., 137 (2015), pp.~1613--1650.






\bibitem{LieLos-01} E.~H. Lieb and M.~Loss, Analysis, Graduate Studies in Mathematics, American Mathematical Society, 2001.
	
\bibitem{LinPo-15} F. Linaris and G. Ponce, Introduction to Nonlinear Dispersive Equations Second equation, Universitext,  Springer, 2015.

\bibitem{NamNap-15}  P.~T. Nam and M.~Napi\'orkowski, Bogoliubov correction to the mean-field dynamics of interacting bosons, Adv. Theor. Math. Phys. 21 (2017), 683--738.

\bibitem{NamNap-17} P.~T. Nam and M.~Napi\'orkowski, A note on the validity of Bogoliubov correction to mean-field dynamics, J. Math. Pure. Appl. 108 (2017), 662--688.

\bibitem{NamNap-17/2} P.T. Nam and M.  Napi\'orkowski, Norm approximation for many-body quantum dynamics: focusing cases in low dimensions. Preprint (2017).

\bibitem{NamNapSol-16}  P.~T. Nam, M.~Napi\'orkowski, and J. P. Solovej, Diagonalization of bosonic quadratic Hamiltonians by Bogoliubov transformations, J. Funct. Anal., 270 (11) (2016), pp.~4340--4368.

\bibitem{NamRouSei-16} P. T. Nam, N. Rougerie, and R. Seiringer, Ground states of large bosonic systems: The Gross-Pitaevskii limit revisited, Analysis \& PDE 9 (2016), 459-485. 

\bibitem{Pickl-15}
{P.~Pickl}, { Derivation of the
  time dependent {G}ross {P}itaevskii equation with external fields}, Rev.
  Math. Phys. { 27} (2015), 1550003.
  
\bibitem{Pizzo-15} A. Pizzo. Bose particles in a box I. A convergent expansion of the ground state of a three-modes
Bogoliubov Hamiltonian, Preprint (2015) arXiv:1511.07022.


\bibitem{RodSch-09}
{I.~Rodnianski and B.~Schlein}, { Quantum fluctuations and rate of
  convergence towards mean field dynamics}, Commun. Math. Phys. { 291} (2009),
  31--61.

\bibitem{Strich-77}  R. S. Strichartz, Restriction of Fourier Transform to Quadratic Surfaces and Decay of Solutions of Wave Equations, Duke Math. J., 44 (1977), 70 5?774. 

\bibitem{Spohn-80} H. Spohn, Kinetic equations from Hamiltonian dynamics: Markovian limits, Rev.
Modern Phys. 52 (1980), pp. 569--615.

\bibitem{Tao-06} T. Tao, Introduction to Nonlinear Dispersive Equations Second equation, CBMS 106,  AMS, 2006.

\end{thebibliography}
\end{document}